\documentclass[10pt]{amsart}
\usepackage{graphicx,amssymb,amsfonts,amsmath,amsthm,newlfont}
\usepackage{epsfig,url}
\usepackage{axodraw4j}
\usepackage{color}

\newcommand{\CC}{{\mathbb C}}

\newcommand{\FF}{{\mathbb F}}

\newcommand{\PP}{{\mathbb P}}
\newcommand{\QQ}{{\mathbb Q}}
\newcommand{\RR}{{\mathbb R}}

\newcommand{\ZZ}{{\mathbb Z}}
\newcommand{\ee}{\mathrm{e}}
\newcommand{\ii}{\mathrm{i}}

\newcommand{\dd}{\mathrm{d}}
\font \rus= wncyr10
\newcommand{\sha}{\,\hbox{\rus x}\,}
\newcommand{\aaa}{{\overline{a}}}
\newcommand{\bb}{{\overline{b}}}

\newcommand{\kk}{{\overline{k}}}

\renewcommand{\ll}{{\overline{\ell}}}
\newcommand{\mm}{{\overline{m}}}

\newcommand{\zz}{{\overline{z}}}

\newcommand{\sB}{\mathcal{B}}
\newcommand{\sE}{\mathcal{E}}
\newcommand{\sG}{\mathcal{G}}
\newcommand{\sH}{\mathcal{H}}

\newcommand{\sL}{\mathcal{L}}
\newcommand{\sM}{\mathcal{M}}

\newcommand{\sO}{\mathcal{O}}
\newcommand{\sV}{\mathcal{V}}
\newcommand{\sR}{\mathcal{R}}
\newcommand{\sS}{\mathcal{S}}
\newcommand{\sX}{\mathcal{X}}

\newcommand{\FLT}{\mathrm{FLT}\,}
\newcommand{\Li}{\mathrm{Li}\,}
\newcommand{\eval}{\mathrm{eval}\,}

\newcommand{\res}{\mathrm{res}\,}
\renewcommand{\Im}{\mathop\mathrm{Im}}
\renewcommand{\Re}{\mathop\mathrm{Re}}

\theoremstyle{plain}
\newtheorem{thm}{Theorem}
\newtheorem{lem}[thm]{Lemma}

\newtheorem{cor}[thm]{Corollary}
\newtheorem{prop}[thm]{Proposition}
\newtheorem{quest}[thm]{Question}
\newtheorem{prob}[thm]{Problem}
\newtheorem{remark}[thm]{Remark}
\newtheorem{defn}[thm]{Definition}

\newtheorem{ex}[thm]{Example}

\title{Generalized single-valued hyperlogarithms}
\author{Oliver Schnetz}
\address{Department Mathematik\\
Cauerstra{\ss}e 11\\
91058 Erlangen, Germany}
\email{schnetz@mi.uni-erlangen.de}

\begin{document}
\begin{abstract}
Single-valued hyperlogarithms are generalized to include primitives of differential forms $\dd z/(az\zz+bz+c\zz+d)$, $a,b,c,d\in\CC$, where $\zz$ is the complex conjugate of
the variable $z\in\CC$. The construction of these generalized single-valued hyperlogarithms (GSVHs) relies on a commutative hexagon which allows one to express primitives
as anti-primitives. The article provides a proved constructive theory of GSVHs.
\end{abstract}

\maketitle

\section{Introduction}
Single-valued hyperlogarithms on the punctured complex plane have a plethora of applications in mathematics. A classical example is the Bloch-Wigner dilogarithm
(see \cite{Zagierdilog} and the references therein),
\begin{equation}\label{BWD}
D(z)=\mathrm{Im}\,(\Li_2(z)+\log(1-z)\log|z|),
\end{equation}
where $\Li_2(z)=\sum_{k=1}^\infty z^k/k^2$ is the (multivalued) dilogarithm. There exists a natural class of single-valued functions which generalizes the Block-Wigner dilogarithm
to hyperlogarithms of higher weights. This class of single-valued hyperlogarithms is well studied and understood (see, e.g., \cite{BrSVMP,BrSVMPII,gf}).

In recent years, certain calculations in perturbative Quantum Field Theory (pQFT) led to single-valued functions on the punctured complex plane (see, e.g., \cite{SYM,gf}).
It quickly became clear that the class of single-valued hyperlogarithms is too restricted to describe many of these functions (even if they are of hyperlogarithmic nature)
\cite{Duhr,gf,numfunct}.

While single-valued hyperlogarithms are iterated integrals of differential forms with poles in $\CC$,
$$
\frac{\dd z}{z-a},\quad a\in\CC,
$$
there is demand for including differential forms with denominators which are bilinear in $z$ and $\zz$ (the complex conjugate of $z$),
\begin{equation}\label{bilinform}
\frac{\dd z}{az\zz+bz+c\zz+d},\quad a,b,c,d\in\CC.
\end{equation}
One expects that (iterated) single-valued primitives of these differential forms exist if the zero locus of the denominator is empty in $\CC$.
By integrating $\dd z/(z+1/\zz)$, e.g., we get $\log(z\zz+1)$ which is single-valued on $\CC$ (Example \ref{wt1ex}).

This setup, however, still is too specific. In practice, one often has denominators which vanish on a curve in $\CC$. A frequent case is the denominator $z-\zz$ which vanishes on the
real axis. In general, these differential forms have no single-valued primitives on the punctured complex plane. The logarithm $\log(z-\zz)$ inherits the singular real axis from the
differential form.

Nonetheless, it may happen that a numerator cancels the singularity of the denominator to render the function single-valued on the punctured complex plane.
Because $D(z)$ in (\ref{BWD}) vanishes on the real line, one examples of this type is (see Example \ref{Dzzex}, Section \ref{sect01zz}, and \cite{Duhr})
\begin{equation}\label{introexD}
\frac{D(z)}{z-\zz}.
\end{equation}

Another example of a function with lifted singularity is (Example \ref{wt2ex})
\begin{equation}\label{introexlog}
\frac{\log(z\zz)}{z-1/\zz}.
\end{equation}
Both functions have single-valued primitives on the (punctured) complex plane. Generalized single-valued hyperlogarithms (GSVHs) are the most general class of single-valued functions
which emerge from iteratively integrating differential forms (\ref{bilinform}), see Definition \ref{Gdef}.
In this article we present a theory of these functions. We prove the existence of single-valued primitives (and other fundamental properties of GSVHs) in Theorem \ref{Gthm}
using the commutativity of the hexagon in Figure \ref{fig:gsvh}. The proof is constructive (Section \ref{sectimplementation}) and implemented in the Maple package {\tt HyperlogProcedures}
\cite{Shlog} (which also contains application to pQFT).

In detail, the paper is organized as follows. In Sections \ref{sectitint} and \ref{secthyp} we set up the notation and compile some results on iterated integrals and hyperlogarithms.
In Section \ref{sectgh} we generalize to hyperlogarithms in two variables $z$ and $\zz$ which can either be considered independent or complex conjugates.
We fix a quadratically closed number field $\FF\subseteq\CC$ which is closed under complex conjugation, $\overline{\FF}=\FF$, and construct the $\CC$-vector-space $\sG\sH_\Sigma^\CC$ of
(possibly multivalued) generalized hyperlogarithms with a finite singular set $z\in\Sigma(\zz)\subset\{-(c\zz+d)/(a\zz+b),\;a,b,c,d\in\FF\}$ which emerges from iteratively integrating
differential forms of type (\ref{bilinform}). In Section \ref{sectsv} we define the $\CC$-algebra $\sS\sV_{\Sigma_0}$ of single-valued functions with point-like singularities in a finite set
$\Sigma_0\subset\CC$ (which -- at this point -- is not related to $\Sigma(\zz)$). Our notion of single-valuedness does not only demand the absence of monodromies but the (stronger) existence of
single-valued log-Laurent expansions
$$
\sum_{\ell=0}^L\sum_{m,\mm\geq M_a}c_{\ell,m,\mm}^a[\log(z-a)(\zz-\aaa)]^\ell(z-a)^m(\zz-\aaa)^\mm,\quad M_a\in\ZZ,
$$
at all $a\in\Sigma$ and at infinity (Definition \ref{svdef}; these expansions are natural in the context of pQFT \cite{gfe,gf,numfunct}).

Generalized single-valued hyperlogarithms are defined in Section \ref{sectgsvh} as the intersection of both spaces
$$
\sG_\Sigma^\CC=\sG\sH_\Sigma^\CC\cap\sS\sV_\CC,\quad\text{where}\quad\sS\sV_\CC=\bigcup_{\Sigma_0\subset\CC}\sS\sV_{\Sigma_0}
$$
is the union of $\sS\sV_{\Sigma_0}$ over all finite singular loci $\Sigma_0$.

Likewise, $\sG_\FF^\CC$ is the union over all $\sG_\Sigma^\CC$. We define $\partial_z\sG_\FF^\CC$, $\partial_\zz\sG_\FF^\CC$, $\partial_z\partial_\zz\sG_\FF^\CC$
which -- by Theorem \ref{Gthm} -- become derivatives of $\sG_\FF^\CC$. The functions (\ref{introexD}) and (\ref{introexlog}), e.g., are in $\partial_z\sG_\FF^\CC$.
These function spaces are subsets of the general space of GSVHs $\sG_\FF$ where hyperlogarithms in $\sG_\FF^\CC$ may have coefficients with bilinear denominators in $z$ and $\zz$
(which are consistent with single-valuedness).

With these notations we prove the commutativity of the hexagon in Figure \ref{fig:gsvh}, where projections $\pi_{\partial_z}$ and $\pi_{\partial_\zz}$ kill (anti-)residues in
$\partial_z\sG_\FF^\CC$ and $\partial_\zz\sG_\FF^\CC$, respectively.

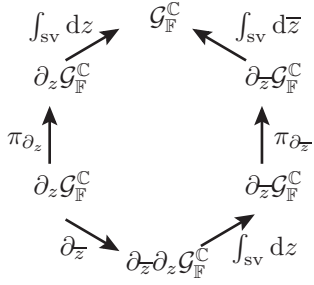
\begin{figure}
\begin{center}
\fcolorbox{white}{white}{
  \begin{picture}(160,135) (19,-11)
    \SetWidth{1.0}
    \SetColor{Black}
    \Line[arrow,arrowpos=1,arrowlength=5,arrowwidth=2,arrowinset=0.2](48,38)(48,56)
    \Line[arrow,arrowpos=1,arrowlength=5,arrowwidth=2,arrowinset=0.2](128,38)(128,56)
    \Line[arrow,arrowpos=1,arrowlength=5,arrowwidth=2,arrowinset=0.2](54,17)(71,7)
    \Line[arrow,arrowpos=1,arrowlength=5,arrowwidth=2,arrowinset=0.2](122,77)(105,87)
    \Line[arrow,arrowpos=1,arrowlength=5,arrowwidth=2,arrowinset=0.2](105,7)(122,17)
    \Line[arrow,arrowpos=1,arrowlength=5,arrowwidth=2,arrowinset=0.2](54,77)(71,87)
    \Text(42,24)[lb]{\Black{$\partial_z\sG_\FF^\CC$}}
    \Text(122,24)[lb]{\Black{$\partial_{\zz}\sG_\FF^\CC$}}
    \Text(122,65)[lb]{\Black{$\partial_{\zz}\sG_\FF^\CC$}}
    \Text(42,65)[lb]{\Black{$\partial_z\sG_\FF^\CC$}}
    \Text(32,43)[lb]{\Black{$\pi_{\partial_z}$}}
    \Text(133,43)[lb]{\Black{$\pi_{\partial_\zz}$}}
    \Text(40,84)[lb]{\Black{$\int_{\mathrm{sv}}\dd z$}}
    \Text(118,84)[lb]{\Black{$\int_{\mathrm{sv}}\dd\zz$}}
    \Text(77,-5)[lb]{\Black{$\partial_{\zz}\partial_z\sG_\FF^\CC$}}
    \Text(86,88)[lb]{\Black{$\sG_\FF^\CC$}}
    \Text(52,2)[lb]{\Black{$\partial_{\zz}$}}
    \Text(117,0)[lb]{\Black{$\int_{\mathrm{sv}}\dd z$}}
  \end{picture}
}
\end{center}
\caption{The inductive construction of GSVHs by a commutative hexagon.}
\label{fig:gsvh}
\end{figure}

We use the commutative hexagon to prove the structure theorem for GSVHs (Theorems \ref{hexthm} and \ref{Gthm}).

\begin{thm}\label{introthm}
The hexagon in Figure \ref{fig:gsvh} commutes. The space $\sG_\FF$ of GSVHs is stable under taking (anti\nobreakdash-)primitives, complex conjugation, and fractional linear transformations.
For any finite set of fractional linear transformations $\Sigma(\zz)$ and any finite set $\Sigma_0$ of singular points in $\FF$ we obtain an exact sequence
\begin{equation}\label{exseq}
0\longrightarrow\CC\longrightarrow\sG_\FF^\CC\cap\sS\sV_{\Sigma_0}\stackrel{\partial_z}{\longrightarrow}\partial_z\sG_\FF^\CC\cap\sS\sV_{\Sigma_0}\longrightarrow0.
\end{equation}
Moreover, $\sG_\Sigma^\CC=\sG\sH_\Sigma^\CC\cap\sS\sV_{\Sigma(\zz)\cap\CC}$.
\end{thm}

In the last two sections we investigate GSVHs more concretely. Section \ref{sectimplementation} explains how to explcitly use the commutative hexagon for the construction of GSVHs.
The results of this section are the basis for the Maple implementation of GSVHs in {\tt HyperlogProcedures} \cite{Shlog}.
In Section \ref{sectconstruction} we focus on GSVHs whose letters in $\Sigma(\zz)$ relate to involutions on the complex plane.
We conclude the article with the special case $\Sigma(\zz)=\{0,1,\zz\}$, where the involution is complex conjugation (Section \ref{sect01zz}).
The $\CC$-vector-space $\sG_{\{0,1,\zz\}}^\CC$ has finite dimensions at fixed weights (the number of iterated integrations). We conjecture these dimensions up to weight nine
(Table \ref{tab1}) with a proof up to weight five (Theorem \ref{thmdims}). In Theorem \ref{thm01zz} we prove that GSVHs in $\{0,1,\zz\}$ are defined over
the $\QQ$-algebra of multiple zeta values (MZVs, see (\ref{MZVdef})).
\section*{Acknowlegements}
The author is indebted to Dirk Kreimer for his support while being a visiting scientist at Humboldt University, Berlin, from 2011 to 2015 where the project started and many results were obtained.
He is also very grateful to F. Brown for his interest which lead to many inspiring and helpful discussions. The author is supported by DFG grant SCHN~1240.

\section{Iterated integrals}\label{sectitint}
Let
\begin{equation}\label{Sigmadef}
\Sigma=\{s_1,s_2,\ldots,s_N\}\subset\CC
\end{equation}
be a finite set of points in $\CC$. Let $\Sigma^\ast$ be the set of words with letters in $\Sigma$, with $e\in\Sigma^\ast$ being the empty word.
Let $|w|$ be the length of the word $w$ and $\widetilde{w}$ be $w$ in reversed order. With the shuffle product
\begin{equation}\label{sha}
ua\sha vb=(u\sha vb)a+(ua\sha v)b,\quad\text{for }u,v\in\Sigma^\ast,\;a,b\in\Sigma,
\end{equation}
and $e\sha u=u=u\sha e$ the $\ZZ$-span $\langle\Sigma^\ast\rangle_\ZZ$ becomes a commutative ring.

For any word $w=a_1\ldots a_n\in\Sigma^\ast$ and any path $\gamma:[0,1]\rightarrow\CC\backslash\Sigma$ with $a_0=\gamma(0)$ and $a_{n+1}=\gamma(1)$ we define the iterated integral \cite{Chen}
\begin{equation}\label{Idef}
I(a_0,a_1\ldots a_n,a_{n+1})=\int_{0<t_1<\ldots<t_n<1}\gamma^\ast\frac{\dd t_1}{t_1-a_1}\wedge\ldots\wedge\gamma^\ast\frac{\dd t_n}{t_n-a_n},
\end{equation}
where $\gamma^\ast\omega$ is the pullback of the differential form $\omega$ by $\gamma$.

Because (in this one-dimensional case) the integrand is closed, it is clear that the iterated integral is a homotopy invariant. It neither depends on the parametrization
of the curve $\gamma$ nor on the shape of the path $\gamma$ as long as the endpoints are fixed and no singularities are crossed.
It is convenient to generalize to the potentially singular case where the endpoints $a_0,a_{n+1}$ are in $\Sigma$.
In this case the value of the iterated integral may depend on the direction in which the endpoints are approached (see, e.g., Section 2.2 in \cite{gf} for an elementary approach to iterated
integrals). In the case of hyperlogarithms (Section \ref{secthyp}) the initial point is 0 and the endpoint is the variable $z$. One typically choses the initial point to be approached from the
positive real axis (this, e.g., gives $I(0,0,z)=\log z$). An equivalent alternative is to shuffle-regularize iterated integrals, see Section \ref{sectreg}.

In the case of single-valued functions (Section \ref{sectsv}) the regularization convention for $\gamma$ is insignificant as long as it is used consistently. In \cite{Shlog} the path $\gamma$ is
always a straight line between the endpoints (this gives $I(0,0,z)=\log|z|$).

In general, iterated integrals are very convenient to handle. An (incomplete) list of identities (all of which are easy to verify) is \cite{gf}
\begin{description}
\item[I0]
$I(a_0,a_1)=1$ by definition.
\item[I1]
$I(a_0,w,a_{n+1})$ is independent of the parametrization of $\gamma$.
\item[I2]
$I(a_0,w,a_{n+1})$ is a homotopy invariant.
\item[I3]
$I(a_0,w,a_0)=0$ for the constant path $\gamma=a_0$ and $|w|\geq1$.
\item[I4]
$I(a_0,w,a_{n+1})=(-1)^{|w|}I(a_{n+1},\widetilde{w},a_0)$ if the path $\gamma$ is reversed.
\item[I5]
For any $x\in\CC$,
$$
I(a_0,a_1\ldots a_n,a_{n+1})=\sum_{k=0}^nI(a_0,a_1\ldots a_k,x)I(x,a_{k+1}\ldots a_n,a_{n+1})
$$
by path composition $\gamma=\gamma_1\gamma_2$ with $\gamma_1(1)=x=\gamma_2(0)$.
\item[I6]
For any $x\in\CC$,
$$
I(a_0,u,x)I(a_0,v,x)=I(a_0,u\sha v,x)
$$
by shuffling the $t_i$ in the integration domain of (\ref{Idef}).
\item[I7]
For any $A\in\CC^\times$, $B\in\CC$ and $a_0\neq a_1$, $a_n\neq a_{n+1}$,
$$
I(a_0,a_1\ldots a_n,a_{n+1})=I(Aa_0+B,(Aa_1+B)\ldots(Aa_n+B),Aa_{n+1}+B)
$$
by substituting $t_i\mapsto (t_i-B)/A$ in (\ref{Idef}).
\item[I8]
For letters $a_0,\ldots,a_{n+1}$ which depend on a variable $x$,
$$
\partial_xI(a_0,a_1\ldots a_n,a_{n+1})=\sum_{k=1}^n\left(\partial_x\log\frac{a_{k+1}-a_k}{a_k-a_{k-1}}\right)I(a_0,a_1\ldots a_{k-1}a_{k+1}\ldots a_n,a_{n+1})
$$
(where $\partial_x\log0=0$) by using integration by parts in (\ref{Idef}).
\end{description}
Note that in {\bf I7} the alphabet $\Sigma$ changes to $A\Sigma+B$. There exists an analogous formula for any M\"obius transformation $t\mapsto\frac{at+b}{ct+d}$ \cite{gf}.
Singular iterated integrals are more subtle: The transformation $t\mapsto t/A$ formally maps $I(0,0,1)=0$ to $I(0,0,A)=\log(A)$.
If $A=1$, however, {\bf I7} holds for all iterated integrals. Likewise, identities {\bf I0}--{\bf I6} and {\bf I8} also hold for singular iterated integrals.

There also exists a formula for the Galois coaction on iterated integrals \cite{MMZ,Gon}. Although the existence of a coaction is central in the theory of iterated integrals,
we will not use it in this article.

\section{Hyperlogarithms}\label{secthyp}
\subsection{General theory}
In this section we mostly follow the first sections of \cite{BrSVMPII}. For completeness we give (somewhat independent) proofs for the results in this section.

Consider the ring of regular functions on $\CC\backslash\Sigma$,
\begin{equation}\label{OSigma}
\sO_\Sigma=\CC[z,((z-a)^{-1})_{a\in\Sigma}].
\end{equation}
(One may think of the generator $z$ as being related to a puncture at infinity.) For any word $w\in\Sigma^\ast$ we define the multivalued analytic hyperlogarithm $L_w(z)$ as iterated integral,
\begin{equation}\label{LwI}
L_w(z)=I(0^+,w,z).
\end{equation}
The $L_w(z)$ span the space of hyperlogarithms over $\sO_\Sigma$,
$$
\sH\sL_\Sigma=\langle L_w(z),\;w\in\Sigma^\ast\rangle_{\sO_\Sigma}.
$$
The homotopy dependence of $I(0^+,w,z)$ reflects the multivaluedness of $L_w(z)$. We get from {\bf I8}
\begin{equation}\label{pLw}
\partial_zL_{wa}(z)=\frac{1}{z-a}L_w(z)
\end{equation}
with (regualrized) limits $L_w(0)=0$. If $0\in\Sigma$, we get for a string $0^{\{n\}}$ of $n$ letters 0
\begin{equation}\label{0}
L_{0^{\{n\}}}(z)=\frac{1}{n!}\log^n(z).
\end{equation}
If the path of the iterated integral does not encircle singularities then $L_w(z)$ can be expressed as multiple sum.
With $w=b_10^{\{n_1-1\}}b_20^{\{n_2-1\}}\ldots b_r0^{\{n_r-1\}}$ (i.e.\ $b_1=a_1$, $b_2=a_{n_1+1}$, \ldots) we get
\begin{equation}\label{zeta}
L_w(z)=(-1)^r\sum_{0<k_1<k_2<\ldots<k_r}\frac{(\frac{b_2}{b_1})^{k_1}(\frac{b_3}{b_2})^{k_2}\cdots(\frac{z}{b_r})^{k_r}}{k_1^{n_1}k_2^{n_2}\cdots k_r^{n_r}}\quad
\text{for small }|z|.
\end{equation}

We refer to $|w|$ as the weight of the hyperlogarithm $L_w$. The number $r$ of non-zero letters in $w$ is the depth of $L_w$ (and of $w$).
It will follow from Theorem \ref{Brthm} that weight and depth of hyperlogarithms are well-defined.

Let $\CC\langle\Sigma\rangle$ be the $\CC$-vector-space freely generated by $\Sigma^\ast$.
Following F. Brown in \cite{BrSVMP,BrSVMPII} we use the formal algebra (note that in a computer implementation one would use this formal setup)
$$
\Sigma^\ast\otimes\sO_\Sigma:=\CC\langle\Sigma\rangle\otimes_\CC\sO_\Sigma
$$
as model of hyperlogarithms. We consider $\Sigma^\ast\otimes\sO_\Sigma$ as free $\sO_\Sigma$ module where multiplication is on the
right\footnote{In \cite{BrSVMPII} multiplication is on the left. Because of the connection to motivic Galois theory where the Galois group conveniently coacts to the right we
employ a left to right notation.}.

We consider hyperlogarithms as (analytic) realizations of the map
\begin{equation}\label{eval}
\eval:\begin{array}{rcl}
\Sigma^\ast\otimes\sO_\Sigma&\rightarrow&\sH\sL_\Sigma\\
w\otimes\phi(z)&\mapsto&\phi(z)L_w(z).
\end{array}
\end{equation}
The $\sO_\Sigma$ module $\Sigma^\ast\otimes\sO_\Sigma$ has a trivial Hopf-algebra structure with shuffle multiplication (\ref{sha}) and deconcatenation as
coproduct\footnote{The Galois coaction structure has no connection to this Hopf-algebra structure.}.
The antipode is given by $w\mapsto(-1)^{|w|}\widetilde{w}$ which, by {\bf I4}, corresponds to path reversal.

By {\bf I6} eval maps the shuffle product to the pointwise product of complex functions,
\begin{equation}\label{prod}
\eval[(u\sha v)\otimes\phi(z)]=\phi(z)L_u(z)L_v(z),\quad\text{for }u,v\in\Sigma^\ast,\;\phi\in\sO_\Sigma,
\end{equation}
which lifts eval to a homomorphism between the $\CC$-algebras $\Sigma^\ast\otimes\sO_\Sigma$ and $\sH\sL_\Sigma$.

Words $0^{\{k\}}aw$, $a\neq0$, $w\in\Sigma^\ast$, which begin in 0 have a unique representation (see Lemma 3 in \cite{motg2})
\begin{equation}\label{0aw}
0^{\{k\}}aw=\sum_{i=0}^k(-1)^i0^{\{k-i\}}\sha a[0^{\{i\}}\sha w].
\end{equation}
With (\ref{0}), (\ref{zeta}), and (\ref{prod}) this defines integral and sum representations for hyperlogarithms of words which begin with 0,
\begin{equation}\label{L0aw}
L_{0^{\{k\}}aw}=\sum_{i=0}^k(-1)^i\frac{(\log z)^{k-i}}{(k-i)!}L_{a[0^{\{i\}}\sha w]}(z),
\end{equation}
where the hyperlogarithms are extended to $\Sigma^\ast$ by linearity. One can use (\ref{0}), (\ref{zeta}), (\ref{L0aw}) or (\ref{pLw}) and $L_w(0)=0$ as alternative definitions of
hyperlogarithms.

We make $\Sigma^\ast\otimes\sO_\Sigma$ a differential algebra by defining a derivative $\partial_z$.
For the empty word $e$ we set
$$
\partial_z (e\otimes\phi)=e\otimes\partial_z\phi(z).
$$
For all other words we define
\begin{equation}
\partial_z (wa\otimes\phi)=wa\otimes\partial_z\phi(z)+w\otimes\frac{\phi(z)}{z-a}.
\end{equation}
For any $a\in\Sigma$ we define left and right derivatives---with respect to $\sha$, see (\ref{sha})---$\delta_a^{\mathrm{l}}$ and $\delta_a^{\mathrm{r}}$ on $\Sigma^\ast$ by
\begin{equation}\label{drdef}
\delta_a^{\mathrm{l}}w=\left\{\begin{array}{cc}u&\text{if }w=au,\\0&\text{otherwise,}\end{array}\right.\qquad\qquad
\delta_a^{\mathrm{r}}w=\left\{\begin{array}{cc}v&\text{if }w=va,\\0&\text{otherwise.}\end{array}\right.
\end{equation}
We get
\begin{equation}\label{partial}
\partial_z=\text{id}\otimes\partial_z+\sum_{a\in\Sigma}\delta_a^{\mathrm{r}}\otimes\frac{1}{z-a}.
\end{equation}
The above formula implies that $\partial_z$ is a derivation on $\Sigma^\ast\otimes\sO_\Sigma$. We find that eval is a homomorphism of differential algebras.

We also define a residue $\res_a^{\sO_\Sigma}$ for $a\in\CC$ acting on $\Sigma^\ast\otimes\sO_\Sigma$,
$$
\res_a^{\sO_\Sigma}(w\otimes\phi)=w\otimes\res_a\phi,
$$
where $\res_a$ on the right hand side is the residue at $a$ in $\sO_\Sigma$.

Because $\res_a\partial_z=0$ we find that $\res_a^{\sO_\Sigma}\circ\partial_z$ decreases the weights of the words in $\Sigma^\ast\otimes\sO_\Sigma$.
The residue in $\sH\sL_\Sigma$ needs a further evaluation at $z=a$,
$$
\res_a=\eval_a\circ\res_a^{\sO_\Sigma}:=\eval|_{z=a}\circ\res_a^{\sO_\Sigma}.
$$

Without restriction we assume that any $x=\sum_ww\otimes\phi_w\in\Sigma^\ast\otimes\sO_\Sigma$ is a finite sum over distinct words $w$ with non-zero canceled fractions $\phi_w$.
Any such representation is unique. The maximum weight (length) of words in $x$ induces a strict partial order on $\Sigma^\ast\otimes\sO_\Sigma$.
It is a standard technique in our proofs to show that a minimal counter-example does not exist. To use this technique efficiently it is convenient to refine the order induced by the weight.

\begin{defn}\label{deforder}
Let $x=\sum_ww\otimes\phi_w\in\Sigma^\ast\otimes\sO_\Sigma$ as above with weight $|x|=\max\{|w|\}$. Let $N_x=\#\{w,|w|=|x|\}$ be the number of words in the maximum weight part of $x$.
The strict partial order `$\prec$' on $\Sigma^\ast\otimes\sO_\Sigma$ is defined by
\begin{equation}
x\prec y\Leftrightarrow |x|<|y|\text{ or }(|x|=|y|\text{ and }N_x<N_y).
\end{equation}
\end{defn}
Note that $\partial_zf\not\succ f$ for all $f\in\Sigma^\ast\otimes\sO_\Sigma$.

\begin{lem}\label{lemker}
The kernel of $\partial_z$ in $\Sigma^\ast\otimes\sO_\Sigma$ is $e\otimes\CC$.
\end{lem}
\begin{proof}
Let $x\in\Sigma^\ast\otimes\sO_\Sigma$ be of minimum weight $n$ such that $\partial_zx=0$. Modulo lower weights $\partial_z$ only acts on the right hand side of the tensor product.
For $\partial_zx=0$ we hence need that $\partial_z$ trivializes the fractions of the maximum weight part of $x$. Hence, these fractions are constant.
If the weight of $x$ is at least 1 then the maximum weight part of $x$ is of the form $\sum_{wa} wa\otimes c_{wa}$ with $0\neq c_{wa}\in\CC$, $a\in\Sigma$, and $w\in\Sigma^\ast$
of weight $n-1$. We get $0=\res_a^{\sO_\Sigma}\partial_zx=\sum_w w\otimes c_{wa}$ modulo terms of weight $\leq n-2$. Hence all $c_{wa}=0$ and the claim follows by contradiction.
\end{proof}
We get the following structure theorem for hyperlogarithms.

\begin{thm}[F. Brown 2004 \cite{BrSVMPII}]\label{Brthm}
Let $\Sigma\subset\CC$ be a finite set.
\begin{enumerate}
\item The map \rm{eval} is an isomorphism. In particular, $\sH\sL_\Sigma$ is a free $\sO_\Sigma$ module.
\item The sequence
$$
0\longrightarrow\CC\longrightarrow\sH\sL_\Sigma\stackrel{\partial_z}{\longrightarrow}\sH\sL_\Sigma\longrightarrow0
$$
is exact. I.e.\ the kernel of $\partial_z$ in $\sH\sL_\Sigma$ is $\CC$ and every $f\in\sH\sL_\Sigma$ has a primitive $F\in\sH\sL_\Sigma$ with $\partial_zF=f$.
\item $\sH\sL_\Sigma$ is differentially simple. I.e.\ for every $0\neq f\in\sH\sL_\Sigma$ there exists a differential operator $D$ such that $Df=1$.
\end{enumerate}
\end{thm}
\begin{proof}
We know that eval is a surjective differential homomorphism. To prove (1) we show that there exists no $0\neq x\in\Sigma^\ast\otimes\sO_\Sigma$ with $\eval x=0$.
Because the order $\prec$ is discrete we may assume that $x$ is minimal with respect to $\prec$ and all $\Sigma$. Let $w\otimes\phi$ be a term in $x$ of maximum weight.
Because $\{\phi=0\}$ is a finite set in $\CC$ we may assume without restriction that $\{\phi=0\}\subseteq\Sigma$. We have $\phi^{-1}\in\sO_\Sigma$.
Because the term $w\otimes1$ in $(e\otimes\phi^{-1})x$ loses weight upon differentiation we have $\partial_z(e\otimes\phi^{-1})x\prec x$.
Hence $\partial_z(e\otimes\phi^{-1})x=0$ by minimality of $x$. From Lemma \ref{lemker} we get $(e\otimes\phi^{-1})x\in e\otimes\CC$. Hence $x=e\otimes c\phi$ for some $c\in\CC$.
From $\eval x=c\phi=0$ we get $c=0$ which is a contradiction.

By statement (1) and Lemma \ref{lemker} the kernel of $\partial_z$ is $\CC$. To prove (2) it hence suffices to show that every $f\in\sH\sL_\Sigma$ has a primitive.
We do this by induction over the weight $n\in\{-\infty,0,1,\ldots\}$ of $f$. A primitive of $0$ is 0.
For $f\neq0$ we may assume by linearity using partial fraction decomposition that $f(z)=(z-a)^kL_w(z)$ for $w\in\Sigma^\ast$, $a\in\Sigma$, and $k\in\ZZ$.
A primitive of $f(z)=(z-a)^{-1}L_w(z)$ is $L_{wa}(z)\in\sH\sL_\Sigma$. If $k\neq-1$, integration by parts yields
$$
\int (z-a)^kL_w(z)\,\dd z=\frac{(z-a)^{k+1}}{k+1}L_w(z)-\frac{1}{k+1}\int(z-a)^{k+1}\partial_zL_w(z)\,\dd z.
$$
With (\ref{pLw}) the right hand side exists in $\sH\sL_\Sigma$ by induction.

We also prove (3) by induction over the weight $n$ of $f$. Let $0\neq f\in\sH\sL_\Sigma$ and let $Q$ be the common denominator of all coefficients of terms with weight $n$ in $f$.
Then $Qf$ has only polynomial coefficients in the maximum weight piece. Let $m$ be the maximum degree of these polynomial. With $D_0=\partial_z^mQ$
we find that $D_0f$ has weight $n$ but $\partial_zD_0f$ has weight $\leq n-1$. If $D_0f=c\in\CC^\times$ then $D=D_0/c$. Otherwise we obtain from statement (2) that $\partial_zD_0f\neq0$.
By induction there exists a differential operator $D_1$ such that $D=D_1\partial_zD_0$ fulfills $Df=1$.
\end{proof}

\subsection{Regularized evaluation}\label{sectreg}
With the definition
\begin{equation}\label{0k}
L_{0^{\{n\}}}(0)=0,\quad\text{if }n\geq1,
\end{equation}
we get $L_w(0)=0$ for all $w\neq e$.

The evaluation of a hyperlogarithm $L_w(z)$ at $z=a\neq0$ is singular if $w$ ends in $a$.
In this case we can use (\ref{0aw}) from the right (swap the order of all letters) to shuffle-regularize $w$.
With the map $t\mapsto a-t$ in (\ref{Idef}) (see {\bf I7}) and path reversal {\bf I4} it is natural to define
\begin{equation}\label{ak}
L_{a^{\{n\}}}(a)=(-1)^nL_{0^{\{n\}}}(a).
\end{equation}
Identities (\ref{0k}) and (\ref{ak}) regularize the evaluation of hyperlogarithms $L_w(a)$ for any word $w$ at any $a\in\Sigma$.

Note that the prescription is a definition. Had we used the map $t\mapsto t-a$ in {\bf I7}, the right hand side of (\ref{ak}) would have been
$(-1)^nL_{0^{\{n\}}}(-a)$ (which possibly generates a new letter $-a$). In the case of single-valued functions where $\log|a|$ replaces $\log(a)$ both prescriptions are equivalent.
We skip a detailed treatment of regularization here and use (\ref{ak}) as definition.
\begin{ex}
We want to evaluate $L_{012}(z)$ at $z=2$. Using (\ref{0aw}), (\ref{L0aw}), and (\ref{ak}) we get
\begin{eqnarray*}
L_{012}(2)&=&L_2(2)L_{01}(2)-L_{201}(2)-L_{021}(2)\\
&=&-L_0(2)[L_0(2)L_1(2)-L_{10}(2)]-L_{201}(2)-L_0(2)L_{21}(2)+L_{201}(2)+L_{210}(2)\\
&=&-(\log2)^2L_1(2)+\log2[L_{10}(2)-L_{21}(2)]+L_{210}(2).
\end{eqnarray*}
\end{ex}

\subsection{Series expansions of hyperlogarithms}
Let
$$
\sH\sL_\Sigma^\CC=\eval\Sigma^\ast\otimes\CC
$$
be the $\CC$-algebra of hyperlogarithms with constant coefficients.

With (\ref{zeta}) we obtain from (\ref{L0aw}) an expansion of $L_w(z)$ for any word $w$,
\begin{equation}\label{ex0}
L_w(z)=\sum_{\ell=0}^{|w|}\sum_{m=0}^\infty c_{\ell,m}^0(\log z)^\ell z^m,\quad c_{\ell,m}^0\in\CC,
\end{equation}
which converges in a neighborhood of 0. This generalizes to all values $a\in\Sigma$.
\begin{lem}\label{explem}
Let $f\in\sH\sL_\Sigma$ have weight $n$. For any $a\in\Sigma$ there exists an $M_a\in\ZZ$ such that
\begin{equation}\label{exa}
f(z)=\sum_{\ell=0}^n\sum_{m=M_a}^\infty c_{\ell,m}^a[\log(z-a)]^\ell (z-a)^m,\quad c_{\ell,m}^a\in\CC,
\end{equation}
in a neighborhood of $z=a$. If $f\in\sH\sL_\Sigma^\CC$ then $M_a=0$.
\end{lem}
\begin{proof}
For $f\in\sH\sL_\Sigma^\CC$ the result follows from (\ref{ex0}) by {\bf I7} with $A=1,B=-a$ and path concatenation {\bf I5} at $x=0$.
To obtain the expansion for a general $f\in\sH\sL_\Sigma$ one has to Laurent expand the coefficients in $\sO_\Sigma$. In this case $M_a$ is the
maximum pole order of the coefficients at $z=a$.
\end{proof}
Note that the values of the coefficients $c_{\ell,m}^a$ depend on the sheet on which the multivalued hyperlogarithm $f$ is expanded (see next subsection).

At $z=\infty$ there exists an analogous expansion in $1/z$ (use Section \ref{sectitintpar} and (\ref{ex0}) for $f(1/z)$),
\begin{equation}\label{exinfty}
f(z)=\sum_{\ell=0}^n\sum_{m=-\infty}^{M_\infty}c_{\ell,m}^\infty(\log z)^\ell z^m,\quad c_{\ell,m}^\infty\in\CC,
\end{equation}
with $M_\infty=0$ if $f\in\sH\sL_\Sigma^\CC$. For $a\in\CC\setminus\Sigma$ the function $f\in\sH\sL_\Sigma$ is holomorphic at $z=a$.

\subsection{Monodromy}\label{sectmono}
From now on we assume that the path of integration in $L_w(z)$ is a straight line from $0^+$ to $z$ with regularizations (\ref{0k}) and (\ref{ak}) if necessary.

For $z\in\CC\backslash\Sigma$ fix a cycle basis $\{\gamma_a,a\in\Sigma\}$ of the homotopy group of $\CC\backslash\Sigma$ where a cycle $\gamma_a=\gamma_{za}\gamma_{aa}\gamma_{za}^{-1}
\subset\CC\backslash\Sigma$ goes from $z$ to a regular neighborhood of $a$ (by $\gamma_{za}$), encircles $a$ in counter-clockwise direction (by $\gamma_{aa}$), and goes back to $z$
following $\gamma_{za}$ with reversed orientation. Define
\begin{equation}
\sM_aL_w(z)=\gamma_a^\ast L_w(z),
\end{equation}
which means we analytically continue $L_w(z)$ from $z$ back to $z$ along $\gamma_a$. Note that, by construction, $\sM_a$ is an algebra homomorphism.
The operators $\sM_a$ generate a representation of the homotopy group of $\CC\backslash\Sigma$. A change of the cycle basis amounts a conjugation of the generators $\sM_a$.

Let us first study $\sM_0$. For $\gamma_{z0}$ we use the straight line from $z$ to a small positive $\epsilon$.
We assume that there are no singularities on this path (otherwise we have to move $z$ slightly). We get $\gamma_0^\ast L_0(z)=\gamma_0^\ast \log(z)=L_0(z)+2\pi\ii$.
With (\ref{0}) we find
\begin{equation}\label{mono0}
\sM_0L_{0^{\{n\}}}(z)=\sum_{j=0}^n\frac{(2\pi\ii)^j}{j!}L_{0^{\{n-j\}}}(z).
\end{equation}
With (\ref{zeta}) it is clear that $\sM_0$ is trivial on $L_w(z)$ for all words $w$ that do not begin in 0.
\begin{lem}
For any word $0^{\{k\}}w$, $w\in\Sigma^\ast$ with $\delta_0^{\mathrm{l}}w=0$, Equation (\ref{mono0}) generalizes to (Lemma 2.2 in \cite{gf})
\begin{equation}\label{mono0a}
\sM_0L_{0^{\{k\}}w}(z)=\sum_{j=0}^k\frac{(2\pi\ii)^j}{j!}L_{0^{\{k-j\}}w}(z).
\end{equation}
If we pull $\sM_0$ back to $\Sigma^\ast\otimes\sO_\Sigma$ using the isomorphism eval we get
\begin{equation}\label{mono0aw}
\sM_00^{\{k\}}w\otimes\phi=\sum_{j=0}^k0^{\{k-j\}}w\otimes\frac{(2\pi\ii)^j\phi}{j!}.
\end{equation}
\end{lem}
\begin{proof}
With (\ref{mono0}) it suffices to show (\ref{mono0aw}) for $w=au$ with $0\neq a\in\Sigma$. From (\ref{0aw}) and (\ref{mono0}) we get for the shuffle homomorphism $\sM_0$,
$$
\sM_0 0^{\{k\}}au\otimes\phi=\sM_0\sum_{i=0}^k(-1)^i0^{\{k-i\}}\sha a[0^{\{i\}}\sha w]\otimes\phi
=\sum_{i=0}^k\sum_{j=0}^{k-i}(-1)^i0^{\{k-i-j\}}\sha a[0^{\{i\}}\sha w]\otimes\frac{(2\pi\ii)^j\phi}{j!}.
$$
We swap the sums and use (\ref{0aw}) again to evaluate the sum over $i$ yielding (\ref{mono0aw}).
\end{proof}
It is convenient to write $\sM_0$ as an exponential,
\begin{equation}\label{monom0w}
\sM_0=\exp(2\pi\ii m_0),\quad\text{with } m_0=\delta_0^{\mathrm{l}}\otimes1.
\end{equation}
The {\em infinitesimal monodromy} $m_0$ is a nilpotent derivation on $\Sigma^\ast\otimes\sO_\Sigma$.
From (\ref{monom0w}) it is clear that $\sM_0$ is unipotent.

At $z=0$ the right hand side of (\ref{mono0a}) vanishes if $w\neq e$. We get for the regularized limit
\begin{equation}\label{mono00}
\eval_0\sM_0L_w(z)=\left\{\begin{array}{cl}\frac{(2\pi\ii)^n}{n!}&\text{if }w=0^{\{n\}},\\0&\text{otherwise.}\end{array}\right.
\end{equation}

For general $a$ we choose the path $\gamma_a=\gamma_{\epsilon z}^{-1}\gamma_{\epsilon a}\gamma_{aa}\gamma_{\epsilon a}^{-1}\gamma_{\epsilon z}$ to have five sections:
From $z$ straight to $\epsilon$, from $\epsilon$ straight to $a$, a tiny circle around $a$ in counter-clockwise direction, from $a$ straight back to $\epsilon$,
and from $\epsilon$ straight back to $z$. We assume that $\epsilon>0$ is small and that there are no singularities on this path.

In the limit $\epsilon\to0^+$ the first part of $\gamma_a$ cancels the path $\gamma_{0z}$ from $0^+$ to $z$ of $L_w(z)$.
By path concatenation we get $\gamma_{0z}\gamma_a\to\gamma_{0a}\gamma_{aa}\gamma_{0a}^{-1}\gamma_{0z}$.
Translating (\ref{mono00}) to $z=a$ the local monodromy operator $\gamma_{aa}$ maps sequences $a^{\{k\}}$ to $(2\pi\ii)^k/k!$ and is zero otherwise. With {\bf I4} and {\bf I5} we obtain
\begin{equation}\label{monoMaw}
\sM_a (w\otimes1)=\sum_{k\geq0}\sum_{w=ua^{\{k\}}vx}\frac{(-1)^{|v|}(2\pi\ii)^k}{k!}L_u(a)L_{\widetilde{v}}(a)\,x\otimes1
\end{equation}
for words $u,v,x\in\Sigma^\ast$. The infinitesimal monodromy $m_a$ is the $2\pi\ii$-coefficient,
\begin{equation}\label{monomaw}
\sM_a=\exp(2\pi\ii m_a),\quad\text{with}\quad m_a (w\otimes1)=\sum_{w=uavx}(-1)^{|v|}L_u(a)L_{\widetilde{v}}(a)\,x\otimes1.
\end{equation}
If $w$ contains a sequence of at least two letters $a$ then (\ref{monoMaw}) and (\ref{monomaw}) have singular terms which are defined by regularitzation, see Section \ref{sectreg}.
One can use the reversed version of (\ref{0aw}) to derive a fully regular formula for the monodromy of a regular word (i.e.\ a word that does not begin in 0). So, the monodromy
of regular words does not depend on the regularization prescription.
  
To see that $\sM_a=\exp(2\pi\ii m_a)$ one has to use that $\sum_{w=vu}(-1)^{|v|}L_{\widetilde{v}}(a)L_u(a)=\delta_{w,e}$ by path-reversal {\bf I4}, path composition {\bf I5},
and the triviality of the constant path {\bf I3} ($\delta_{a,b}=1$ if $a=b$ and zero otherwise).
The sum in $m_a^k$ collapses from $w=u_1av_1u_2av_2\ldots u_kav_kx$ to $w=u_1a^{\{k\}}v_kx$.

Likewise, we see that the evaluation at $z=a$ cuts off a letter $a$ from the right,
\begin{equation}\label{evalmonoma}
\eval_am_a=\eval_a\delta_a^{\mathrm{r}}\otimes1.
\end{equation}
With regularization (\ref{0k}) it is also clear that (\ref{monomaw}) reduces to (\ref{monom0w}) in the case $a=0$.

We see that $\partial_z$ and $m_a$ commute,
\begin{equation}\label{pzma}
[\partial_z,m_a]=0.
\end{equation}
With eval we lift (\ref{monomaw}) to a formula in $\sH\sL_\Sigma$,
\begin{equation}\label{monomaLw}
m_aL_w(z)=\sum_{w=uavx}(-1)^{|v|}L_u(a)L_{\widetilde{v}}(a)L_x(z).
\end{equation}
Modulo lower weights we have $m_a=\delta_a^{\mathrm{l}}\otimes1$, i.e.\
\begin{equation}\label{mamod}
m_aL_{bw}(z)=\delta_{a,b}L_w(z)+\text{terms of lower weight.}
\end{equation}
For later use we rewrite (\ref{evalmonoma}) using the residue operator on $\sH\sL_\Sigma$,
\begin{equation}\label{evalm}
\eval_am_a=\res_a\partial_z.
\end{equation}

\subsection{A commutative hexagon}\label{comhex1}
We have defined two derivations on $\sH\sL_\Sigma$, namely $\partial_z$ and $m_a$.
Because hyperlogarithms are analytic the derivative $\partial_\zz$ with respect to the complex conjugate of $z$ is zero. Later, when we define GSVHs,
we sacrifice analyticity for single-valuedness. This means that $m_a$ becomes trivial whereas $\partial_\zz$ becomes non-trivial.
The main tool for the construction of GSVHs will be a commutative hexagon which has a (simpler) analogue in the analytic case with $\partial_\zz$ replaced by $m_a$. We present this
analogue as introduction to GSVHs. An application of the commutative hexagon for hyperlogarithms is integration in the motivic $f$-alphabet \cite{fhlog}.

Before we prove commutativity of the hexagon we need to define a generic infinitesimal monodromy (covering all points $a\in\Sigma$). To this end we extend
$\sH\sL_\Sigma$ to $\Sigma^\ast\otimes\sH\sL_\Sigma$ where the first factor keeps track of the puncture $a$.
By $f\mapsto e\otimes f$ we embed $\sH\sL_\Sigma$ into $\Sigma^\ast\otimes\sH\sL_\Sigma$.
(A model for $\Sigma^\ast\otimes\sH\sL_\Sigma$ is $\Sigma^\ast\otimes\Sigma^\ast\otimes\sO_\Sigma$ which we do not use here.) We define
\begin{equation}
m:\begin{array}{rcl}
\Sigma^\ast\otimes\sH\sL_\Sigma&\longrightarrow&\Sigma^\ast\otimes\sH\sL_\Sigma\\
w\otimes f(z)&\mapsto&\sum_{a\in\Sigma}wa\otimes m_af(z).
\end{array}
\end{equation}
The derivative $\partial_z$ only acts on the right factor of $\Sigma^\ast\otimes\sH\sL_\Sigma$.

\begin{prop}\label{propseqm}
Let $\Sigma^\ast_>=\Sigma^\ast\backslash\{e\}$ be the set of non-empty words. The sequence
$$
0\longrightarrow\Sigma^\ast\otimes\sO_\Sigma\longrightarrow\Sigma^\ast\otimes\sH\sL_\Sigma\stackrel{m}{\longrightarrow}\Sigma^\ast_>\otimes\sH\sL_\Sigma\longrightarrow0
$$
is exact. I.e.\ the kernel of $m$ in $\Sigma^\ast\otimes\sH\sL_\Sigma$ is $\Sigma^\ast\otimes\sO_\Sigma$ and for every $x\in\Sigma^\ast_>\otimes\sH\sL_\Sigma$ there exists an
$X\in\Sigma^\ast\otimes\sH\sL_\Sigma$ such that $mX=x$.
\end{prop}
\begin{proof}
We first show that the kernel of $m$ is $\Sigma^\ast\otimes\sO_\Sigma$. It is clear that $m\,\Sigma^\ast\otimes\sO_\Sigma=0$.
Assume that there exists an $x\in\Sigma^\ast\otimes\sH\sL_\Sigma$ with $mx=0$ and weight $n\geq1$ (in the right factor). We uniquely write
$x=\sum_{a\in\Sigma,u,v\in\Sigma^\ast}u\otimes\phi_{uav}(z)L_{av}(z)$ with $0\neq\phi_{uav}\in\sO_\Sigma$.
By (\ref{mamod}) we have $0=mx=\sum_{a,u,v}ua\otimes\phi_{uav}(z)L_v(z)$ modulo weights $\leq n-2$. By Theorem \ref{Brthm} (1) we get
$\sum_{a,u}ua\otimes\phi_{uav}(z)=0$ for all $v$ of weight $n-1$. Hence $\phi_{uav}(z)=0$ which is a contradiction.

We prove the existence of $X$ by induction over the weight $n$ of $x\in\Sigma^\ast_>\otimes\sH\sL_\Sigma$. For $x=0$ we set $X=0$.
For $n\geq0$ we use linearity to assume that $x=ua\otimes\phi(z)L_v(z)$ with $a\in\Sigma,u,v\in\Sigma^\ast,\phi\in\sO_\Sigma,L_v\in\sH\sL_\Sigma$.
We set $X=u\otimes\phi(z)L_{av}(z)-Y$ and get from (\ref{mamod}) that $mX=x+y-mY$ with $y$ of weight $\leq n-1$. By induction there exists an $Y$ such that $mY=y$ which completes the proof.
\end{proof}

Consider the $\CC$-algebra of hyperlogarithms with constant coefficients $\sH\sL_\Sigma^\CC$.
The derivative $\partial_z$ maps $\sH\sL_\Sigma^\CC$ onto $\partial_z\sH\sL_\Sigma^\CC$, the $\CC$-vector-space of hyperlogarithms with simple pole coefficients.
Likewise $m$ maps $\sH\sL_\Sigma^\CC$ to $m\sH\sL_\Sigma^\CC=\Sigma\otimes\sH\sL_\Sigma^\CC$, the $\CC$-vector-space of hyperlogarithms with single letters as left factors.
We also define $\partial_zm\sH\sL_\Sigma^\CC=m\partial_z\sH\sL_\Sigma^\CC=\Sigma\otimes\partial_z\sH\sL_\Sigma^\CC$, see (\ref{pzma}).

On $\sH\sL_\Sigma$ we have the residue $\res_a:\sH\sL_\Sigma\rightarrow\CC$. We define the projection $\pi_{\partial_z}$ onto the residue-free part of $\sH\sL_\Sigma$,
\begin{equation}\label{pip}
\pi_{\partial_z}:\begin{array}{rcl}
\sH\sL_\Sigma&\rightarrow&\sH\sL_\Sigma\\
f(z)&\mapsto&f(z)-\sum_{a\in\Sigma}\frac{\res_af}{z-a}
\end{array}
\end{equation}
with trivial extension to $\Sigma^\ast\otimes\sH\sL_\Sigma$. The projection $\pi_m$ is the identity on $e\otimes\sH\sL_\Sigma$. On $\Sigma^\ast_>\otimes\sH\sL_\Sigma$ we define
\begin{equation}
\pi_m:\begin{array}{rcl}
\Sigma^\ast_>\otimes\sH\sL_\Sigma&\rightarrow&\Sigma^\ast_>\otimes\sH\sL_\Sigma\\
wa\otimes f(z)&\mapsto&wa\otimes (f(z)-f(a)).
\end{array}
\end{equation}

We have the obvious identities
\begin{equation}\label{mpi}
m\pi_{\partial_z}=m\quad\text{and}\quad\partial_z\pi_m=\partial_z.
\end{equation}

Finally, we define on $\partial_z\sH\sL_\Sigma^\CC$ and on $m\sH\sL_\Sigma^\CC$ an integral and an inverse of $m$ (respectively) which maps into the subspace of $\sH\sL_\Sigma^\CC$
which evaluates to 0 at $z=0$ (possibly after regularization, Section \ref{sectreg}),
\begin{eqnarray}\label{int0def}
\int_0\dd z:\partial_z\sH\sL_\Sigma^\CC\rightarrow\sH\sL_\Sigma,&&f\mapsto F(z)-F(0)\quad\text{if}\quad \partial_zF=f,\nonumber\\
m^{-1}:m\sH\sL_\Sigma\rightarrow\sH\sL_\Sigma,&& x\mapsto X(z)-X(0)\quad\text{if}\quad mX=x.
\end{eqnarray}
The operators $\int_0\dd z$ and $m^{-1}$ exist and are well-defined by Theorem \ref{Brthm} (2) and by Proposition \ref{propseqm}.
The integral $\int_0\dd z$ acts on the right factor of $\Sigma^\ast\otimes\partial_z\sH\sL_\Sigma^\CC$.

\begin{figure}
\begin{center}
\fcolorbox{white}{white}{
  \begin{picture}(160,135) (19,-11)
    \SetWidth{1.0}
    \SetColor{Black}
    \Line[arrow,arrowpos=1,arrowlength=5,arrowwidth=2,arrowinset=0.2](48,38)(48,56)
    \Line[arrow,arrowpos=1,arrowlength=5,arrowwidth=2,arrowinset=0.2](128,38)(128,56)
    \Line[arrow,arrowpos=1,arrowlength=5,arrowwidth=2,arrowinset=0.2](54,17)(71,7)
    \Line[arrow,arrowpos=1,arrowlength=5,arrowwidth=2,arrowinset=0.2](122,77)(105,87)
    \Line[arrow,arrowpos=1,arrowlength=5,arrowwidth=2,arrowinset=0.2](105,7)(122,17)
    \Line[arrow,arrowpos=1,arrowlength=5,arrowwidth=2,arrowinset=0.2](54,77)(71,87)
    \Text(42,24)[lb]{\Black{$\partial_z\sH\sL_\Sigma^\CC$}}
    \Text(122,24)[lb]{\Black{$m\sH\sL_\Sigma^\CC$}}
    \Text(122,65)[lb]{\Black{$m\sH\sL_\Sigma^\CC$}}
    \Text(42,65)[lb]{\Black{$\partial_z\sH\sL_\Sigma^\CC$}}
    \Text(32,43)[lb]{\Black{$\pi_{\partial_z}$}}
    \Text(133,43)[lb]{\Black{$\pi_m$}}
    \Text(40,84)[lb]{\Black{$\int_0\dd z$}}
    \Text(118,84)[lb]{\Black{$m^{-1}$}}
    \Text(69,-5)[lb]{\Black{$m\partial_z\sH\sL_\Sigma^\CC$}}
    \Text(78,88)[lb]{\Black{$\sH\sL_\Sigma^\CC$}}
    \Text(52,2)[lb]{\Black{$m$}}
    \Text(117,0)[lb]{\Black{$\int_0\dd z$}}
  \end{picture}
}
\end{center}
\caption{The inductive construction of hyperlogarithms with a commutative hexagon.}
\label{fig:hlog}
\end{figure}
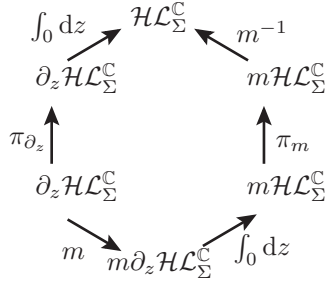

\begin{thm}\label{thmhexm}
The diagram in Figure \ref{fig:hlog} commutes.
\end{thm}
\begin{proof}
Let $f\in\partial_z\sH\sL_\Sigma^\CC$ and $f_1=\int_0\pi_{\partial_z} f\dd z$, $f_2=m^{-1}\pi_m\int_0 mf\dd z$. We need to show that $f_1=f_2$.

By (\ref{pzma}) and (\ref{mpi}) we have $\partial_zm(f_1-f_2)=m\pi_{\partial_z}f-\partial_z\pi_m\int_0mf\dd z=mf-mf=0$.
By Theorem \ref{Brthm} (2) we get $m(f_1-f_2)\in\Sigma\otimes\CC$. From (\ref{mamod}) and Proposition \ref{propseqm} we get
$$
f_1(z)-f_2(z)=\sum_{a\in\Sigma}c_aL_a(z)+d\quad\text{with }c_a,d\in\CC.
$$
Because $f_1(0)=f_2(0)=0$ we have $d=0$. We choose $b\in\Sigma$ and apply $\eval_bm_b$ to the above equation,
$$
[m_b(f_1-f_2)](b)=c_b.
$$
By (\ref{evalm}) we have $(m_bf_1)(b)=\res_b\partial_zf_1=0$ because $\partial_zf_1=\pi_{\partial_z}f$ is in the image of $\pi_{\partial_z}$.
Likewise $(m_bf_2)(b)=0$ because $mf_2$ is in the image of $\pi_m$. We find $c_b=0$ for all $b\in\Sigma$ and hence $f_1=f_2$.
\end{proof}

The following examples illustrate the role of $\pi_{\partial_z}$ and $\pi_m$.
\begin{ex}\label{1zex}
Let $f(z)=z^{-1}\in\partial_z\sH\sL_{\{0\}}^\CC$. Then $\int_0f(z)\dd z=L_0(z)=\log(z)$ whereas $mf(z)=0$. The hexagon commutes because $\pi_{\partial_z} z^{-1}=0$.
\end{ex}
\begin{ex}\label{Labex}
Let $f(z)=L_a(z)/(z-b)\in\partial_z\sH\sL_{\{a,b\}}^\CC$ for $a\neq b\in\CC$. Then $\pi_{\partial_z} f(z)=[L_a(z)-L_a(b)]/(z-b)$ and
$$
\int_0\pi_{\partial_z}f(z)\dd z=L_{ab}(z)-L_a(b)L_b(z)\in\sH\sL_{\{a,b\}}^\CC.
$$
On the other hand we have $mf(z)=a\otimes (z-b)^{-1}$ and $\int_0mf(z)\dd z=a\otimes L_b(z)$. With (\ref{monomaLw}), we get
$$
mL_{ab}(z)=a\otimes L_b(z)-a\otimes L_b(a)+b\otimes L_a(b).
$$
Therefore
$$
m^{-1}a\otimes L_b(z)=L_{ab}(z)+L_b(a)L_a(z)-L_a(b)L_b(z).
$$

Direct application of $m^{-1}$ to $a\otimes L_b(z)$ does not reproduce $L_{ab}(z)-L_a(b)L_b(z)$.
With $\pi_m a\otimes L_b(z)=a\otimes L_b(z)-a\otimes L_b(a)$ and $m^{-1}a\otimes L_b(a)=L_b(a)L_a(z)$ the spurious term $L_b(a)L_a(z)$ cancels.
\end{ex}

\section{Generalized hyperlogarithms}\label{sectgh}
Let $\FF=\overline{\FF}\subseteq\CC$ be a quadratically closed number field, i.e.\
\begin{equation}\label{Fdef}
x\in\FF\subseteq\CC\Rightarrow\overline{x},\pm\sqrt{x}\in\FF.
\end{equation}
The smallest field $\FF$ are the constructible numbers, the largest field is $\FF=\CC$.

\begin{remark}
For applications in massless pQFT the fields $\QQ(\ii)$ and $\QQ(\ee^{\pi\ii/3})$ possibly suffice \cite{BSmod,motg2,Sc2}.
We expect that pQFT uses only a small subspace of the GSVHs defined in Section \ref{sectgsvh}.
\end{remark}

To describe single-valued hyperlogarithms on the punctured complex plane we need to give up analyticity.
Because the differentials $\partial_z$ and $\partial_\zz$ commute, in most situations one has the option to either consider $z$ and $\zz$ as complex conjugates
or as independent variables. Depending on the context we shift from one picture to the other. We use double arguments $z,\zz$ in functions if we consider $z$ and $\zz$ as independent variables.

\begin{remark}
With minimal modifications it is possible to develop a very similar theory of GSVHs where $z$ and $\zz$ are replaced by any linear combination of $z$ and $\zz$
(or more general functions of $z,\zz$). One may, e.g., use real and imaginary parts of $z$. For the (rather special) Gegenbauer method in \cite{gfe,gf}
one needs to use the modulus and the argument of $z$. In many situations, however, the variables $z$ and $\zz$ are most efficient.
In pQFT, e.g., one has to invert the Laplacian $\partial_z\partial_\zz$ which factorizes in $z$ and $\zz$.
\end{remark}

In contrast to the concept of single-valuedness (Section \ref{sectsv}), the definition of generalized hyperlogarithms is of technical nature.
In pQFT there exists no fundamental property that requires generalized hyperlogarithms. They are merely the most general functions that are well understood and can be handled efficiently.
For small graphs most graphical functions can be expressed in terms of generalized hyperlogarithms \cite{gfe,gf}.

\subsection{Fractional linear transformations}
Let
\begin{equation}\label{FLTdef}
\FLT_\FF=\Big\{\beta(z)=\frac{az+b}{cz+d},\quad a,b,c,d\in\FF\Big\}
\end{equation}
be the set of fractional linear transformations (FLTs) with coefficients in $\FF$.
The FLT $\beta$ is invertible if $ad-bc\neq0$. Then $\beta$ is an automorphism of $\CC\cup\{\infty\}$ (a M\"obius transformation). Because
$$
\partial_z\frac{az+b}{cz+d}=\frac{ad-bc}{(cz+d)^2},
$$
$ad-bc=0$ is equivalent to constant $\beta$. For any $\beta(z)=(az+b)/(cz+d)$ with $ad-bc\neq0$ the inverse of $\beta$ is (we do {\em not} use $\beta^{-1}$ for the reciprocal of $\beta$)
$$
\beta^{-1}(z)=\frac{dz-b}{-cz+a}.
$$

\begin{lem}\label{partialFLT}
For any variables $z,\zz$ and $\beta(z)=(az+b)/(cz+d)\in\FLT_\FF$ with $ad-bc\neq0$ we have
$$
\partial_z\log(\zz-\beta(z))=\frac{1}{z-\beta^{-1}(\zz)}-\frac{1}{z+d/c},
$$
where the second term on the right hand side is absent if $c=0$.
\end{lem}
\begin{proof}
Because $ad-bc\neq0$ either $a\neq0$ or $c\neq0$. We get
$$
\zz-\beta(z)=\frac{(c\zz-a)(z-\beta^{-1}(\zz))}{cz+d}.
$$
Taking $\partial_z\log$ gives the result.
\end{proof}
In general, the product of two FLTs is not a sum of FLTs (it may have squares). Here, we need the converse property that the difference of two FLTs is the product of two FLTs.

\begin{lem}\label{betalem1}
Let $\beta_1,\beta_2\in\FLT_\FF$.
\begin{enumerate}
\item There exist $\beta_3,\beta_4\in\FLT_\FF$ such that
$$
\beta_1(z)-\beta_2(z)=\beta_3(z)\cdot\beta_4(z)
$$
\item We have
$$
\det\left(\begin{array}{cc}\partial_z\log(z-\beta_1(\zz))&\partial_\zz\log(z-\beta_1(\zz))\\
\partial_z\log(z-\beta_2(\zz))&\partial_\zz\log(z-\beta_2(\zz))\end{array}\right)=
[\partial_\zz\log(\beta_1(\zz)-\beta_2(\zz))]\left[\partial_z\log\frac{z-\beta_1(\zz)}{z-\beta_2(\zz)}\right].
$$
\end{enumerate}
\end{lem}
\begin{proof}
For $a_1,b_1,c_1,d_1,a_2,b_2,c_2,d_2\in\FF$ we have
$$
\frac{a_1z+b_1}{c_1z+d_1}-\frac{a_2z+b_2}{c_2z+d_2}=\frac{(a_1c_2-a_2c_1)z^2 + (a_1d_2-a_2d_1+b_1c_2-b_2c_1)z+b_1d_2-b_2d_1}{(c_1z+d_1)(c_2z+d_2)}.
$$
Because $\FF$ is quadratically closed, the numerator factorizes in $\FF$ and the product may be expressed as $\beta_3(z)\beta_4(z)$
for a suitable choice of $\beta_3(z)$ and $\beta_4(z)$. Statement (2) is an explicit calculation.
\end{proof}
Statement (1) is the sole reason for our restriction to quadratically closed fields.

\subsection{Iterated integrals with parameters}\label{sectitintpar}
Let
$$
\sH\sL_\FF=\bigcup_{\Sigma\subset\FF}\sH\sL_\Sigma.
$$
Likewise $\sH\sL_\FF^\CC$ is the $\CC$-algebra of hyperlogarithms in $\sH\sL_\FF$ with constant coefficients in $\CC$.

Any iterated integral with letters in $\CC(z)$ is in $\sH\sL_\CC^\CC$. A bootstrap algorithm for the conversion to constant letters is given by differentiating the iterated integral with
respect to $z$ using {\bf I8}. This lowers the weight of the iterated integral. By induction the derivative is in $\sH\sL_\CC^\CC$. A priori, the primitive is in $\sH\sL_\CC$.
Using {\bf I8} it is not hard to see that the primitive is in $\sH\sL_\CC^\CC$. This algorithm was suggested by F. Brown \cite{BrH1,BrH2} and later implemented by E. Panzer \cite{Panzer:HyperInt}.

To emphasize the dependence on the variable $z$ (or $\zz$) we will give $\sH\sL_\FF$ and $\FLT_\FF$ an argument.
So, $\sH\sL_\Sigma(z)$ denotes the $\CC$-algebra of hyperlogarithms in $z$ and $\FLT_\FF(z)$ is the set of FLTs over $\FF$ in $z$.

\begin{lem}\label{limitlem}
We consider $z$ and $\zz$ as independent variables. Let $f(z,\zz)=L_{\beta_1(z)\ldots\beta_n(z)}(\zz)$ be a hyperlogarithm in $\zz$ with letters in $\FLT_\FF(z)$.
For any $a\in\FF$ we get $\lim_{z\to a}f(z,\zz)\in\sH\sL_\FF^\CC(\zz)$, where the explicit result may depend on the direction in which $z$ approaches $a$.
\end{lem}
\begin{proof}
We use induction over $n$. For $n=0$ we have $f=1$ and the claim is trivial. If $n\geq1$ then $f(z,0)=0$ and
$$
\partial_\zz f(z,\zz)=\frac{1}{\zz-\beta_n(z)}L_{\beta_1(z)\ldots\beta_{n-1}(z)}(\zz).
$$
By Brown's algorithm, the hyperlogarithm on the right hand side can be converted into a hyperlogarithm in $z$ (with $\zz$-dependent letters).
By Lemma \ref{explem} it admits a log expansion at $z=a$. By induction the constant term is in $\sH\sL_\FF^\CC(\zz)$. The factor $(\zz-\beta_n(z))^{-1}$ has an expansion at $z=a$
with constant term $(\zz-\beta_n(a))^{-1}$ if $\beta_n(a)\neq\infty$. In this case $\beta_n(a)\in\FF$ and integration with respect to $\zz$ gives $f\in\sH\sL_\FF^\CC(\zz)$ plus terms that
vanish in the limit $z\to a$. If $\beta_n(a)=\infty$, the constant term vanishes and $\lim_{z\to a}f(z,\zz)=0$.
\end{proof}

\begin{ex}\label{exLza}
Consider $f(z,\zz)=L_{z-a}(\zz)$ for $a\in\FF$. We get the regularized limit (see Section \ref{sectreg})
$$
\lim_{z\to a}f(z,\zz)=\lim_{z\to a}\log(1-\zz/(z-a))=\lim_{z\to a}\log(-\zz/(z-a))=\log(\zz)+c=L_0(\zz)+c
$$
for some constant $c\in\CC$ whose value depends on the direction in which $z$ approaches $a$.
\end{ex}

\begin{remark}
Note that due to possible singularities, taking limits in hyperlogarithms cannot be performed by substitution (in general). As an example consider
\begin{align*}
&\lim_{z\to0}L_{z0}(\zz)=\lim_{z\to0}L_{10}(\zz/z)=-\lim_{z\to0}\Li_2(\zz/z)=\lim_{z\to0}\frac{[\log(z/\zz)\pm\ii\pi]^2}{2}-\zeta(2)=\frac{[\log(\zz)\mp\ii\pi]^2}{2}-\zeta(2)\\
&\neq\;L_{00}(\zz)=\frac{(\log\zz)^2}{2},
\end{align*}
where the sign ambiguity is related to monodromy. The subtle step in Brown's algorithm is to accurately take the limits $z\to0$ that provide the integration constants \cite{Panzer:HyperInt}.
\end{remark}

\begin{prop}\label{ccprop}
Let $f(z,\zz)=L_{\beta_1(z)\ldots\beta_n(z)}(\zz)$ be a hyperlogarithm in $\zz$ with letters in $\FLT_\FF(z)$. Then
$f$ is a linear combination of hyperlogarithms in $z$ with letters in $\FLT_\FF(\zz)$ and coefficients in $\sH\sL_\FF^\CC(\zz)$.
\end{prop}
\begin{proof}
We consider $z$ and $\zz$ as independent variables and use induction over $n$ (following Brown's algorithm). For $n=0$ the claim is trivial.

By differentiation with respect to $z$ we get from {\bf I8} a sum of hyperlogarithms in $\zz$ of weight $n-1$ with factors $\pm\partial_z\log(\beta_{k+1}-\beta_k)$
for $k=0,\ldots,n$ (where $\beta_0=0$ and $\beta_{n+1}=\zz$). By induction the hyperlogarithms are linear combinations of hyperlogarithms in $z$ with letters in $\FLT_\FF(\zz)$
and coefficients in $\sH\sL_\FF^\CC(\zz)$.

With Lemma \ref{betalem1} (1) we express $\beta_{k+1}-\beta_k$ for $k<n$ as product of two FLTs over $\FF$.
Upon differentiation the log becomes a sum of simple poles in $z$ at values in $\FF$.

The term $k=n$ is only non-zero if $\beta_n$ is not constant. In this case we get from Lemma \ref{partialFLT} that the coefficient is a sum of simple poles with
values in $\FLT_\FF(\zz)$.

By (\ref{pLw}) both cases have primitives which are linear combinations of hyperlogarithms in $z$ with letters in $\FLT_\FF(\zz)$.
These hyperlogarithms vanish in the limit $z\to0$. To determine $f$ we need to add the integration constant (in $z$) which is $\lim_{z\to0}f(z,\zz)$.
By Lemma \ref{limitlem} this limit is in $\sH\sL_\FF^\CC(\zz)$
\end{proof}

\begin{ex}
Consider $f(z,\zz)=L_{z-a}(\zz)$, $a\in\FF$, from Example \ref{exLza}. We have
$$
f(z,\zz)=\log\Big(1-\frac{\zz}{z-a}\Big)=\log\Big[\Big(1-\frac{z}{\zz+a}\Big)\Big(1+\frac{\zz}{a}\Big)\Big(1-\frac{z}{a}\Big)^{-1}\Big]=L_{\zz+a}(z)+L_{-a}(\zz)-L_a(z).
$$
Differentiation with respect to $z$ gives $\partial_zf(z,\zz)=(z-\zz-a)^{-1}-(z-a)^{-1}$. Upon integration we get the first and the third term. The second term is $\lim_{z\to0}f(z,\zz)$.
\end{ex}
To obtain an explicit result one can use the path convention at the beginning of Section \ref{sectmono}. Otherwise the result depends on the sheets where the hyperlogarithms are evaluated.
This will become insignificant when we pass to single-valued functions.

\subsection{Definition of generalized hyperlogarithms}\label{sectdefGH}
Let $\zz$ be the complex conjugate of $z$. Let
$$
\Sigma(\zz)=\{\beta_1(\zz),\beta_2(\zz),\ldots,\beta_N(\zz)\}\subset\FLT_\FF(\zz)
$$
be a finite set of FLTs over $\FF$ in the variable $\zz$ with subset of constants $\Sigma(\zz)\cap\FF$.
We define the ring of regular functions on $\CC\backslash\Sigma(\zz):=\CC\backslash\{z=\beta_i(\zz),i=1,\ldots,N\}$ as the ring generated
by $\sO_{\Sigma(\zz)}(z)$ and $\sO_{\overline{\Sigma(\zz)\cap\FF}}(\zz)$, see (\ref{OSigma}),
\begin{equation}\label{OSigmadef}
\sO_\Sigma(z,\zz)=\CC[z,(z-\beta(\zz))^{-1}_{\beta\in\Sigma(\zz)},\zz,(\zz-\bb)^{-1}_{b\in\Sigma(\zz)\cap\FF}].
\end{equation}
We also define
\begin{equation}\label{OFdef}
\sO_\FF(z)=\bigcup_{\Sigma\subset\FF}\sO_\Sigma(z)\quad\text{and}\quad\sO_\FF(z,\zz)=\bigcup_{\Sigma(\zz)\subset\FLT_\FF(\zz)}\sO_\Sigma(z,\zz).
\end{equation}

\begin{ex}\label{bilinex}
Reciprocals of bilinear forms in $z$ and $\zz$ are in $\sO_\FF(z,\zz)$. For $a,b,c,d\in\FF$, $a\neq0$,
$$
\frac{1}{az\zz+bz+c\zz+d}=\frac{1}{a}\cdot\frac{1}{\zz+b/a}\cdot\frac{1}{z+(c\zz+d)/(a\zz+b)}\in\sO_{\{-\overline{b}/\aaa,-(c\zz+d)/(a\zz+b)\}}(z,\zz).
$$
\end{ex}

\begin{remark}\label{sOrk}
By partial fraction decomposition first in $z$ and then (the coefficients) in $\zz$ every $f\in\sO_\Sigma(z,\zz)$ has a unique representation as
$$
f(z,\zz)=\sum_{\beta(\zz)\in\Sigma(\zz),b\in\FF,m,\mm}c^{\beta,b}_{m,\mm}(z-\beta(\zz))^m(\zz-\overline{b})^\mm,
$$
where the sum is finite with $m,\mm\in\ZZ$ and $m,\mm<0$ if $\beta\neq0,b\neq0$, respectively. Because partial fraction decomposition may generate (spurious) singularities in $\zz$
we cannot assume that $b\in\Sigma(\zz)\cap\FF$. By Lemma \ref{betalem1} (1), however, we have $b\in\FF$.
\end{remark}

\begin{defn}\label{GHdef}
Let ($\Sigma^\ast$ is the sets of words in $\Sigma$)
\begin{equation}\label{GHeq}
\sG\sH_\Sigma=\langle L_v(\zz)L_{w(\zz)}(z),v\in\overline{\Sigma(\zz)\cap\FF}^\ast,w(\zz)\in\Sigma(\zz)^\ast\rangle_{\sO_\Sigma(z,\zz)}
\end{equation}
be the space of generalized hyperlogarithms on $\CC\backslash\Sigma(\zz)$. We use the notation $\sG\sH^\CC_\Sigma$ for the set of generalized hyperlogarithms with
coefficients in $\CC$. The sets $\sG\sH_\FF$ and $\sG\sH^\CC_\FF$ are the unions of $\sG\sH_\Sigma$ and $\sG\sH^\CC_\Sigma$ over all finite $\Sigma(\zz)\subset\FLT_\FF(\zz)$ (respectively).

By Theorem \ref{GHthm} (1) there exists a unique representation of any $f\in\sG\sH_\Sigma$ as a linear combination of hyperlogarithms
$L_v(\zz)L_{w(\zz)}(z)$ with non-zero canceled fractions as coefficients. The weight of $f$ is the maximum total weight $|v|+|w(\zz)|$ in this representation.
\end{defn}
We suppress the $\zz$ dependence of $\Sigma(\zz)$ in subscripts. In $\sO_\Sigma(z,\zz)$ we denote both arguments $z$ and $\zz$ to avoid confusion with $\sO_\Sigma(z)$.

\subsection{Properties of generalized hyperlogarithms}
In analogy to Theorem \ref{Brthm} and Proposition \ref{propseqm} we obtain the following theorem.
\begin{thm}\label{GHthm}
Let $\Sigma(\zz)\subset\FLT_\FF$ be finite.
\begin{enumerate}
\item $\sG\sH_\Sigma$ is a free $\sO_\Sigma(z,\zz)$-module which is closed under multiplication. In particular, $\sG\sH_\Sigma^\CC$ is a $\CC$-algebra.
\item The sequence
$$
0\longrightarrow\sH\sL_\FF(\zz)\longrightarrow\sG\sH_\FF\stackrel{\partial_z}{\longrightarrow}\sG\sH_\FF\longrightarrow0
$$
is exact. I.e.\ the kernel of $\partial_z$ in $\sG\sH_\FF$ is $\sH\sL_\FF(\zz)$ and every $f\in\sG\sH_\FF$ has a primitive $F\in\sG\sH_\FF$
with $\partial_zF=f$.
\item $\sG\sH_\Sigma$ is differentially simple. I.e.\ for every $0\neq f\in\sG\sH_\Sigma$ there exists a differential operator $D$ such that $Df=1$.
\item $\sG\sH_\FF$ and $\sG\sH_\FF^\CC$ are stable under complex conjugation,
\begin{equation}\label{cc}
\overline{\sG\sH_\FF}=\sG\sH_\FF,\quad\overline{\sG\sH_\FF^\CC}=\sG\sH_\FF^\CC.
\end{equation}
\item $\sG\sH_\FF$ and $\sG\sH_\FF^\CC$ are stable under transformations in $\FLT_\FF$: For any $f\in\sG\sH_\FF$, $g\in\sG\sH_\FF^\CC$, and $\beta\in\FLT_\FF$ we have
\begin{equation}\label{flt}
f(\beta(z))\in\sG\sH_\FF\quad\text{and}\quad g(\beta(z))\in\sG\sH_\FF^\CC.
\end{equation}
\end{enumerate}
\end{thm}
\begin{proof}
We consider $z$ and $\zz$ as independent complex variables. For (1) we apply Theorem \ref{Brthm} (1) to the hyperlogarithms in $z$.
The coefficients are hyperlogarithms in $\zz$ on which we apply Theorem \ref{Brthm} (1) again. This proves that $\sG\sH_\Sigma$ is free.
With (\ref{prod}) it is clear that $\sG\sH_\Sigma$ is closed under multiplication.

We fix $\zz$. By Theorem \ref{Brthm} (2) the kernel of $\partial_z$ on $\sG\sH_\FF$ are functions which do not depend on $z$. By statement (1) the space of these functions
is $\sH\sL_\FF(\zz)$. To prove the existence of primitives in $\sG\sH_\FF$ we use the explicit inductive construction in the proof of Theorem \ref{Brthm} (2).
To do so we need partial fraction decomposition with respect to $z$ in $\sO_\Sigma(z,\zz)$ which can be achieved by repeatedly using the formula
$$
\frac{1}{(z-\beta_i(\zz))(z-\beta_j(\zz))}=\frac{1}{(\beta_i(\zz)-\beta_j(\zz))(z-\beta_i(\zz))}-\frac{1}{(\beta_i(\zz)-\beta_j(\zz))(z-\beta_j(\zz))}.
$$
By Lemma \ref{betalem1} (1) it is clear that both terms on the right hand side are in $\sO_{\Sigma'}(z,\zz)$ for some $\Sigma'(\zz)\supseteq\Sigma(\zz)$
which may contain new (constant) singularities due to zeros of $\beta_i(\zz)-\beta_j(\zz)$.

For (3) we follow the steps in the proof of Theorem \ref{Brthm} (3).
Now, multiplication with the common denominator $Q$ ensures that the maximum weight part of $Qf$ is a polynomial in $z$ and $\zz$ with maximum degree $(m,\mm)$ in $(z,\zz)$.
We define $D_0=\partial_z^m\partial_\zz^\mm Q$ and find that $D_0f$ has weight $n$ but $\partial_zD_0f$ and $\partial_\zz D_0f$ have weight $\leq n-1$.
If $D_0f=c\in\CC^\times$ then $D=D_0/c$. Otherwise $\partial_zD_0f\neq0$ or $\partial_\zz D_0f\neq0$. If $\partial_zD_0f\neq0$ (or $\partial_\zz D_0f\neq0$)
then (by induction) there exists a differential operator $D_1$ such that $D=D_1\partial_zD_0$ (or $D=D_1\partial_\zz D_0$) fulfills $Df=1$.

For (4) we observe that (see Example \ref{bilinex}) the complex conjugate of $\sO_\Sigma(z,\zz)$ can be written as $\sO_{\Sigma'}(z,\zz)$ for some new set $\Sigma'(\zz)$.
So, it suffices to prove the second identity in (\ref{cc}). The complex conjugate of $L_v(\zz)L_{w(\zz)}(z)$ is $L_{\overline{v}}(z)L_{\overline{w}(z)}(\zz)$.
The first factor is in $\sG\sH_\FF^\CC$. By Proposition \ref{ccprop} we can write the second factor in terms of hyperlogarithms in $z$ with letters in $\FLT_\FF(\zz)$ and
coefficients in $\sH\sL_\FF^\CC(\zz)$. Hence, also the second factor is in $\sG\sH_\FF^\CC$. By statement (1) the product is in $\sG\sH_\FF^\CC$.

To show (5) we first observe that $\sO_\FF(z,\zz)$ is stable under transformations in $\FLT_\FF$. A generator $1/(z-\beta_i(\zz))$ in $\sO_\FF(z,\zz)$
is transformed to $1/(\beta(z)-\gamma(\zz))$ with $\gamma=\beta_i\circ\overline{\beta}\in\FLT_\FF$.
The result has a denominator which is bilinear in $z$ and $\zz$. By Example \ref{bilinex}, it is in $\sO_\FF(z,\zz)$.

By linearity it suffices to show that $g(\beta(z))\in\sG\sH_\FF^\CC$ for $g(z)=L_v(\zz)L_{w(\zz)}(z)$, $v\in\overline{\Sigma(\zz)\cap\FF}^\ast$, $w\in\Sigma^\ast$, see (\ref{GHeq}).
We use Proposition \ref{ccprop} for $\beta_1\ldots\beta_n\mapsto w\circ\overline{\beta}\in\FLT_\FF^\ast$ and $(z,\zz)\mapsto(\zz,\beta(z))$ to obtain that
$L_{w(\overline{\beta(z)})}(\beta(z))$ is a linear combination of hyperlogarithms in $\zz$ with letters in $\FLT_\FF(\beta(z))\subseteq\FLT_\FF(z)$ and coefficients in
$\sH\sL_\FF^\CC(\beta(z))$. By Brown's algorithm (see Section \ref{sectitintpar}) the coefficiens are in $\sH\sL_\FF^\CC(z)$. Likewise $L_v(\overline{\beta(z)})\in\sH\sL_\FF^\CC(\zz)$.
The result shuffles to a function in $\overline{\sG\sH_\FF^\CC}=\sG\sH_\FF^\CC$ (see statement (4)).
\end{proof}

\begin{remark}
With statement (4) in the above theorem we see that the complex conjugate of statement (2) holds for $\partial_\zz$. In particular, $\sG\sH_\FF$ is closed under taking anti-primitives.
\end{remark}

\begin{ex}\label{HLex}
Every product of hyperlogarithms in $z$ and in $\zz$ with letters in $\FF$ is in $\sG\sH_\FF$,
$$
\sH\sL_\FF(z)\sH\sL_\FF(\zz)\subset\sG\sH_\FF.
$$
\end{ex}

\begin{ex}
Extending Example \ref{bilinex} to generalized hyperlogarithms, we observe that every iterated integral in forms with denominators which are bilinear in $z$ and $\zz$ over $\FF$
is in $\sG\sH_\FF$. This is clear from using Brown's algorithm as in the proof of Proposition \ref{ccprop}.
\end{ex}

In general, it is not easy to see if a given integral is in $\sG\sH_\FF$.

\subsection{Singular decomposition}
Any function $f\in\sG\sH_\FF$ has a log-Laurent expansion at $z=0$ whose coefficients have log-Laurent expansions at $\zz=0$, see (\ref{ex0}). In general, with every order in
$z$ the pole order in $\zz$ increases. It is also possible to first expand $f$ in $\zz$ and the coefficients in $z$. In this case one may get increasingly negative orders in $z$.
For GSVHs we are interested in functions where both expansions are equal. This is equivalent to the existence of a global limit for pole orders in $z$ and $\zz$.

\begin{defn}\label{reg0def}
A function $f$ admits a (simultaneous) log-Laurent expansion at $z=\zz=0$ if
\begin{equation}\label{ex00}
f(z,\zz)=\sum_{\ell,\ll=0}^L\sum_{m,\mm=M}^\infty c_{\ell,\ll,m,\mm}(\log z)^\ell(\log\zz)^\ll z^m\zz^\mm,\quad 0\leq L,M\in\ZZ,\;c_{\ell,\ll,m,\mm}\in\CC,
\end{equation}
in a neighborhood of $z=\zz=0$. The $\CC$-algebra of functions with log-Laurent expansions at $0$ is $\sR^0$. If $L=M=0$ we say that $f$ is $\CC$-analytic at 0, see Definition \ref{svdef}.
\end{defn}
One may consider $\sR^0$ as space of regular functions at 0. 
Note that the definition of $\sR^0$ does not depend on the sheets on which the logarithms are evaluated. By (\ref{ex0}) we get $\sH\sL(z)\sH\sL(\zz)\subset\sR^0$.

\begin{ex}
The function $\log(z-\zz)=L_\zz(z)+L_0(\zz)\pm\ii\pi$ is not in $\sR^0$.
\end{ex}
In this section we derive a unique decomposition of a function $f\in\sG\sH_\Sigma$ into a regular and a singular part such that $f$ is regular at 0 if and only if its singular part vanishes.
We need the decomposition to prove Theorem \ref{techthm} which, in turn, is essential to prove the main theorem for GSVHs, Theorem \ref{Gthm}.

\begin{lem}\label{singlem1}
Let $w(\zz)$ be a word with letters in $\FLT_\FF(\zz)$ and $\delta_0^{\mathrm{l}}w(\zz)=0$. Then (the empty sum is zero)
\begin{equation}\label{singpropeq}
L_{0^{\{L\}}w(\zz)}(z)=\sum_{\ell=0}^L\Big(g_\ell(z,\zz)\log^\ell(z)+\sum_{\genfrac{}{}{0pt}{}{w=u\beta v}{\beta\neq0=\beta(0)}}f_\ell^{\beta v}(\zz)L_{0^{\{\ell\}}\beta(\zz)v(\zz)}(z)\Big),
\end{equation}
where $u,v\in\FLT_\FF^\ast$ and $\beta\in\FLT_\FF$. The functions $g_\ell(z,\zz)\in\sG\sH_\FF^\CC$ are $\CC$-analytic at $0$ with $g_\ell(0,\zz)=0$.
The functions $f_\ell^{\beta v}(\zz)\in\sH\sL_\FF^\CC(\zz)$ are anti-holomorphic at $\zz=0$.
\end{lem}
\begin{proof}
The proof is by induction over the weight $n$ of $w(\zz)$. The case $n=0$ is trivial by (\ref{0}).

Assume that $n\geq1$ and $w(\zz)=x(\zz)\beta_n(\zz)$ with $x\in\FLT_\FF^\ast$ and $\beta_n\in\FLT_\FF$. By induction we obtain
\begin{equation}\label{singproppfeq}
\partial_zL_{0^{\{L\}}x(\zz)\beta_n(\zz)}(z)=\frac{1}{z-\beta_n(\zz)}\Big(\sum_{\ell=0}^Lg_\ell(z,\zz)\log^\ell(z)+
\sum_{\genfrac{}{}{0pt}{}{x=u\beta v}{\beta\neq0=\beta(0)}}f_\ell^{\beta v}(\zz)L_{0^{\{\ell\}}\beta(\zz)v(\zz)}(z)\Big).
\end{equation}
If $\beta_n(\zz)=0$ we use integration by parts $\ell+1$ times in the first term of the sum over $\ell$, iteratively integrating $g_\ell(z,\zz)/z$
(which is $\CC$-analytic at $0$ because $g_\ell(0,\zz)=0$) and differentiating $\log^\ell(z)$. With (\ref{int0def}) we obtain a set of new functions
$\tilde g_\ell\in\sG\sH_\FF^\CC$ with $\tilde g_\ell(0,\zz)=0$.
Integration of the second term merely adds a letter 0 to $v(\zz)$. The integration constant is zero because both sides of (\ref{singpropeq}) vanish at $z=0$
(and any value of $\zz$ in the neighborhood of 0). This gives the result for $\beta_n(\zz)=0$.

If $\beta_n(0)\neq0$ then $1/(z-\beta_n(\zz))$ is $\CC$-analytic at $0$ and the result follows like in the case $\beta_n(\zz)=0$.

We are left with the case $\beta_n(\zz)\neq0=\beta_n(0)$ which implies $\beta_n(\zz)=a\zz/(b\zz+1)$ for some constants $a,b\in\FF$, $a\neq0$. In (\ref{singproppfeq}) we write
\begin{equation}\label{singproppfeq2}
g_\ell(z,\zz)=[g_\ell(z,\zz)-g_\ell(\beta_n(\zz),\zz)]+g_\ell(\beta_n(\zz),\zz).
\end{equation}
Because $g_\ell(z,\zz)$ is $\CC$-analytic, $g_\ell(\beta_n(\zz),\zz)$ is anti-holomorphic at $\zz=0$. By $g_\ell\in\sG\sH_\FF^\CC$ we have
$$
g_\ell(\beta_n(\zz),\zz)=\sum_{v,w(\zz)}c_{v,w(\zz)}L_v(\zz)L_{w(\zz)}(\beta_n(\zz)),\quad c_{v,w(\zz)}\in\CC.
$$
With Brown's algorithm and Lemma \ref{betalem1} (1) we obtain $L_{w(\zz)}(\beta_n(\zz))\in\sH\sL_\FF^\CC(\zz)$.
Therefore $g_\ell(\beta_n(\zz),\zz)\in\sH\sL_\FF^\CC(\zz)$. This implies $g_\ell(z,\zz)-g_\ell(\beta_n(\zz),\zz)\in\sG\sH_\FF^\CC$.

By the theory of complex functions (or by explicit calculation) we see that
$$
g_\ell(z,\zz)-g_\ell(\beta_n(\zz),\zz)=[z-\beta_n(\zz)]h_\ell(z,\zz)
$$
with functions $h_\ell$ which are $\CC$-analytic at $0$. We follow the argument of the case $\beta_n(\zz)=0$ to see that integrating the first term of (\ref{singproppfeq2})
gives rise to a set of new functions $\tilde g_\ell\in\sG\sH_\FF^\CC$ with $\tilde g_\ell(0,\zz)=0$.
Integrating the second term of (\ref{singproppfeq2}) gives $\ell!g_\ell(\beta_n(\zz),\zz)L_{0^{\{\ell\}}\beta_n(\zz)}(z)$ which is the term $v=e$ in the second sum of (\ref{singpropeq}).

The terms $v\neq e$ are trivially obtained from integrating the second sum in (\ref{singproppfeq}).
\end{proof}

Let $\sR^0(\zz)$ be the space of anti-analytic functions in $\sR^0$. Complementary to Definition \ref{reg0def} we define singular functions at 0.

\begin{defn}\label{sing0def}
For any finite set $\Sigma(\zz)\subset\FLT_\FF(\zz)$ we define the $\CC$-algebra $\sS^0_\Sigma$ of singular functions at 0 as the $\sR^0(\zz)$-span of hyperlogarithms
$L_{0^{\{\ell\}}\beta(\zz)v(\zz)}(z)$ with $\beta(\zz)\neq0=\beta(0)$. So, every $f\in\sS^0_\Sigma$ has a unique representation as
\begin{equation}\label{sing0eq}
f(z,\zz)=\sum_{\ell=0}^L\sum_{\genfrac{}{}{0pt}{}{\beta v}{\beta\neq0=\beta(0)}}f_\ell^{\beta v}(\zz)L_{0^{\{\ell\}}\beta(\zz)v(\zz)}(z),
\quad\text{with }L\geq0,\;\beta\in\Sigma,\;v\in\Sigma^\ast,\;f_\ell^{\beta v}(\zz)\in\sR^0(\zz).
\end{equation}
\end{defn}

We consider the $\sO_\Sigma(z,\zz)$ modules $\sO_\Sigma(z,\zz)\sR^0$ and $\sO_\Sigma(z,\zz)\sS^0_\Sigma$.
Note that a function $f\in\sO_\Sigma(z,\zz)\sR^0$ does not in general have an expansion (\ref{ex00}). After the multiplication with a common denominator $Q$, however, we get
$Qf\in\CC[z,\zz]\sR^0=\sR^0$.

The module $\sO_\Sigma(z,\zz)\sR^0$ is stable under differentiation with respect to $z$ (because $\partial_z\sO_\Sigma(z,\zz)\subset\sO_\Sigma(z,\zz)$ and $\partial_z\sR^0\subset\sR^0$)
whereas $\sO_\Sigma(z,\zz)\sS^0_\Sigma$ is not, $\partial_z\sS^0_\Sigma\ni\partial_zL_{0^{\{\ell\}}\beta(\zz)}(z)=L_{0^{\{\ell\}}}(z)/(z-\beta(\zz))\notin\sO_\Sigma(z,\zz)\sS^0_\Sigma$.

\begin{prop}\label{singprop3}
For any finite set $\Sigma(\zz)\subset\FLT_\FF(\zz)$ we have
\begin{equation}\label{singreg}
\sO_\FF(z,\zz)\sS^0_\Sigma\cap\sO_\FF(z,\zz)\sR^0\cap\sG\sH_\FF=\{0\}.
\end{equation}
\end{prop}
\begin{proof}
Consider $f(z,\zz)=\sum_wf_w(z,\zz)L_{w(\zz)}(z)\in\sX:=\sO_\FF(z,\zz)\sS^0_\Sigma\cap\sO_\FF(z,\zz)\sR^0\cap\sG\sH_\FF$, where $w(\zz)=0^{\{\ell\}}\beta(\zz)v(\zz)$
with $\beta(\zz)\neq0=\beta(0)$ (by $f\in\sO_\FF(z,\zz)\sS^0_\Sigma$) and $f_w\in\sO_\FF(z,\zz)\sH\sL_\FF^\CC(\zz)$ (by $f\in\sG\sH_\FF$).
We assume that $0\neq f\in\sX$ is a minimal counter-example with respect to `$\prec$' in Definition \ref{deforder} for the weight $|w(\zz)|$.

We consider $\partial_zf$ and $\partial_\zz f$ which are in $\sO_\FF(z,\zz)\sR^0\cap\sG\sH_\FF$. The differential operators act on $f_w(z,\zz)$ and on $L_{w(\zz)}(z)$.
In the latter case we get from {\bf I8} a sum of hyperlogarithms with one letter of $w(\zz)$ removed times factors in $\sO_\FF(z,\zz)$.
By Lemma \ref{singlem1} the hyperlogarithms split into terms in $\sR^0\cap\sG\sH_\FF^\CC$ and terms in $\sS^0_\Sigma\cap\sG\sH_\FF^\CC$.
We subtract all terms of the first type from $\partial_zf$ (or $\partial_\zz f$) to obtain an expression in $\sR^0\cap\sG\sH_\FF^\CC$ which also is in $\sS^0_\Sigma$.
The expression is in $\sX$ and it has the coefficients $\partial_zf_w$ (or $\partial_\zz f_w$) in the highest weight part.

For any polynomial $0\neq P\in\CC[z,\zz]$ we get $0\neq Pf\in\sX$.
From a given function in $\sX$ we can hence construct new functions in $\sX$ where the coefficients of highest weight are multiplied with polynomials and repeatedly differentiated
with $\partial_z$ and $\partial_\zz$. From (the proof of) Theorem \ref{GHthm} (3) we construct a differential operator $D$ with $Df_v(z,\zz)=1$ for a fixed
$v(\zz)$ of maximum weight. We may hence assume $f_v=1$ without restriction.

By induction, the differential operators $\partial_z$ and $\partial_\zz$ map $f$ to polynomials in $\log(z)$ with coefficients in $\sO_\FF(z,\zz)\sH\sL_\FF^\CC(\zz)$,
see (\ref{singpropeq}). Integrating $\partial_zf$ provides hyperlogarithms with words $0^{\{\ell\}}\beta(\zz)$. Because $f\in\sO_\FF(z,\zz)\sS^0_\Sigma$ we get
$$
f(z,\zz)=\sum_{\ell=0}^L\sum_{\genfrac{}{}{0pt}{}{\beta(\zz)\in\Sigma(\zz)}{\beta\neq0=\beta(0)}}f_{0^{\{\ell\}}\beta}(z,\zz)L_{0^{\{\ell\}}\beta(\zz)}(z).
$$
Because $\partial_zf$ and $\partial_\zz f$ are polynomials in $\log(z)$ we get $f_{0^{\{L\}}\beta}(z,\zz)=:f_\beta\in\CC$. We assume that $L$ is minimal, i.e.\ at least one
$f_\beta\neq0$.

We multiply $f\in\sO_\Sigma(z,\zz)\sR^0$ with the common denominator $Q(z,\zz)=\sum_{k+\kk\leq d}c_{k,\kk}z^k\zz^\kk$, $c_{k,\kk}\in\CC$, so that $Qf\in\sR^0$.
Using (\ref{L0aw}) we calculate the coefficient of $(\log z)^L$ in the log-Laurent expansion of $Qf$,
$$
c_L(z,\zz)=\frac{Q(z,\zz)}{L!}\sum_{\genfrac{}{}{0pt}{}{\beta(\zz)\in\Sigma(\zz)}{\beta\neq0=\beta(0)}}f_\beta L_{\beta(\zz)}(z).
$$
The FLTs $\beta(\zz)$ are of the form $a\zz/(b\zz+1)$ with $a,b\in\FF$, $a\neq0$. We expand $Q(z,\zz)$ and $L_{\beta(\zz)}(z)=\log(1-z(b\zz+1)/a\zz)$ first in $z$ and then in $\zz$ to obtain
$$
c_L(z,\zz)=-\frac{1}{L!}\sum_{\genfrac{}{}{0pt}{}{k,\kk\geq0}{k+\kk\leq d}}\sum_{k_1=1}^\infty\sum_{i=0}^{k_1}\sum_{\beta(\zz)=\frac{a\zz}{b\zz+1}\in\Sigma(\zz)}
c_{k,\kk}f_\beta z^{k+k_1}\zz^{\kk-k_1+i}\frac{\genfrac(){0pt}{}{k_1}{i}b^i}{k_1a^{k_1}}.
$$
Large values of $k_1$ give poles in $\zz$. We read off the $z^m\zz^{\mm-m}$ coefficient of $c_L$,
\begin{equation}\label{eqc}
c_{L,m,\mm-m}=-\frac{1}{L!}\sum_{\genfrac{}{}{0pt}{}{k,\kk\geq0}{k+\kk\leq\min\{d,\mm\}}}\sum_{\beta(\zz)=\frac{a\zz}{b\zz+1}\in\Sigma(\zz)}
c_{k,\kk}f_\beta\frac{\genfrac(){0pt}{}{m-k}{\mm-k-\kk}b^{\mm-k-\kk}}{(m-k)a^{m-k}},\quad\text{if }m>\mm.
\end{equation}
Because $Qf\in\sR^0$ we get $c_{L,m,\mm-m}=0$ for fixed $\mm$ and sufficiently large $m$. This leads to a system of linear equations for $c_{k,\kk}f_\beta$.
With an infinite number of values for $m$ we may consider $m$ as a variable (by interpolation). For $k+\kk=\mm$ the summand has poles $a^{k-m}/(m-k)$.
For all other values of $k,\kk$ the summands are holomorphic in $m$. By linear independence every pole term has to vanish separately. If $\mm\leq d$,
$$
\sum_{\beta(\zz)=\frac{a\zz}{b\zz+1}\in\Sigma(\zz)}c_{k,\mm-k}f_\beta\frac{1}{(m-k)a^{m-k}}\quad\text{for all }k=0,\ldots,\mm.
$$
Because at least one $c_{k,\mm-k}\neq0$ we get $\sum_\beta f_\beta\,a^{k-m}=0$ which is a Vandermonde system in the parameter $a$. The system has trivial kernel and we obtain
$$
\sum_{b:\beta(\zz)=\frac{a\zz}{b\zz+1}\in\Sigma(\zz)}f_\beta=0\quad\text{for all }a.
$$
We conclude that there exists at least one $\beta(\zz)=\frac{a\zz}{b\zz+1}\in\Sigma(\zz)$ with $b\neq0$ and $f_\beta\neq0$.
Now, we also consider $\mm$ as a variable. We multiply (\ref{eqc}) with $(m-\mm)!(\mm-d)!/(m-1)!$ and obtain for $\mm>d$,
$$
0=\sum_{\genfrac{}{}{0pt}{}{k,\kk\geq0}{k+\kk\leq d}}\sum_{\genfrac{}{}{0pt}{}{\beta(\zz)=\frac{a\zz}{b\zz+1}\in\Sigma(\zz)}{b\neq0}}
\frac{c_{k,\kk}f_\beta\,a^{k-m}b^{\mm-k-\kk}}{(m-k,k)(m-\mm+1,\kk)(\mm-d+1,d-k-\kk)},
$$
where we used the Pochhammer symbol $(x,n)=(x+n-1)!/(x-1)!$. Note that terms with $b=0$ vanish.
We prove by induction over $n\leq d$ that $c_{k,\kk}=0$ for all $k+\kk\leq n$. For $n=d$ this is a contradiction to $Q\neq0$.

For $n=0$ we consider the pole $1/\mm$ which only exists in the term $k=\kk=0$. By linear independence of pole terms we get $\sum_{\beta:b\neq0}c_{0,0}f_\beta a^{-m}b^\mm/(\mm-d+1,d)=0$.
After multiplication with $(\mm-d+1,d)$ this is a Vandermonde system for $a$ and $b$ in $m$ and $\mm$, respectively.
Because there exists a $\beta$ with $b\neq0$ and $f_\beta\neq0$ we get $c_{0,0}=0$.

For general $n\leq d$ we need to consider poles in $\mm$ and in $m$. To do this we use a partial fraction basis in $\mm$ whose coefficients are $b^{\mm-k-\kk}$ times a function in $m$
(see Remark \ref{sOrk}). Consider the pole $1/(\mm-n)$ which has contributions from terms with $k+\kk\leq n$. By induction we can restrict ourselves to terms with $k+\kk=n$ and obtain
for the coefficients
\begin{align*}
0&=\sum_{k=0}^n\sum_{\beta:b\neq0}\frac{c_{k,n-k}f_\beta\,a^{k-m}b^{\mm-n}}{(m-k,k)(m-n+1,n-k)(n-d+1,d-n)}\\
&=\frac{(m-n+1,n-1)}{(n-d+1,d-n)}\sum_{k=0}^n\sum_{\beta:b\neq0}\frac{c_{k,n-k}f_\beta\,a^{k-m}b^{\mm-n}}{m-k}.
\end{align*}
The individual terms in the last sum vanish as they have different poles in $m$. From the Vandermonde system we get $c_{k,n-k}f_\beta=0$ for all $\beta$ with $b\neq0$ and all $k=0,\ldots,n$. The claim follows because at least one $f_\beta\neq0$.
\end{proof}

\begin{remark}
The stronger $\sO_\FF(z,\zz)\sS^0_\Sigma\cap\sO_\FF(z,\zz)\sR^0=\{0\}$ is also true. Here, we only need (\ref{singreg}) which is easier to prove.
\end{remark}

\begin{thm}\label{decompthm}
Let $\Sigma\subset\FLT_\FF$ be finite and consider the point $z=0$. Every $f\in\sG\sH_\Sigma$ has a unique decomposition $f(z)=f_\mathrm{r}(z)+f_\mathrm{s}(z)$ into a regular part
$f_\mathrm{r}\in\sO_\Sigma(z,\zz)\sR^0\cap\sG\sH_\FF$ and a singular part $f_\mathrm{s}\in\sO_\Sigma(z,\zz)\sS^0_\Sigma\cap\sG\sH_\FF$.
If $f\in\sG\sH_\Sigma^\CC$ then $f_\mathrm{r}\in\sR^0\cap\sG\sH_\FF^\CC$ and $f_\mathrm{s}\in\sS^0_\Sigma\cap\sG\sH_\FF^\CC$.
\end{thm}
\begin{proof}
We first show the existence of the decomposition. By linearity it suffices to show the existence for the individual terms $f(z)=\phi(z)L_v(\zz)L_{w(\zz)}(z)$ in (\ref{GHeq}),
$\phi\in\sO_\Sigma(z,\zz)$, $v\in\overline{\Sigma(\zz)\cap\FF}^\ast$, $w(\zz)\in\Sigma(\zz)^\ast$. If $f\in\sG\sH_\Sigma^\CC$ then $\phi\in\CC$.
We use (\ref{singpropeq}) for $L_{w(\zz)}(z)$ and shuffle $\phi(z)L_v(\zz)$ with the result. Because $L_v(\zz)\in\sR^0(\zz)$, see (\ref{ex0}), we find
that the terms with $\log^\ell(z)$ in (\ref{singpropeq}) are in $\sO_\Sigma(z,\zz)\sR^0\cap\sG\sH_\FF$ and in $\sR^0\cap\sG\sH_\FF^\CC$ if $f\in\sG\sH_\Sigma^\CC$.
Likewise we get that the other terms are in $\sO_\Sigma(z,\zz)\sS_\Sigma^0\cap\sG\sH_\FF$ or in $\sS_\Sigma^0\cap\sG\sH_\FF^\CC$, respectively.

For uniqueness we assume that $f(z)=f_\mathrm{s}(z)+f_\mathrm{r}(z)=g_\mathrm{s}(z)+g_\mathrm{r}(z)$ are two decompositions. We get
$f_\mathrm{s}(z)-g_\mathrm{s}(z)=-f_\mathrm{r}(z)+g_\mathrm{r}(z)\in\sO_\Sigma(z,\zz)\sS^0_\Sigma\cap\sO_\Sigma(z,\zz)\sR^0\cap\sG\sH_\FF$.
By Proposition \ref{singprop3} we get $f_\mathrm{s}(z)=g_\mathrm{s}(z)$ and $f_\mathrm{r}(z)=g_\mathrm{r}(z)$.
\end{proof}
The above theorem implies that $f\in\sG\sH_\Sigma^\CC$ is in $\sR^0$ if and only if its singular part vanishes.

\section{Single-valued functions}\label{sectsv}
In the previous section we defined functions which may have one-dimensional singular loci in $\CC$ (such as $\log(z-\zz)$).
For single-valuedness we now consider a finite set $\Sigma\subset\CC$ of point-like singularities.

In the context of pQFT, single-valuedness is a fundamental property of all graphical functions (\cite{gf} with a proof in \cite{par}).
Graphical functions have single-valued log-Laurent expansions at their singular points $0$, $1$, and $\infty$ (conjectured in \cite{par,gf} with a partial proof in \cite{gfe}).

\begin{defn}\label{svdef}
Let $\Sigma$ be a finite set of points in $\CC$. A function $f$ on $\CC\backslash\Sigma$ has a single-valued log-Laurent expansion at $a\in\Sigma$ if
\begin{equation}\label{aexpansion}
f(z)=\sum_{\ell=0}^{L_a}\sum_{m,\mm=M_a}^\infty c_{\ell,m,\mm}^a[\log(z-a)(\zz-\aaa)]^\ell(z-a)^m(\zz-\aaa)^\mm
\end{equation}
in some neighborhood of $a$. A function $f$ has a single-valued log-Laurent expansion at $\infty$ if
\begin{equation}\label{inftyexpansion}
f(z)=\sum_{\ell=0}^{L_\infty}\sum_{m,\mm=-\infty}^{M_\infty}c_{\ell,m,\mm}^\infty[\log z\zz]^\ell z^m\zz^\mm
\end{equation}
in some neighborhood of $\infty$. We say that $f$ is $\CC$-analytic (i.e.\ $\CC$-valued real-analytic) at $a\in\CC\cup\{\infty\}$ if $L_a=M_a=0$.

The space of $\CC$-analytic functions on $\CC\backslash\Sigma$ with single-valued log-Laurent expansions at $\Sigma\cup\{\infty\}$ is $\sS\sV_\Sigma$. We also define
$$
\sS\sV_\CC=\bigcup_{\Sigma\subset\CC}\sS\sV_\Sigma.
$$
\end{defn}

\begin{remark}\label{svrk}\mbox{}
\begin{enumerate}
\item If $f$ has a single-valued log-Laurent expansion at $a$ then $f$ has trivial monodromy at $a$.
\item Because single-valuedness is defined locally it is clear that $\sS\sV_\Sigma$ is a bi-differential $\CC$-algebra.
\item In general, one expects that every $f\in\sS\sV_\Sigma$ has a primitive in $\sS\sV_\Sigma$. Because integration has a non-local character, it is not easy to construct these
primitives in general.
\item The space $\sS\sV_\CC$ is invariant under complex conjugation, $\overline{\sS\sV_\CC}=\sS\sV_\CC$.
\item The space $\sS\sV_\CC$ is invariant under linear transformations $z\mapsto az+b$ for $a,b\in\CC$ and under the inversion $z\mapsto1/z$. These transformations generate $\FLT_\CC$;
for any $f\in\sS\sV_\CC$ and $\beta\in\FLT_\CC$ we have $f(\beta(z))\in\sS\sV_\CC$.
\end{enumerate}
\end{remark}

\begin{ex}\label{exsv1}
From the expansion formulae (\ref{exa}), (\ref{exinfty}) for hyperlogarithms in $\sH\sL_\Sigma$ we deduce that every single-valued linear combination of products
of hyperlogarithms in $z$ and $\zz$ has single-valued log-Laurent expansions at $\Sigma\cup\{\infty\}$. These single-valued hyperlogarithms were constructed and studied by F. Brown in
\cite{BrSVMP,BrSVMPII}. Examples of low weights are ($D$ is the Bloch-Wigner dilogarithm (\ref{BWD}))
$$
\log(z)+\log(\zz)\in\sS\sV_{\{0\}},\quad\log(1-z)+\log(1-\zz)\in\sS\sV_{\{1\}},\quad D(z)\in\sS\sV_{\{0,1\}}.
$$
\end{ex}

\begin{ex}\label{counterex1}
The function $f(z)=\log(z-\zz+1)\in\sG\sH_{\{\zz-1\}}$ is $\CC$-analytic in $\CC$. It fails to have a single-valued log-Laurent expansion at infinity because $f(z)=\log(z)+\log(1-(\zz-1)/z)$
has increasing powers of $\zz$ in the expansion at $z=\infty$. Therefore $f\notin\sS\sV_\CC$.
\end{ex}

\begin{ex}\label{exsv2}
The function $f(z)=\log(z\zz+1)$ is in $\sS\sV_\emptyset$. It is real-analytic in $\CC$ and at infinity,
$$
\log(z\zz+1)=\log(z\zz)+\sum_{m=-\infty}^{-1}\frac{(-1)^m}{m}(z\zz)^m\quad\text{for }|z|>1.
$$
On the other hand, $\log(z\zz-1)$ is singular on the unit circle and hence not in $\sS\sV_\CC$.
\end{ex}

Generalizing the previous example (where $\beta(\zz)=-1/\zz$) we define for any field $\FF\subseteq\CC$
\begin{equation}\label{FLTempty}
\FLT_\FF^\emptyset(\zz)=\Big\{\beta(\zz)=\frac{a\zz+b}{c\zz+d}\in\FLT_\FF, c\neq0,\;\{z=\beta(\zz),z\in\CC\}=\emptyset\Big\}.
\end{equation}
Exclusion of the case $c=0$ avoids the situation of Example \ref{counterex1}.

\begin{ex}\label{exsv2a}
The function $f(z)=\log(az\zz+bz+c\zz+d)$ is in $\sS\sV_\emptyset$ if $-(c\zz+d)/(a\zz+b)\in\FLT_\CC^\emptyset(\zz)$.
By definition of $\FLT_\CC^\emptyset(\zz)$ it is clear that $f$ is $\CC$-analytic in $\CC$. Because $a\neq0$, $f$ has a single-valued log-Laurent expansion at infinity as in Example \ref{exsv2}.
\end{ex}
More interesting (and more important in pQFT, see Remark \ref{constrrk} (2)) are single-valued functions with letters which are not in $\FLT_\CC^\emptyset(\zz)$.

\begin{ex}\label{exsv3}
We will see in Example \ref{Dzzex} that there exists a single-valued primitive of $f(z)=D(z)/(z-\zz)$ (see also \cite{Duhr}).
Note that the denominator of $f$ vanishes on the real line. This singularity is lifted by the numerator, so that $f\in\sS\sV_{\{0,1\}}$.
\end{ex}

\begin{ex}
Graphical functions in even dimensions $\geq4$ (conjecturally) are in $\sS\sV_{\{0,1\}}$ \cite{gfe}.
\end{ex}

\begin{ex}\label{zm2zzex}
The function $1/(z-2\zz)$ is $\CC$-analytic in $\CC\backslash\{0\}$. It has trivial monodromy at 0 but fails to have a single-valued log-Laurant expansion at $0$.
Therefore $1/(z-2\zz)\notin\sS\sV_\CC$
\end{ex}

We define the projection onto the anti-residue-free part $\pi_{\partial_\zz}$ as the complex conjugate of (\ref{pip}),
$$
\pi_{\partial_\zz}f(z)=f(z)-\sum_{a\in\Sigma}\frac{\overline{\res}_af}{\zz-\aaa},
$$
where $\overline{\res}_af$ is the anti-residue of $f$ at $a$, i.e.\ the coefficient $c^a_{0,0,-1}$ in (\ref{aexpansion}) (whereas $\res_af=c^a_{0,-1,0}$).

In analogy to (\ref{evalm}) we have on $\sS\sV_\Sigma$
\begin{equation}\label{respar}
\res_a\partial_z=\overline{\res}_a\partial_\zz
\end{equation}
for any $a\in\Sigma$ because both sides are the coefficient $c^a_{1,0,0}$ in (\ref{aexpansion}).

\section{Generalized single-valued hyperlogarithms}\label{sectgsvh}
Generalized single-valued hyperlogarithms are generalized hyperlogarithms which are single-valued in the sense of Definition \ref{svdef}.

\subsection{Definitions and first results}
\begin{defn}\label{Gdef}
Let $\FF=\overline{\FF}$ be a quadratically closed number field, see (\ref{Fdef}). Let $\Sigma(\zz)\subset\FLT_\FF(\zz)$ be finite, see (\ref{FLTdef}). Then
\begin{equation}\label{gdef}
\sG_\Sigma=\sG\sH_\Sigma\cap\sS\sV_\CC
\end{equation}
is the space of generalized single-valued hyperlogarithms (GSVHs) on $\CC\backslash(\Sigma(\zz)\cap\FF)$, see Remark \ref{GSVHrk} (2).
Likewise, $\sG_\Sigma^\CC=\sG\sH_\Sigma^\CC\cap\sS\sV_\CC$ is the $\CC$-algebra of GSVHs with constant coefficients. We also define
$$
\sG_\FF=\bigcup_{\Sigma(\zz)\subset\FLT_\FF(\zz)}\sG_\Sigma,\quad\sG_\FF^\CC=\bigcup_{\Sigma(\zz)\subset\FLT_\FF(\zz)}\sG_\Sigma^\CC.
$$
The weight of $f\in\sG_\Sigma$ is the weight of $f\in\sG\sH_\Sigma$ in Definition \ref{GHdef}.
\end{defn}

\begin{ex}
Examples of GSVHs are in Examples \ref{exsv1}, \ref{exsv2}, \ref{exsv2a}, and \ref{exsv3}.
Graphical functions are not GSVHs in general because they can fail to be generalized hyperlogarithms.
\end{ex}

By Theorem \ref{GHthm} and Remark \ref{svrk} it is clear that $\sG_\FF$ is a bi-differential algebra which is closed under complex conjugation and transformations in $\FLT_\FF$.

\begin{remark}\label{GSVHrk}\mbox{}
\begin{enumerate}
\item In our context, constants have weight 0. One can set up a motivic version of GSVHs where, e.g., $\zeta(3)$ has weight 3. In this context one may add the motivic weight of
constants to the weight of GSVHs which lifts the weight to a grading on GSVHs. We will not pursue this here. The weight in the above definition is a filtration.
\item We will prove in Theorem \ref{GHSigmathm} that single-valued functions in $\sG\sH_\Sigma$ can only have singularities at $\Sigma(\zz)\cap\FF$ (and at infinity).
We have
$$
\sG_\Sigma=\sG\sH_\Sigma\cap\sS\sV_{\Sigma\cap\FF}.
$$
It is possible that a function in $\sG_\Sigma$ with minimal set $\Sigma(\zz)$ is $\CC$-analytic at a point in $\Sigma(\zz)\cap\FF$, see Example \ref{exexept}.
\end{enumerate}
\end{remark}

We will prove that for any finite set $\Sigma_0\subset\CC$ of singularities the space $\sG_\FF\cap\sS\sV_{\Sigma_0}$ is closed under taking (anti-)primitives (Theorem \ref{Gthm}).
To do this we will use the commutative hexagon in Figure \ref{fig:gsvh}
where we use the notation $\partial_z\sG_\FF^\CC$, $\partial_\zz\sG_\FF^\CC$, $\partial_z\partial_\zz\sG_\FF^\CC$ for the set of functions obtained by differentiating
$\sG_\FF^\CC$. The technical difficulty in the proof is to show Corollary \ref{partialcor} which states that an expression in $\partial_z\sG_\Sigma$ with linearly independent coefficients
has single-valued individual terms. We need a series of auxiliary results.

\begin{lem}\label{gsvhlemma}
Let $\beta\in\FLT_\FF\backslash\FF$ and $a\in\FF$ with $a=\beta(\aaa)$. Let $f(z,\zz)$ be $\CC$-analytic at $a$.
Then, there exists a function $g(z,\zz)$ which is $\CC$-analytic at $a$ such that
$$
\frac{f(z,\zz)}{z-\beta(\zz)}=\frac{f(\beta(\zz),\zz)}{z-\beta(\zz)}+g(z,\zz).
$$
\end{lem}
\begin{proof}
Fix $\zz$ in a neighborhood of $\aaa$ such that $f$ is holomorphic at $z=\beta(\zz)$ (we have $\beta(\aaa)=a$), $f(z,\zz)=\sum_{k=0}^\infty f_k(\zz)(z-\beta(\zz))^k$ with anti-holomorphic functions $f_k$. This gives
$$
g(z,\zz)=\sum_{k=0}^\infty f_{k+1}(\zz)[z-a-(\beta(\zz)-\beta(\aaa))]^k.
$$
Because $\beta(\aaa)\neq\infty$ we get that $\beta(\zz)-\beta(\aaa)$ expands into a Taylor series in $\zz-\aaa$ with low degree one.
The above series expands into a Taylor series at $(z,\zz)=(a,\aaa)$.
\end{proof}

\begin{ex}\label{wt2ex0}
Consider $f(z)=\log(z\zz)/(z-1/\zz)$. Because $\log(z\zz)\in\sG_{\{0\}}$ we obtain from Lemma \ref{gsvhlemma} that $f$ is $\CC$-analytic at $a$ for $|a|=1$ and hence
at $\CC\backslash\{0\}$. Because $(z-1/\zz)^{-1}$ is $\CC$-analytic at $0$ and at $\infty$ we get $f\in\sG_{\{0,\zz^{-1}\}}$.
\end{ex}

\begin{lem}\label{gsvhlem1}
Let $\beta_i\in\FLT_\FF\backslash\FF$ and $a\in\FF$ with $a=\beta_i(\aaa)$ for $i=1,\ldots,n$ and $\beta_i\neq\beta_j$ for $i\neq j$.
Let $f_i\in\sG_\FF$. Then $\sum_if_i(z,\zz)/(z-\beta_i(\zz))$ has a single-valued log-Laurent expansion (\ref{aexpansion}) at $a$
if and only if $f_i(\beta_i(\zz),\zz)=0$ for all $i$ in a neighborhood of $\zz=\aaa$.
\end{lem}
\begin{proof}
By (\ref{aexpansion}) we have
\begin{equation}\label{eqgsvh1}
f_i(z,\zz)=\sum_{\ell=0}^L\frac{[\log(z-a)(\zz-\aaa)]^\ell}{[(z-a)(\zz-\aaa)]^{-M}}f_{\ell,i}(z,\zz),
\end{equation}
with $\CC$-analytic functions $f_{\ell,i}$ at $a$ and $L,M\in\ZZ$. Because $f_i\in\sG_\FF$ we can use (\ref{GHeq}) to see that the coefficients of $\log^\ell(\zz-\aaa)$ in the
expansion of $f_i$ are in $\sG\sH_\FF$. Hence $f_{\ell,i}\in\sG\sH_\FF$. Using Lemma \ref{gsvhlemma} for $f_{\ell,i}(z,\zz)$ gives
\begin{equation}\label{eqpflem1}
\sum_{i=1}^n\frac{f_i(z,\zz)}{z-\beta_i(\zz)}-g(z,\zz)=\sum_{\ell=0}^L\frac{[\log(z-a)(\zz-\aaa)]^\ell}{[(z-a)(\zz-\aaa)]^{-M}}
\sum_{i=1}^n\frac{f_{\ell,i}(\beta_i(\zz),\zz)}{z-\beta_i(\zz)}
\end{equation}
with a function $g$ which has a single-valued log-Laurent expansion at $a$.

Assume $f_1(\beta_1(\zz),\zz)=\ldots=f_n(\beta_n(\zz),\zz)=0$ but there exists an $i$ such that $f_{\ell,i}(\beta_i(\zz),\zz)\neq0$ for some $\ell$.
Let $L'$ be the maximum of these $\ell$. From (\ref{eqgsvh1}) we get ($a=\beta_i(\aaa)$)
$$
0=\sum_{\ell=0}^{L'}\frac{[\log(\beta_i(\zz)-\beta_i(\aaa))(\zz-\aaa)]^\ell}{[(\beta_i(\zz)-\beta_i(\aaa))(\zz-\aaa)]^{-M}}f_{\ell,i}(\beta_i(\zz),\zz).
$$
Because the $f_{\ell,i}$ are $\CC$-analytic at $a$ we get that the $f_{\ell,i}(\beta_i(\zz),\zz)$ are anti-holomorphic at $\zz=\aaa$.
The coefficient of $\log^{L'}(\zz-\aaa)$ on the right hand side is zero which implies $f_{L',i}(\beta_i(\zz),\zz)=0$. The contradiction gives $f_{\ell,i}(\beta_i(\zz),\zz)=0$ for all
$\ell,i$ and the right hand side of (\ref{eqpflem1}) vanishes.

If there exists an $i$ such that $f_i(\beta(\zz),\zz)\neq0$ then $f_{\ell,i}(\beta_i(\zz),\zz)\neq0$ for some $\ell$. Let $L'$ be the maximum of these $\ell$.
To use Theorem \ref{decompthm} we shift $(z,\zz)\mapsto(z+a,\zz+\aaa)$ in (\ref{eqpflem1}). By partial fraction decomposition in $z$ we obtain a primitive $H$ of the right hand side.
From (\ref{0}) we get
$$
H(z,\zz)=L'![(\beta_i(\zz+\aaa)-a)\zz]^Mf_{L',i}(\beta_i(\zz+\aaa),\zz+\aaa))L_{0^{\{L'\}}(\beta_i(\zz+\aaa)-a)}(z)+\ldots.
$$
Because $f_{\ell,i}\in\sG\sH_\FF$ we get $H\in\sG\sH_\FF$. The above term of $H$ is in $\sS^0_\Sigma$ for $\Sigma(\zz)={\{0,\beta_i(\zz+\aaa)-a\}}$.

In Theorem \ref{decompthm} we get $H_{\mathrm{s}}\neq0$. Hence $H\notin\sR^0\subset\sO_\Sigma(z,\zz)\sR^0$. The expansion of $H$ in $z$ generates poles in $\zz$ of increasing order
(see the proof of Proposition \ref{singprop3}). This structure is preserved upon differentiation (this is trivial for the maximum power of $\log(z\zz)$)).
Because $g(z+a,\zz+\aaa)\in\sR^0$ in (\ref{eqpflem1}) this implies that
$\sum_if_i(z+a,\zz+\aaa)/(z+a-\beta_i(\zz+\aaa))\notin\sR^0$. So, $\sum_if_i(z,\zz)/(z-\beta_i(\zz))$ has no single-valued log-Laurent expansion at $a$.
\end{proof}

\begin{lem}\label{gsvhlem2}
Let $\beta_i(\zz)=a_i\zz+b_i$ with $a_i,b_i\in\FF$, $a_i\neq0$, for $i=1,\ldots,n$, and $\beta_i\neq\beta_j$ for $i\neq j$.
Let $f_i\in\sG_\FF$. Then $\sum_if_i(z,\zz)/(z-\beta_i(\zz))$ has a single-valued log-Laurent expansion (\ref{inftyexpansion})
at $\infty$ if and only if $f_i(\beta_i(\zz),\zz)=0$ for all $i$ in a neighborhood of $\zz=\infty$.
\end{lem}
\begin{proof}
Let $g(z,\zz)=\sum_if_i(1/z,1/\zz)/(1/z-\beta_i(1/\zz))$. We need to study $g$ at $z=\zz=0$. We get
$$
g(z,\zz)=\sum_{i=1}^n\frac{f_i(1/z,1/\zz)}{1/z-a_i/\zz-b_i}=\sum_{i=1}^n\frac{-z\zz f_i(1/z,1/\zz)}{(b_i\zz+a_i)(z-\zz/(b_i\zz+a_i))}.
$$
Consider $\tilde f_i=-z\zz f_i(1/z,1/\zz)/(b_i\zz+a_i)$ and $\tilde\beta_i(\zz)=\zz/(b_i\zz+a_i)$.
We have $\tilde f_i\in\sG_\FF$ by Theorem \ref{GHthm} (5) and $\tilde\beta_i(0)=0$ because $a_i\neq0$.
With Lemma \ref{gsvhlem1} we obtain that $g$ has a single-valued log-Laurent expansion at 0 if and only if
$$
-\frac{\tilde\beta_i(\zz)\zz f_i(1/\tilde\beta_i(\zz),1/\zz)}{b_i\zz+a_i}=-\frac{\zz^2 f_i(a_i/\zz+b_i,1/\zz)}{(b_i\zz+a_i)^2}=0
$$
for all $i$ in a neighborhood of $\zz=0$. This is equivalent to $f_i(a_i\zz+b_i,\zz)=0$ for all $i$ in a neighborhood of $\zz=\infty$.
\end{proof}

\begin{ex}\label{Dzzex0}
Consider $f(z)=D(z)/(z-\zz)$, see Example \ref{exsv3}. By Example \ref{exsv1} we have $D\in\sG_{\{0,1\}}$. If $D=D(z,\zz)$ is considered as a function in $z,\zz$ then $D(\zz,\zz)=0$
in the neighborhood of any $\zz=\aaa\in\RR\cup\{\infty\}$. From  Lemma \ref{gsvhlemma} we get that $f$ is $\CC$-analytic at $\RR\backslash\{0,1\}$ and hence in $\CC\backslash\{0,1\}$.
By Lemmas \ref{gsvhlem1} and \ref{gsvhlem2} the function $f$ has single-valued log-Laurent expansions at $z=0,1,\infty$. So, $f\in\sG_{\{0,1,\zz\}}$.
\end{ex}

\begin{thm}\label{techthm}
Let $\Sigma(\zz)\subset\FLT_\FF(\zz)$ be finite. Let $\sB_\Sigma$ be a $\CC$-basis of $\sO_\Sigma(z,\zz)$, see (\ref{OSigmadef}).
Every $f\in\sG_\Sigma$ has a unique $\sB_\Sigma$ representation as
\begin{equation}\label{basisexp}
f(z)=\sum_{\phi\in\sB_\Sigma}\phi(z)h_\phi(z)\quad\text{with}\quad h_\phi\in\sG_\Sigma^\CC.
\end{equation}
\end{thm}
\begin{proof}
Existence and uniqueness of a $\sB_\Sigma$ representation with $h_\phi\in\sG\sH_\Sigma^\CC$ follow from Theorem \ref{GHthm} (1). We need to show that $h_\phi\in\sG_\Sigma^\CC$.

We first show that $h_\phi\in\sR^0$. From $f\in\sG_\Sigma$ and Theorem \ref{decompthm} we get that in the decomposition of $f$ we have
$0=f_{\mathrm{s}}=\sum_\phi\phi h_{\phi,\mathrm{s}}$ with $h_\phi=h_{\phi,\mathrm{r}}+h_{\phi,\mathrm{s}}$ and $h_{\phi,\mathrm{r}}, h_{\phi,\mathrm{s}}\in\sG\sH_\FF^\CC$.
From Theorem \ref{GHthm} (1) we get $h_{\phi,\mathrm{s}}=0$ by linear independence of the $\phi(z)$. Therefore $h_\phi=h_{\phi,\mathrm{r}}\in\sR^0$.

Consider the monodromy $\sM_0$ at 0, see Section \ref{sectmono}. Because $f\in\sS\sV_\CC$ we have $f=\sM_0f=\sum_\phi\phi\sM_0h_\phi$ with $\sM_0h_\phi\in\sG\sH_\Sigma^\CC$.
By linear independence of the $\phi$ we get $\sM_0h_\phi=h_\phi$. We rearrange the (multivalued) log-Laurent expansion of $h_\phi$ at $z=0$ so that it has the terms
$(\log z\zz)^\ell(\log\zz)^\ll z^m\zz^\mm$. We get $\sM_0(\log z\zz)^\ell(\log\zz)^\ll=(\log z\zz)^\ell(\log\zz)^\ll-2\pi\ii\ll(\log\zz)^{\ll-1}$ plus terms with
lower powers in $\log\zz$. If there exist terms with $\ll\geq1$ we consider the two leading powers in $\log\zz$ yielding $\ll=0$.
By contradiction, $h_\phi$ has a single-valued log-Laurent expansion at 0.

For general $a\in\CC$ we use $f(z+a)\in\sG_\FF$ which follows from Theorem \ref{GHthm} (5) and Remark \ref{svrk} (5). We get that $h_\phi(z+a)$ has a single-valued log-Laurent expansion at 0.
Equivalently, $h_\phi(z)$ has a single-valued log-Laurent expansion at $a$. Considering $f(1/z)$ we find that $h_\phi(1/z)$ has a single-valued log-Laurent expansion at 0.
Therefore $h_\phi(z)$ has a single-valued log-Laurent expansion at $\infty$.

In the neighborhood of the expansion points single-valued log-Laurent expansions are $\CC$-analytic. By compactness of $\CC\cup\{\infty\}$ it follows that $h_\phi$ only has a finite number
of singularities. Hence $h_\phi\in\sS\sV_\CC$. This implies $h_\phi\in\sG_\Sigma^\CC=\sG\sH_\Sigma^\CC\cap\sS\sV_\CC$.
\end{proof}

\begin{remark}
It is possible to define $\delta_\beta^{\mathrm{dR}}=\res_\beta\partial_z$ on $\sG_\Sigma^\CC$ for any $\beta(\zz)\in\Sigma(\zz)$ in analogy to (\ref{drdef}).
Theorem \ref{techthm} implies that $\delta_\beta^{\mathrm{dR}}$ is an endomorphism of $\sG_\Sigma^\CC$.
To prove this property is a core difficulty in this article. In Corollary \ref{partialcor} we prove the stronger result that $(z-\beta(\zz))^{-1}\delta_\beta^{\mathrm{dR}}$
maps $\sG_\Sigma^\CC$ into $\sG_\Sigma$.
\end{remark}

\begin{thm}\label{GHSigmathm}
let $\Sigma(\zz)\subset\FLT_\FF(\zz)$ be finite. Then $\sG_\Sigma=\sG\sH_\Sigma\cap\sS\sV_{\Sigma\cap\FF}$ and $\sG_\Sigma^\CC=\sG\sH_\Sigma^\CC\cap\sS\sV_{\Sigma\cap\FF}$.
If $f\in\sG_\Sigma^\CC$ then $M_a=0$ in the single-valued log-Laurent expansions (\ref{aexpansion}), (\ref{inftyexpansion}) at $a\in(\Sigma(\zz)\cap\FF)\cup\{\infty\}$.
\end{thm}
\begin{proof}
Any function in $\sG_\FF^\CC$ has single-valued log-Laurent expansions with $M_a=0$. For the expansion in $z$ this follows from (\ref{GHeq}), Lemma \ref{explem}, and
(\ref{exinfty}). By complex conjugation, Theorem \ref{GHthm} (4), this also holds for the expansion in $\zz$. So, we only need to prove the first statement.

We first show that $\sG_\Sigma^\CC=\sG\sH_\Sigma^\CC\cap\sS\sV_{\Sigma\cap\FF}$. We trivially have $\sG_\Sigma^\CC\supseteq\sG\sH_\Sigma^\CC\cap\sS\sV_{\Sigma\cap\FF}$.
For $\sG_\Sigma^\CC\subseteq\sG\sH_\Sigma^\CC\cap\sS\sV_{\Sigma\cap\FF}$ we fix $a\in\CC\backslash(\Sigma(\zz)\cap\FF)$. Let $f\in\sG_\Sigma^\CC$ have weight $n$.
We prove by induction over $n$ that $f$ is $\CC$-analytic at $a$. For $n=0$ the function $f$ is constant and the claim is trivial. If $n\geq1$ then
$\partial_zf(z)=g(z)/(z-\beta(\zz))\in\sG_\Sigma$ with $\beta(\zz)\in\Sigma(\zz)$ and $g\in\sG_\Sigma^\CC$ with weight $n-1$. By induction $g$ is $\CC$-analytic at $a$.

For $\beta$ we have two possibilities: If $a\neq\beta(\aaa)$ then $(z-\beta(\zz))^{-1}$ is $\CC$-analytic at $a$. Therefore $\partial_zf$ is $\CC$-analytic at $a$.
Otherwise $a=\beta(\aaa)\notin\Sigma(\zz)\cap\FF$. Therefore $\beta\notin\FF$. We use Lemma \ref{gsvhlem1} yielding $g(\beta(\zz),\zz)=0$ in a neighborhood of $\zz=\aaa$.
From Lemma \ref{gsvhlemma} we get that $\partial_zf$ is  $\CC$-analytic at $a$.

Consider the single-valued log-Laurent expansion of $f\in\sG_\Sigma^\CC$ at $a$. By explicit differentiation with respect to $z$ we see that a $\CC$-analytic $\partial_zf$ comes
from a $\CC$-analytic $f$ up to a pole term in $\CC[1/(\zz-\aaa)]$. The pole term is zero because $M_a=0$. Hence $f\in\sG\sH_\Sigma^\CC\cap\sS\sV_{\Sigma\cap\FF}$.

Now, let $f\in\sG_\Sigma$ and fix $a\notin\Sigma(\zz)\cap\FF$. We need to show that $f$ is $\CC$-analytic at $a$.
Consider a basis of $\sO_\Sigma(z,\zz)$ with each element in $\sO_\Sigma(z,\zz)$. From Theorem \ref{techthm} we get $h_\phi\in\sG_\Sigma^\CC$.
Hence, $h_\phi\in\sS\sV_{\Sigma\cap\FF}$. Therefore $h_\phi$ is $\CC$-analytic at $a$.

We consider $z$ and $\zz$ as independent variables. The function $\phi\in\sO_\Sigma(z,\zz)$ has poles at $\{z=\beta(\zz),\beta(\zz)\in\Sigma(\zz)\}$ and at $\{\zz=\bb,b\in\Sigma(\zz)\cap\FF\}$.
Because $a\neq\beta(\zz)$, $\phi$ is holomorphic at $z=a$. Hence $f$ is holomorphic at $a$. Likewise, $\phi$ is anti-holomorphic at $\zz=\aaa$ because $z\neq\beta(\aaa)$ and $\aaa\neq\bb$
for all $b\in\Sigma(\zz)\cap\FF$. So, $f$ is anti-holomorphic at $\zz=\aaa$. Because $f$ has a single-valued log-Laurent expansion at $a$ we get that $f$ is $\CC$-analytic at $a$.
\end{proof}

\begin{defn}\label{O0def}
Let $\Sigma(\zz)\subset\FLT_\FF(\zz)$ be finite. With $\FLT_\FF^\emptyset(\zz)$ as in (\ref{FLTempty}) we define
\begin{equation}\label{OSigma0def}
\sO_\Sigma^\emptyset(z,\zz)=\CC[z,\zz,((z-\beta(\zz))^{-1})_{\beta(\zz)\in\Sigma(\zz)\cap\FLT_\FF^\emptyset(\zz)},((z-b)^{-1},(\zz-\bb)^{-1})_{b\in\Sigma(\zz)\cap\FF}].
\end{equation}
\end{defn}

\begin{lem}\label{wt0lem}
The weight zero piece of $\sG_\Sigma$ is $\sO_\Sigma^\emptyset(z,\zz)$: $\sG_\Sigma\cap\sO_\Sigma(z,\zz)=\sO_\Sigma^\emptyset(z,\zz)$.
\end{lem}
\begin{proof}
We have $\sO_\Sigma^\emptyset(z,\zz)\subseteq\sG_\Sigma\cap\sO_\Sigma(z,\zz)$ in analogy to Example \ref{exsv2a}.

Assume that a canceled fraction $f\in\sG_\Sigma\cap\sO_\Sigma(z,\zz)$ has a factor $z-\beta(\zz)$ with $\beta(\zz)\in\Sigma(\zz)\backslash(\FLT_\FF^\emptyset(\zz)\cup\FF)$ in the denominator,
$$
f(z,\zz)=\frac{P(z,\zz)}{(z-\beta(\zz))Q(z,\zz)},
$$
with polynomials $P,Q\in\CC[z,\zz]$. We have $P(z,\zz)/(z-\beta(\zz))\in\sG_\Sigma$. By definition of $\FLT_\FF^\emptyset$ there either exists
an $a\in\CC$ with $a=\beta(\aaa)$ or $\beta$ is linear but not constant. In any case we get from Lemmas \ref{gsvhlem1} and \ref{gsvhlem2} that $P(\beta(\zz),\zz)=0$
in some neighborhood of $\zz=\aaa$ or $\infty$ (respectively). Therefore, the factor $z-\beta(\zz)$ cancels in $f$. By contradiction we obtain $f\in\sO_\Sigma^\emptyset(z,\zz)$.
\end{proof}

\begin{defn}\label{pGdef}
Let $\Sigma(\zz)\subset\FLT_\FF(\zz)$ be finite. Then
\begin{equation}\label{pzghdef}
\partial_z\sG\sH_\Sigma^\CC=\langle(z-\beta(\zz))^{-1},\,\beta(\zz)\in\Sigma(\zz)\rangle_{\sG\sH_\Sigma^\CC}
\end{equation}
is the space of generalized hyperlogarithms with simple pole coefficients. In analogy to (\ref{gdef}) we define $\partial_z\sG_\Sigma^\CC=\partial_z\sG\sH_\Sigma^\CC\cap\sS\sV_\CC$.
Moreover, $\partial_\zz\partial_z\sG_\Sigma^\CC=\partial_\zz(\partial_z\sG_\Sigma^\CC)$ and
$$
\partial_z\sG_\FF^\CC=\bigcup_{\Sigma(\zz)\subset\FLT_\FF(\zz)}\partial_z\sG_\Sigma^\CC,\quad\partial_\zz\sG_\FF^\CC=\overline{\partial_z\sG_\FF^\CC},
\quad\text{and}\quad\partial_\zz\partial_z\sG_\FF^\CC=\bigcup_{\Sigma(\zz)\subset\FLT_\FF(\zz)}\partial_\zz\partial_z\sG_\Sigma^\CC.
$$
\end{defn}
Clearly, $\partial_z(\sG_\FF^\CC)\subseteq\partial_z\sG_\FF^\CC$.
We also have $\partial_\zz(\sG_\FF^\CC)\subseteq\partial_\zz\sG_\FF^\CC$, see {\bf I8} in Section \ref{sectitint} and Lemmas \ref{partialFLT} and \ref{betalem1}.
We will show that equality holds (Corollary \ref{existencecor}) which also implies that $\partial_\zz\partial_z\sG_\FF^\CC=\partial_z(\partial_\zz\sG_\FF^\CC)$.

\begin{prop}\label{techprop}
Let $\Sigma(\zz)\subset\FLT_\FF(\zz)$ be finite and let $\sB_\Sigma$ be a $\CC$-basis of $\langle(z-\beta(\zz))^{-1},\beta(\zz)\in\Sigma(\zz)\rangle_{\sO_\Sigma^\emptyset(z,\zz)}$.
If $f\in\sO_\Sigma^\emptyset(z,\zz)\partial_z\sG\sH_\Sigma^\CC\cap\sS\sV_\CC\subseteq\sG_\Sigma$,
$$
f(z)=\sum_{\phi\in\sB_\Sigma}\phi(z)h_\phi(z),\quad h_\phi\in\sG\sH_\Sigma^\CC,
$$
then $\phi h_\phi\in\sS\sV_\CC$.
\end{prop}
\begin{proof}
The statement of the theorem is stable under basis transformations. Therefore, we can assume without restriction that $\sB_\Sigma$ is the basis of partial fractions in
$z$ and $\zz$ defined in Remark \ref{sOrk}.

Let $f(z)=\sum_{\beta,b,m,\mm}g^{\beta,b}_{m,\mm}(z,\zz)(z-\beta(\zz))^m(\zz-\bb)^\mm\in\sO_\Sigma^\emptyset(z,\zz)\partial_z\sG\sH_\Sigma^\CC\cap\sS\sV_\CC$
where the sum is finite with $g^{\beta,b}_{m,\mm}\in\sG\sH_\Sigma^\CC$, $\beta(\zz)\in\Sigma(\zz)$, $b\in\FF$, $m,\mm\in\ZZ$.
From Theorem \ref{techthm} we get $g^{\beta,b}_{m,\mm}\in\sG_\Sigma^\CC$.

Fix $a\in\CC$. If $a\neq\beta\in\FF$ or $\beta(\aaa)\neq a$ then $(z-\beta(\zz))^{-1}$ is $\CC$-analytic at $a$.
In this case $g^{\beta,b}_{m,\mm}(z,\zz)(z-\beta(\zz))^m(\zz-\bb)^\mm$ inherits the single-valued log-Laurent expansion from $g^{\beta,b}_{m,\mm}$.
This also holds in the case $\beta=a$.

By subtracting these terms from the sum we obtain a reduced sum $\sum_{\beta,b,\mm}g^{\beta,b}_{-1,\mm}(z,\zz)(z-\beta(\zz))^{-1}(\zz-\bb)^\mm$ where all $\beta$ have the property
$\beta\notin\FF$ and $\beta(\aaa)=a$. We have $m=-1$ because $(z-\beta(\zz))^{-1}\notin\sO_\Sigma^\emptyset(z,\zz)$. The reduced sum has a single-valued log-Laurent expansion at $a$
(because $f$ and all subtracted terms have such an expansion).
Lemma \ref{gsvhlem1} for $n=1$ and $f_1(z,\zz)=\sum_{b,\mm}g^{\beta,b}_{-1,\mm}(z,\zz)(\zz-\bb)^\mm$ gives $f_1(\beta(\zz),\zz)=0$ in some neighborhood of $\zz=\aaa$.
We have $g^{\beta,b}_{-1,\mm}(\beta(\zz),\zz)\in\sH\sL_\FF^\CC(\zz)$ and by linear independence (Theorem \ref{Brthm} (1)) we get $g^{\beta,b}_{-1,\mm}(\beta(\zz),\zz)=0$
for all $b$, $\mm$.
Lemma \ref{gsvhlem1} for $n=1$ and $f_1(z,\zz)=g^{\beta,b}_{-1,\mm}(z,\zz)(\zz-\bb)^\mm$ gives that $g^{\beta,b}_{-1,\mm}(z,\zz)(z-\beta(\zz))^{-1}(\zz-\bb)^\mm$ has a single-valued log-Laurent expansion at $a$.

At infinity we observe that terms with $\beta\in\FF$ or $\beta(\infty)\neq\infty$ trivially have single-valued log-Laurent expansions at $\infty$.
We reduce to a sum $\sum_{\beta,b,\mm}g^{\beta,b}_{-1,\mm}(z,\zz)(z-\beta(\zz))^{-1}(\zz-\bb)^\mm$ with linear $\beta\notin\FF$ (again $(z-\beta(\zz))^{-1}\notin\sO_\Sigma^\emptyset(z,\zz)$).
We use Lemma \ref{gsvhlem2} to prove that $g^{\beta,b}_{-1,\mm}(z,\zz)(z-\beta(\zz))^{-1}(\zz-\bb)^\mm$ has a single-valued log-Laurent expansion at $\infty$.

By compactness of $\CC\cup\{\infty\}$, every term $g^{\beta,b}_{-1,\mm}(z,\zz)(z-\beta(\zz))^{-1}(\zz-\bb)^\mm$ has a finite number of singular points and hence is in $\sS\sV_\CC$.
\end{proof}

\begin{ex}\label{wt2ex1}
Proposition \ref{techprop} does not generalize to all functions $f\in\sG_\Sigma$.
Consider $f(z)=\log(z\zz)/(z-1/\zz)$ (see Example \ref{wt2ex0}). We have $f\in\partial_z\sG_{\{0,\zz^{-1}\}}^\CC$ and
$$
\partial_zf(z)=\frac{\zz}{z-1/\zz}-\frac{\zz}{z}-\frac{\log(z\zz)}{(z-1/\zz)^2}\in\sG_{\{0,\zz^{-1}\}}.
$$
Expansion at $z=0$ shows that the terms $\zz/(z-1/\zz)$ and $-\log(z\zz)/(z-1/\zz)^2$ are not in $\sS\sV_\CC$.
\end{ex}

\begin{cor}\label{partialcor}
Let $\Sigma(\zz)\subset\FLT_\FF(\zz)$ be finite. Every individual term in the $\sG\sH_\Sigma^\CC$ span of $\partial_z\sG_\Sigma^\CC$ (see (\ref{pzghdef})) is single-valued
(i.e.\ in $\partial_z\sG_\Sigma^\CC$).
\end{cor}
\begin{proof}
Extend $\{(z-\beta(\zz))^{-1}\}$ to a $\CC$-basis of $\langle(z-\beta(\zz))^{-1},\beta(\zz)\in\Sigma(\zz)\rangle_{\sO_\Sigma^\emptyset(z,\zz)}$ and use Proposition \ref{techprop}.
\end{proof}

\subsection{The commutative hexagon}
\begin{lem}\label{kerlem}
The kernels of $\partial_z$ and $\partial_\zz$ in $\sG_\FF$ are $\sO_\FF(\zz)$ and $\sO_\FF(z)$, respectively.
In particular, the kernels of $\partial_z$ and $\partial_\zz$ in $\sG_\FF^\CC$ are $\CC$.
\end{lem}
\begin{proof}
Let $f\in\sG_\FF$ with $\partial_zf=0$. Then, by Theorem \ref{GHthm} (2), we have $f\in\sH\sL_\FF(\zz)$.
Because $f\in\sS\sV_\CC$ we get from Proposition \ref{propseqm} that $f\in\sO_\FF(\zz)$. For $f\in\sG_\FF^\CC$ we get $f\in\sO_\FF(\zz)\cap\sG_\FF^\CC=\CC$.
The result for the kernel of $\partial_\zz$ follows by complex conjugation.
\end{proof}

\begin{defn}\label{intsvdef}
Let $f\in\partial_z\sG_\FF^\CC$. Then
$$
\int_{\mathrm{sv}}f(z)\,\dd z=F(z)\in\sG_\FF^\CC
$$
if $\partial_zF=f$ and $F(0)=0$ (with regularization, see Section \ref{sectreg}).
For $g\in\partial_\zz\partial_z\sG_\FF^\CC$ we say that a function $G\in\partial_\zz\sG_\FF^\CC$ is a single-valued primitive of $g$, i.e.\ $G=\int_{\mathrm{sv}}g\,\dd z$,
if $\partial_zG=g$. We define the complex conjugate $\int_{\mathrm{sv}}\dd\zz$ such that complex conjugation commutes with integration.
\end{defn}

Note that by Lemma \ref{kerlem} $\int_{\mathrm{sv}}\dd z$ is unique on $\partial_z\sG_\FF^\CC$ whereas it is multivalued on $\partial_\zz\partial_z\sG_\FF^\CC$.
Single-valued primitives on $\partial_\zz\partial_z\sG_\FF^\CC$ differ by anti-analytic functions in $\partial_\zz\sG_\FF^\CC$ which, by Proposition \ref{propseqm},
are in $\sO_\FF(\zz)$.

We need the commutativity of the hexagon in Figure \ref{fig:gsvh} to prove the existence of $\int_{\mathrm{sv}}\dd z$ (which will be extended to $\sG_\FF$ in Definition \ref{intsvGdef}).

\begin{thm}\label{hexthm}
The hexagon in Figure \ref{fig:gsvh} commutes. In particular, for any $f\in\partial_z\sG_\FF^\CC$ there exists a unique single-valued primitive $\int_{\mathrm{sv}}f\dd z\in\sG_\FF^\CC$.
\end{thm}
\begin{proof}
We first show that the hexagon is well-defined. Let $f\in\partial_\zz\partial_z\sG_\FF^\CC$. The primitive $\int_{\mathrm{sv}}f\dd z$
is defined up to an anti-analytic rational function $g\in\sO_\FF(\zz)\cap\partial_\zz\sG_\FF^\CC$. We get $g(\zz)=\sum_{a\in\FF}c_a/(\zz-\aaa)$ with $c_a\in\CC$.
This implies $\pi_{\partial_\zz}g=0$ which ensures that the hexagon is well-defined.

Now, consider $f\in\partial_z\sG_\FF^\CC$. We use induction over the weight $n$ of $f$ in the set $\{-\infty,0,1,2,\ldots\}$.

If $n=-\infty$ then $f=0$ and $\int_{\mathrm{sv}}\pi_{\partial_z}f\dd z=0$. Likewise, we have $G=\int_{\mathrm{sv}}\partial_\zz f\dd z=0$ and $\int_{\mathrm{sv}}\pi_{\partial_\zz}G\,\dd\zz=0$.

Now, assume that $n\geq0$. By linearity and Corollary \ref{partialcor} we may assume that
\begin{equation}\label{eqpf1}
f(z)=\frac{h(z)}{z-\beta(\zz)},\quad h\in\sG_\FF^\CC.
\end{equation}
In particular (see Definition \ref{GHdef}), $h$ is a $\CC$-linear combination of iterated integrals
$$
L_v(\zz)L_{w(\zz)}(z)=I(0,a_1\ldots a_\ell,\zz)I(0,\beta_1(\zz)\ldots\beta_m(\zz),z),
$$
where $\ell+m\leq n$, $a_i\in\FF$, $i=1,\ldots,\ell$, $\beta_k(\zz)\in\FLT_\FF(\zz)$, $k=1,\ldots,m$. By {\bf I8} in Section \ref{sectitint} we get
\begin{eqnarray*}
\partial_\zz\frac{L_v(\zz)L_{w(\zz)}(z)}{z-\beta(\zz)}&=&[\partial_\zz\partial_z\log(z-\beta(\zz))]L_v(\zz)L_{w(\zz)}(z)\\
&&+\,[\partial_z\log(z-\beta(\zz))]\Bigg[\frac{L_{v^\ell}(\zz)L_{w(\zz)}(z)}{\zz-a_\ell}+
L_v(\zz)\sum_{k=1}^m\Big(\partial_\zz\log\frac{\beta_{k+1}-\beta_k}{\beta_k-\beta_{k-1}}\Big)L_{w(\zz)^k}(z)\Bigg],
\end{eqnarray*}
where $\beta_0=0$ and $\beta_{m+1}=z$. The superscripts in $v^\ell$ and $w(\zz)^k$ refer to letters which have to be removed from the words $v$ and $w(\zz)$, respectively.
For the empty word $e$ we define $L_{e^i}=0$. Using Lemma \ref{betalem1} (2) we write the term $[\partial_z\log(z-\beta(\zz))]L_v(\zz)[\partial_\zz\log(z-\beta_m(\zz))]L_{w(\zz)^m}(z)$
in the above sum as
$$
\Bigg[[\partial_\zz\log(z-\beta(\zz))][\partial_z\log(z-\beta_m(\zz))]+[\partial_\zz\log(\beta(\zz)-\beta_m(\zz))]\Big[\partial_z\log\frac{z-\beta(\zz)}{z-\beta_m(\zz)}\Big]\Bigg]
L_v(\zz)L_{w(\zz)^m}(z).
$$
We get
\begin{equation}\label{eqpf2}
\partial_\zz\frac{L_v(\zz)L_{w(\zz)}(z)}{z-\beta(\zz)}=\partial_z[(\partial_\zz\log(z-\beta(\zz)))L_v(\zz)L_{w(\zz)}(z)]+g_{v,w(\zz)}(z)
\end{equation}
with
\begin{eqnarray}\label{geq}
g_{v,w(\zz)}(z)&=&\sum_{k=1}^m\left(\partial_\zz\log\left.\frac{\beta_{k+1}(\zz)-\beta_k(\zz)}{\beta_k(\zz)-\beta_{k-1}(\zz)}\right|_{\beta_{m+1}=\beta}\right)
\frac{L_v(\zz)L_{w(\zz)^k}(z)}{z-\beta(\zz)}+[\partial_\zz\log(\zz-a_\ell)]\frac{L_{v^\ell}(\zz)I_{w(\zz)}(z)}{z-\beta(\zz)}\nonumber\\
&&\quad-\,[\partial_\zz\log(\beta(\zz)-\beta_m(\zz))]\frac{L_v(\zz)L_{w(\zz)^m}(z)}{z-\beta_m(\zz)}.
\end{eqnarray}
We return to the function $f$ by summing over the words $v$ and $w(\zz)$ with their coefficients in $\CC$. Equation (\ref{eqpf2}) becomes
\begin{equation}\label{eqpf3}
\partial_\zz f(z)=\partial_z[(\partial_\zz\log(z-\beta(\zz)))h(z)]+g(z)=\partial_z[-f(z)\partial_\zz\beta(\zz)]+g(z),
\end{equation}
where $g$ is given by the sum over the $g_{v,w(\zz)}$. Because $\partial_zf,\partial_\zz f\in\sS\sV_\CC$ we get $g\in\sG_\FF$. By (\ref{geq}) we get
$g\in\sO_\FF(\zz)\partial_z\sG\sH^\CC_\FF$. Hence $g\in\sO_\Sigma^\emptyset(z,\zz)\partial_z\sG\sH^\CC_\Sigma$ for some $\Sigma(\zz)\subset\FLT_\FF(\zz)$.
We use Proposition \ref{techprop} and obtain that in the partial fraction basis of Remark \ref{sOrk}
$$
g(z,\zz)=\sum_{a\in\FF}\frac{g_a(z,\zz)}{\zz-\aaa},\quad g_a\in\sS\sV_\CC,
$$
where we used Lemma \ref{betalem1} (1) to restrict the (finite) sum to values in $\FF$. From (\ref{geq}) we get $g_a\in\partial_z\sG\sH_\FF^\CC$, hence $g_a\in\partial_z\sG_\FF^\CC$.

The weight of the $g_a$ is less than $n$. We use induction to obtain primitives $G_a\in\sG_\FF^\CC$ for the $g_a$. Altogether we get
$$
\partial_\zz f(z)=\partial_z G(z)\quad\text{with}\quad G(z)=[\partial_\zz\log(z-\beta(\zz))]h(z)+\sum_{a\in\FF}\frac{G_a(z)}{\zz-\aaa}.
$$
The first term in $G$ is in $\sS\sV_\CC$, see (\ref{eqpf3}). By Lemma \ref{partialFLT} it is in $\partial_\zz\sG\sH_\FF^\CC$ and hence it is in $\partial_\zz\sG_\FF^\CC$.
The second term is in $\partial_\zz\sG_\FF^\CC$ because $G_a\in\sG_\FF^\CC$. Therefore $\int_{\mathrm{sv}}\partial_\zz f\dd z=G$.

Now, we construct a single-valued primitive of $\pi_{\partial_z}f$. Consider $z$ and $\zz$ as independent variables and determine a primitive $F_1$
of $\pi_{\partial_z}f$ in $\sG\sH_\FF^\CC$ by multivalued integration with respect to $z$ using (\ref{pLw}). With $G_1=\partial_\zz F_1\in\partial_\zz\sG\sH_\FF^\CC$ we have
$$
\partial_z G_1(z,\zz)=\partial_z\partial_\zz F_1(z,\zz)=\partial_\zz\pi_{\partial_z}f(z,\zz)=\partial_\zz f(z,\zz)=\partial_zG(z,\zz)=\partial_z\pi_{\partial_\zz}G(z,\zz).
$$
By Theorem \ref{GHthm} (2) we get
$$
G_1(z,\zz)=\pi_{\partial_\zz}G(z,\zz)+G_2(\zz),\quad G_2(\zz)\in\partial_\zz\sG\sH_\FF^\CC\cap\sH\sL_\FF(\zz).
$$
Multivalued integration with respect to $\zz$ provides an anti-primitive $F_2(\zz)\in\sH\sL_\FF^\CC(\zz)$ of $G_2(\zz)$. We define
$$
F(z,\zz)=F_1(z,\zz)-F_2(\zz)\in\sG\sH_\FF^\CC
$$
and get
$$
\partial_z F(z,\zz)=\pi_{\partial_z}f(z,\zz)\quad\text{and}\quad\partial_\zz F(z,\zz)=G_1(z,\zz)-G_2(\zz)=\pi_{\partial_\zz}G(z,\zz)=\pi_{\partial_\zz}\int_{\mathrm{sv}}\partial_\zz f(z,\zz)\,\dd z.
$$

We have to show that $F$ is single valued.
Fix a point $a$ in $\CC\cup\{\infty\}$. The function $\partial_zF=\pi_{\partial_z}f\in\sS\sV_\CC$ has a single-valued log-Laurent expansion (\ref{aexpansion}), (\ref{inftyexpansion})
at $a$. By explicit integration $\partial_zF$ has a primitive $\tilde F$ with a single-valued log-Laurent expansion at $a$.
The function $F\in\sG\sH_\FF^\CC$ has a (possibly multivalued) log-Laurent expansion at $(z,\zz)=(a,\aaa)$ which has to coincide with $\tilde F$ up to an anti-analytic function.
Because $\partial_\zz (F-\tilde F)$ is single-valued at $a$ we get for $a\in\CC$,
$$
\partial_\zz[F(z)-\tilde F(z)]=\sum_{\mm=M_a}^\infty c_\mm\,(\zz-\aaa)^\mm,
$$
with an analogous formula for $a=\infty$. By explicit integration we get
\begin{equation}\label{mainpf1}
F(z)=\tilde F(z)+\sum_{\mm\neq-1}\frac{c_\mm}{\mm+1}(\zz-\aaa)^{\mm+1}+c_{-1}\log(\zz-\aaa)+d
\end{equation}
for some constant $d\in\CC$. Because $\partial_\zz F$ has no anti-residues we get
$$
\overline{\res}_a\partial_\zz\tilde F=-c_{-1}.
$$
By (\ref{respar}) this implies $\res_a\partial_z\tilde F=-c_{-1}$. If we apply $\res_a\partial_z$ to both sides of (\ref{mainpf1}) and use that $\partial_z F$ is residue-free
we obtain $0=-c_{-1}$. From (\ref{mainpf1}) and the single-valuedness of $\tilde F$ we get that $F$ is single-valued at $a$.

Because $\CC\cup\{\infty\}$ is compact $F$ can only have a finite number of singularities, hence $F\in\sS\sV_\CC$. Then ($F(0)$ is the regularized limit of $F(z)$ at $z=0$)
$$
\int_{\mathrm{sv}}\pi_{\partial_z}f(z)\,\dd z=\int_{\mathrm{sv}}\pi_{\partial_\zz}G(z)\,\dd\zz=F(z)-F(0)\in\sG_\FF^\CC
$$
is the (unique) single-valued primitive of $\pi_{\partial_z}f$ and anti-primitive of $\pi_{\partial_\zz}G$.
\end{proof}

Note that the proof of the existence of single-valued primitives is fully constructive, see Section \ref{sectimplementation}.
By symmetry under complex conjugation (Theorem \ref{GHthm} (4)) we also get single-valued anti-primitives for functions in $\partial_\zz\sG_\FF^\CC$.
In general, the set of letters $\Sigma(\zz)$ is not stable under taking single-valued (anti\nobreakdash-)primitives. It may be augmented by constants related to
factorizing differences of FLTs in Lemma \ref{betalem1} (1).

\begin{cor}\label{existencecor}
Every $f\in\partial_z\sG_\FF^\CC$ ($\overline{f}\in\partial_\zz\sG_\FF^\CC$) has a single-valued (anti-)primitive $\int_{\mathrm{sv}}f\dd z$ ($\int_{\mathrm{sv}}\overline{f}\dd\zz$).
In particular,
\begin{equation}\label{existenceeq}
\partial_z\sG_\FF^\CC=\partial_z(\sG_\FF^\CC),\quad\partial_\zz\sG_\FF^\CC=\partial_\zz(\sG_\FF^\CC),\quad
\partial_\zz\partial_z\sG_\FF^\CC=\partial_\zz\partial_z(\sG_\FF^\CC)=\partial_z(\partial_\zz\sG_\FF^\CC).
\end{equation}
\end{cor}
\begin{proof}
Let $f\in\partial_z\sG_\Sigma^\CC$ for some finite $\Sigma(\zz)\subset\FLT_\FF(\zz)$. We have $f\in\sS\sV_{\Sigma_0}$ for some finite set $\Sigma_0\subset\CC$. Define
$$
F(z)=\int_{\mathrm{sv}}\pi_{\partial_z}f(z)\,\dd z+\sum_{a\in\Sigma_0}(\res_af)[L_a(z)+L_\aaa(\zz)].
$$
The first term exists by Theorem \ref{hexthm}. We get $F\in\sG_\FF^\CC$, $\partial_zF=f$ by (\ref{pip}), and $F(0)=0$.
Hence $F=\int_{\mathrm{sv}}f\dd z$. Equations (\ref{existenceeq}) follow.
\end{proof}

\subsection{Examples}
\begin{ex}\label{SVIex}
For $\Sigma\subset\FF$ we inductively define single-valued hyperlogarithms $\sL_w(z)$ for $w\in\Sigma^\ast$ by (see Definition \ref{GSVIdef} for an extension to GSVHs)
\begin{equation}\label{SVIeq}
\sL_e(z)=1,\quad\sL_{wa}(z)=\int_{\mathrm{sv}}\frac{\sL_w(z)}{z-a}\,\dd z,\quad w\in\Sigma^\ast,\;a\in\Sigma.
\end{equation}
For every word $w\in\Sigma^\ast$ we get that $\sL_w\in\sG_\Sigma^\CC$ is a sum of terms $L_{\overline{u}}(\zz)L_v(z)$ with $u,v\in\Sigma^\ast$ (see \cite{BrSVMP,BrSVMPII} and Example \ref{HLex}).
\end{ex}
As in the commutative hexagon in Figure \ref{fig:hlog} the subtraction of residues vital.

\begin{ex}\label{1zexa}
Let $f(z)=z^{-1}\in\partial_z\sG_{\{0\}}^\CC$ (see Example \ref{1zex}). Then $\int_{\mathrm{sv}}f(z)\,\dd z=\sL_0(z)=\log(z\zz)$ whereas $\partial_\zz f=0$.
The hexagon commutes because $\pi_{\partial_z}f=0$.
\end{ex}

\begin{ex}\label{Labexa}
Let $f(z)=\sL_a(z)/(z-b)\in\partial_z\sG_{\{a,b\}}^\CC$ for $a,b\in\FF$, $a\neq b$. Then $\pi_{\partial_z}f(z)=[\sL_a(z)-\sL_a(b)]/(z-b)=[L_a(z)+L_\aaa(\zz)-\sL_a(b)]/(z-b)$ and
$$
\int_{\mathrm{sv}}\pi_{\partial_z}f(z)\dd z=\sL_{ab}(z)-\sL_a(b)\sL_b(z)\in\sG_{\{a,b\}}^\CC.
$$
We use the commutative hexagon in Figure \ref{fig:gsvh} to obtain the representation (\ref{GHeq}) for $\sL_{ab}(z)$.
Multivalued integration of $\pi_{\partial_z}f(z)$ gives
\begin{equation}\label{ex1eq}
L_{ab}(z)+L_\aaa(\zz)L_b(z)-\sL_a(b)L_b(z)+g(\zz)
\end{equation}
for some anti-analytic function $g(\zz)$.

On the other hand we have $\partial_\zz f(z)=1/[(\zz-\aaa)(z-b)]$ and $\int_{\mathrm{sv}}\partial_\zz f(z)\,\dd z=\sL_b(z)/(\zz-\aaa)$.
Subtraction of the anti-residue gives $(\sL_b(z)-\sL_b(a))/(\zz-\aaa)$. Multivalued integration with respect to $\zz$ leads to
\begin{equation}\label{ex1eq2}
L_\aaa(\zz)L_b(z)+L_{\overline{b}\aaa}(\zz)-\sL_b(a)L_{\aaa}(\zz)+h(z)
\end{equation}
for some analytic function $h(z)$. The term $L_\aaa(\zz)L_b(z)$ has both, a $z$- and a $\zz$-dependence and is hence reproduced in (\ref{ex1eq}). This is a general feature of the
approach: the calculation of these terms in either of the integrations is redundant. From the anti-analytic terms in the above equation we read off $g(\zz)$ while the analytic
terms in (\ref{ex1eq}) provide $h(z)$,
$$
g(\zz)=L_{\overline{b}\aaa}(\zz)-\sL_b(a)L_{\aaa}(\zz),\qquad h(z)=L_{ab}(z)-\sL_a(b)L_b(z).
$$
The hexagon commutes. We get
$$
\sL_{ab}(z)=L_{ab}(z)+L_\aaa(\zz)L_b(z)+L_{\overline{b}\aaa}(\zz)+\sL_a(b)L_b(\zz)-\sL_b(a)L_{\aaa}(\zz)
$$
with $\sL_{ab}(0)=0$ (as required). We see that the projections $\pi_{\partial_z}$, $\pi_{\partial_\zz}$ are vital for the hexagon to commute.
Note that for the calculation of $\sL_{ab}$ it suffices to consider $z$ and $\zz$ as independent variables and extract $g(\zz)$ from the limit $z=0$ in (\ref{ex1eq2}), see (\ref{inteq}).
\end{ex}

The above examples are single-valued hyperlogarithms \cite{BrSVMP,BrSVMPII}. In the following examples we construct genuine GSVHs of weights 1, 2, and 3.

\begin{ex}\label{wt1ex}
Consider $f(z)=(z+1/\zz)^{-1}\in\partial_z\sG_{\{-\zz^{-1}\}}^\CC$, see Example \ref{exsv2}. Because $z\zz\geq0$ the function $f$ has no singularities in $\CC$. So, $\pi_{\partial_z}f=f$ and
$$
\int_{\mathrm{sv}}f(z)\,\dd z=\log(z+1/\zz)+g(\zz).
$$
Differentiation of $f$ with respect to $\zz$ yields $(z\zz+1)^{-2}$ which integrates (with $\int_{\mathrm{sv}}\dd z$) to
$(\zz+1/z)^{-1}\in\partial_\zz\sG_\FF^\CC$. The result has no anti-residue, so that
$$
\int_{\mathrm{sv}}f(z)\,\dd z=\log(\zz+1/z)+h(z).
$$
Comparison of both equations leads to
$$
\int_{\mathrm{sv}}\frac{1}{z\zz+1}\,\dd z=\frac{\log(z\zz+1)}{\zz}.
$$
\end{ex}

\begin{ex}\label{wt2ex}
Consider $f(z)=\log(z\zz)/(z-1/\zz)\in\partial_z\sG_{\{0,\zz^{-1}\}}^\CC$, see Example \ref{wt2ex0}, Section \ref{sectinvols}, and Example 4.11 in \cite{numfunct}.
The singularity at $z=1/\zz$ is lifted by the numerator. We have $\pi_{\partial_z}f=f$ and with $\log(z\zz)=\sL_0(z)=L_0(z)+L_0(\zz)$ we obtain
$$
\int_{\mathrm{sv}}f(z)\,\dd z=L_{0,\zz^{-1}}(z)+L_0(\zz)L_{\zz^{-1}}(z)+g(\zz).
$$
Differentiation of $f$ with respect to $\zz$ yields $-\log(z\zz)/(z\zz-1)^2+1/(z\zz-1)$. The sum is in $\sG_\FF$ although the individual terms are not.
We use integration by parts for the first term to obtain
\begin{equation}\label{wt2exeq1}
\int_{\mathrm{sv}}\left(-\frac{\log(z\zz)}{(z\zz-1)^2}+\frac{1}{z\zz-1}\right)\dd z=\frac{\log(z\zz)}{\zz-1/z}\in\partial_\zz\sG_\FF^\CC.
\end{equation}
The result has no anti-residues and we get
\begin{equation}\label{wt2exeq}
\int_{\mathrm{sv}} f(z)\,\dd z=L_{0,z^{-1}}(\zz)+L_0(z)L_{z^{-1}}(\zz)+h(z).
\end{equation}
The result is not in the $z,\zz$ basis of (\ref{GHeq}). With the role of $z$ and $\zz$ interchanged, we seem to need a transformation into the $z,\zz$ basis before we
are able to extract $g(\zz)$ and $h(z)$. More efficiently, we use the redundancy of the procedure, set $z=0$ in (\ref{wt2exeq1}) (for general $\zz$), and determine $g(\zz)$ by
multivalued integration with respect to $\zz$, see (\ref{inteq}). Because $\lim_{z\to0}\log(z\zz)/(\zz-1/z)=0$ for all $\zz$, we get $g(\zz)=0$,
$$
\int_{\mathrm{sv}}\frac{\log(z\zz)}{z\zz-1}\,\dd z=\frac{L_{0,\zz^{-1}}(z)+L_0(\zz)L_{\zz^{-1}}(z)}{\zz}.
$$
\end{ex}

\begin{ex}\label{wt2exzz}
To show that we can also use the commutative hexagon to integrate with $\int_{\mathrm{sv}}\dd\zz$ in the $z,\zz$ basis, we now consider the complex conjugate of the previous example.
We take $f(z)=\log(z\zz)/(\zz-1/z)\in\partial_\zz\sG_\FF^\CC$. The complex conjugate commutative hexagon gives the complex conjugates of the previous equations.
We use the complex conjugate of (\ref{wt2exeq}) to obtain
$$
\int_{\mathrm{sv}}f(z)\,\dd\zz=L_{0,\zz^{-1}}(z)+L_0(\zz)L_{\zz^{-1}}(z)+h(\zz).
$$
The right hand side is in the desired $z,\zz$ basis and we determine $h(\zz)$ by taking the limit $z=0$ in $f(z)$ before the multivalued integration with respect to $\zz$.
Because $f(0,\zz)=0$ we obtain
$$
\int_{\mathrm{sv}}\frac{\log(z\zz)}{z\zz-1}\,\dd\zz=\frac{L_{0,\zz^{-1}}(z)+L_0(\zz)L_{\zz^{-1}}(z)}{z}.
$$
\end{ex}

\begin{ex}\label{Dzzex}
Consider $[\sL_{01}(z)-\sL_{10}(z)]/(z-\zz)=4\ii D(z)/(z-\zz)$, see Example \ref{Dzzex0} and Section \ref{sect01zz}.
The single-valued primitive of $D(z)/(z-\zz)\in\partial_z\sG_{\{0,1,\zz\}}^\CC$ was first constructed (using the Galois coaction) in \cite{Duhr}. The construction with the
commutative hexagon is analogous to the previous examples. We obtain
\begin{eqnarray*}
\int_{\mathrm{sv}}\frac{\sL_{01}(z)-\sL_{10}(z)}{z-\zz}\,\dd z&=&L_{01\zz}(z)-L_{10\zz}(z)+L_0(\zz)L_{1\zz}(z)-L_1(\zz)L_{0\zz}(z)\\
&&\quad+\,L_{10}(\zz)L_\zz(z)-L_{01}(\zz)L_\zz(z)+L_{100}(\zz)-L_{101}(\zz)\in\sG_{\{0,1,\zz\}}^\CC.
\end{eqnarray*}
This is the first (by weight) GSVH in the alphabet $0,1,\zz$ which is not a single-valued multiple polylogarithm.
The alphabet $0,1,\zz$ is of particular importance in pQFT because it is ubiquitous in dimensionally regularized amplitudes \cite{numfunct,7loops}.
\end{ex}

\subsection{The structure theorem for GSVHs}
We use the partial fraction basis in Remark \ref{sOrk} to extend single-valued integration to $\sG_\FF$.

\begin{defn}\label{intsvGdef}
Let $f=(z-\beta(\zz))^mh(z)\in\sG_\FF$, $\beta\in\FLT_\FF$, $h\in\sG_\FF^\CC$. If $m=-1$ then $f\in\partial_z\sG_\FF^\CC$ and $\int_{\mathrm{sv}}f\dd z$ is (uniquely) defined
in Definition \ref{intsvdef}. Otherwise we use integration by parts to inductively define
\begin{equation}\label{ibp}
\int_{\mathrm{sv}}(z-\beta(\zz))^mh(z)\,\dd z=\frac{(z-\beta(\zz))^{m+1}}{m+1}h(z)-\frac{1}{m+1}\int_{\mathrm{sv}}(z-\beta(\zz))^{m+1}\partial_z h(z)\,\dd z.
\end{equation}
We extend the definition to $\sG_\FF$ by linearity in $\sO_\FF(\zz)$ using the partial fraction basis in Remark \ref{sOrk}.
The complex conjugate integral $\int_{\mathrm{sv}}\dd\zz$ is defined such that complex conjugation commutes with integration.
\end{defn}

\begin{prop}\label{consisprop}
The single-valued integrals $\int_{\mathrm{sv}}\dd z$ and $\int_{\mathrm{sv}}\dd\zz$ in Definition \ref{intsvGdef} are well-defined endomorphisms of $\sG_\FF$
which are consistent with the single-valued integrals on $\partial_\zz\partial_z\sG_\FF^\CC$ in Definition \ref{intsvdef}.

For every $f\in\sG_\FF$ we get $\partial_z\int_{\mathrm{sv}}f\dd z=f$ and $\partial_\zz\int_{\mathrm{sv}}f\dd\zz=f$.
\end{prop}
\begin{proof}
We have $\int_{\mathrm{sv}}0\,\dd z=0$. Using linearity in the partial fraction basis of Remark \ref{sOrk} for $0\neq f\in\sG_\FF$ we may assume that $f(z)=(z-\beta(\zz))^mh(z)$,
with $\beta\in\FLT_\FF$, $m\in\ZZ$, $h\in\sG_\FF^\CC$. We can iterate (\ref{ibp}) because $(z-\beta(\zz))^{m+1}\partial_z h(z)=((z-\beta(\zz))\partial_z-m)f(z)\in\sG_\FF$.
With every iteration the weight of $h$ is reduced until $h=0$. This ensures that the integration algorithm terminates after a finite number of steps.

We show that $\int_{\mathrm{sv}}\dd z$ maps into $\sG_\FF$. For $m=-1$ this is proved in Theorem \ref{hexthm}.
For $m\neq-1$ it suffices to show that both terms on the right hand side of (\ref{ibp}) are single-valued.
The first term equals $f(z)(z-\beta(\zz))/(m+1)\in\sG_\FF$ while the second term is single-valued by induction over the weight of $h$.
So, for every $f\in\sG_\FF$ we get a well-defined function $F=\int_{\mathrm{sv}}f\dd z\in\sG_\FF$.

By induction over the weight of $f$ we get $\partial_zF=f$ because (\ref{ibp}) becomes an identity upon differentiation with respect to $z$.

Let $f\in\partial_\zz\partial_z\sG_\FF^\CC$ and $F=\int_{\mathrm{sv}}f\dd z$ according to Definition \ref{intsvGdef}. We need to show that $F\in\partial_\zz\sG_\FF^\CC$.
By Corollary \ref{existencecor} we have $f\in\partial_z(\partial_\zz\sG_\FF^\CC)$,
$$
f(z)=\partial_z\sum_{\beta\in\FLT_\FF}\frac{h_\beta(z)}{\zz-\beta(z)}=\sum_{\beta\in\FLT_\FF}\Big[[\partial_\zz\partial_z\log(\zz-\beta(z))]h_\beta(z)+\frac{\partial_zh_\beta(z)}{\zz-\beta(z)}\Big],
$$
where the sum is finite and $h_\beta\in\sG_\FF^\CC$.
From (\ref{GHeq}) we get $\partial_zh_\beta(z)=g_\beta(z)/(z-\gamma(\zz))$ with $g_\beta\in\sG_\FF^\CC$
and $\gamma\in\FLT_\FF$. The terms with $\beta\in\FF$ integrate to $F_0(z)=\sum_{\beta\in\FF}h_\beta(z)/(\zz-\beta)$. 
In the other terms $\beta(z)=(az+b)/(cz+d)$ is invertible and we get from Lemma \ref{partialFLT},
$$
\partial_z\frac{h_\beta(z)}{\zz-\beta(z)}=\frac{[\partial_\zz\beta^{-1}(\zz)]h_\beta(z)}{(z-\beta^{-1}(\zz))^2}+\frac{g_\beta(z)}{(\zz-\beta(z))(z-\gamma(\zz))}.
$$
Using Lemma \ref{partialFLT} with $(z,\beta)\mapsto(\zz,\beta^{-1})$ we get
$$
-\frac{\partial_\zz\beta^{-1}(\zz)}{z-\beta^{-1}(\zz)}=\frac{1}{\zz-\beta(z)}-\frac{1}{\zz-a/c},
$$
where the second term is absent if $c=0$. With this equation, integration by parts in (\ref{ibp}) gives for the integral over the terms with $\beta\notin\FF$,
$$
F_1(z)=\sum_{\beta\in\FLT_\FF\backslash\FF}\Big[\frac{h_\beta(z)}{\zz-\beta(z)}-\frac{h_\beta(z)}{\zz-a/c}+\int_{\mathrm{sv}}\frac{g_\beta(z)}{(\zz-a/c)(z-\gamma(\zz))}\,\dd z\Big]
=\sum_{\beta\in\FLT_\FF\backslash\FF}\frac{h_\beta(z)}{\zz-\beta(z)},
$$
where we used $\int_{\mathrm{sv}}\frac{g_\beta(z)}{z-\gamma(\zz)}\,\dd z=h_\beta(z)$ according to the case $m=-1$ in Definition \ref{intsvGdef}.
Altogether, $F(z)=F_0(z)+F_1(z)=\sum_{\beta\in\FLT_\FF}h_\beta/(\zz-\beta(z))\in\partial_\zz\sG_\FF^\CC$.

By complex conjugation Proposition \ref{consisprop} also holds for integration with respect to $\zz$.
\end{proof}

From Theorem \ref{hexthm} we obtain the following structure theorem for GSVHs.

\begin{thm}\label{Gthm}
Let $\FF=\overline{\FF}$ be a quadratically closed field, see (\ref{Fdef}).
Let $\Sigma(\zz)\subset\FLT_\FF(\zz)$ and $\Sigma_0\subset\FF$ be finite sets of letters and singular points, respectively.
\begin{enumerate}
\item Let $\sO_\Sigma^\emptyset(z,\zz)$ be the ring of regular functions, see Definition \ref{O0def}. Then
$\sG_\Sigma$ is a free $\sO_\Sigma^\emptyset(z,\zz)$-module which is closed under multiplication. In particular, $\sG_\Sigma^\CC$ is a $\CC$-algebra.
\item The sequence
$$
0\longrightarrow\sO_{\Sigma_0}(\zz)\longrightarrow\sG_\FF\cap\sS\sV_{\Sigma_0}\stackrel{\partial_z}{\longrightarrow}\sG_\FF\cap\sS\sV_{\Sigma_0}\longrightarrow0
$$
is exact. I.e.\ the kernel of $\partial_z$ in $\sG_\FF\cap\sS\sV_{\Sigma_0}$ is $\sO_{\Sigma_0}(\zz)$ and every $f\in\sG_\FF\cap\sS\sV_{\Sigma_0}$ has a primitive
$F\in\sG_\FF\cap\sS\sV_{\Sigma_0}$ with $\partial_zF=f$.
\item $\sG_\Sigma$ is differentially simple. I.e.\ for every $0\neq f\in\sG_\Sigma$ there exists a differential operator $D$ such that $Df=1$.
\item $\sG_\FF$ and $\sG_\FF^\CC$ are stable under complex conjugation,
\begin{equation}\label{Gcc}
\overline{\sG_\FF}=\sG_\FF,\quad\overline{\sG_\FF^\CC}=\sG_\FF^\CC.
\end{equation}
In particular, the sequence
$$
0\longrightarrow\sO_{\Sigma_0}(z)\longrightarrow\sG_\FF\cap\sS\sV_{\Sigma_0}\stackrel{\partial_\zz}{\longrightarrow}\sG_\FF\cap\sS\sV_{\Sigma_0}\longrightarrow0
$$
is exact.
\item $\sG_\FF$ and $\sG_\FF^\CC$ are stable under transformations in $\FLT_\FF$. I.e.\ for any $f\in\sG_\FF$, $g\in\sG_\FF^\CC$, and $\beta\in\FLT_\FF$,
\begin{equation}\label{Gflt}
f(\beta(z))\in\sG_\FF,\quad\text{and}\quad g(\beta(z))\in\sG_\FF^\CC.
\end{equation}
\item Every $f\in\sG_\Sigma$ is in $\sS\sV_{\Sigma\cap\FF}$. If $f\in\sG_\Sigma^\CC$, the single-valued log-Laurent expansions (\ref{aexpansion}), (\ref{inftyexpansion}) at
$a\in(\Sigma(\zz)\cap\FF)\cup\{\infty\}$ have $M_a=0$.
\end{enumerate}
\end{thm}
\begin{proof}
It is clear from Theorem \ref{GHthm} (1) and Remark \ref{svrk} (2) that $\sG_\Sigma^\CC$ is a $\CC$-algebra.
Lemma \ref{wt0lem} states that the weight zero part of $\sG_\Sigma$ is $\sO_\Sigma^\emptyset(z,\zz)$. Hence $\sG_\Sigma$ is a $\sO_\Sigma^\emptyset(z,\zz)$-module.
It is free by Theorem \ref{GHthm} (1).

By Theorem \ref{GHthm} (2) and Proposition \ref{propseqm} the kernel of $\partial_z$ in $\sG_\FF\cap\sS\sV_{\Sigma_0}$ is $\sO_\FF(\zz)\cap\sS\sV_{\Sigma_0}=\sO_{\Sigma_0}(\zz)$.
In Proposition \ref{consisprop} it was proved that every $f\in\sG_\FF\cap\sS\sV_{\Sigma_0}$ has a primitive $\tilde F\in\sG_{\tilde\Sigma}$ for some finite set
$\tilde\Sigma(\zz)\subset\FLT_\FF(\zz)$. By Theorem \ref{GHSigmathm} the function $\tilde F$ is $\CC$-analytic outside $\tilde\Sigma_0:=\tilde\Sigma(\zz)\cap\FF$.
If $\tilde\Sigma_0\subset\Sigma_0$ then $F:=\tilde F\in\sS\sV_{\Sigma_0}$. Otherwise let $a\in\tilde\Sigma_0\backslash\Sigma_0$.
Because $f\in\sG_{\Sigma_0}$, the single-valued log-Laurent expansion of $\tilde F$ at $a$ differentiates (with respect to $z$) to a $\CC$-analytic function at $a$.
By explicit differentiation of its expansion we see that $\tilde F$ is the sum of a $\CC$-analytic function at $a$ and a pole part $G_a(\zz)\in\CC[(\zz-\aaa)^{-1}]$.
We consider all points in $\tilde\Sigma_0\backslash\Sigma_0$ and set $F:=\tilde F-\sum_{a\in\tilde\Sigma_0\backslash\Sigma_0}G_a$.
We get $\partial_zF=\partial_z\tilde F=f$ and $F\in\sS\sV_{\Sigma_0}$.

To prove (3) we use the differential operator $D$ that was constructed in the proof of Theorem \ref{GHthm} (3).

Statement (4) follows from Theorem \ref{GHthm} (4) and Remark \ref{svrk} (4).

Statement (5) follows from Theorem \ref{GHthm} (5) and Remark \ref{svrk} (5).

Statement (6) is Theorem \ref{GHSigmathm}.
\end{proof}
Theorem \ref{Gthm} (5) and (6) suggest that GSVHs should be considered as objects on the Riemann sphere $\CC\cup\{\infty\}$.

\section{Single-valued integration}\label{sectimplementation}
The commutative hexagon in Figure \ref{fig:gsvh} establishes a bootstrap algorithm for the construction of GSVHs.
The crucial step is the single-valued integration in $\partial_z\partial_\zz\sG_\FF^\CC$ on the bottom right sector.
After using integration by parts there is always a weight-drop in this integration. The subtraction of residues with $\pi_{\partial_z}$ needs to be compensated
by adding single-valued logarithms. With the strategy of Examples \ref{wt2ex} and \ref{wt2exzz} we get the following inductive formulae for the single-valued integral
of a function $f\in\partial_z\sG_\Sigma^\CC$, $\Sigma(\zz)\subset\FLT(\zz)$, in the $z,\zz$ basis of (\ref{GHeq}) (in terms of the multivalued integrals
$\int_0\dd z$ and $\int_0\dd\zz$, see (\ref{int0def})),
\begin{eqnarray}\label{inteq}
\int_{\mathrm{sv}}f(z)\,\dd z&=&\int_0\pi_{\partial_z}f(z)\,\dd z+\int_0\Big(\pi_{\partial_\zz}\int_{\mathrm{sv}}\partial_\zz f(z)\,\dd z\Big)\Big|_{z=0}\dd\zz
+\sum_{a\in\Sigma(\zz)\cap\FF}(\res_af)\sL_a(z),\nonumber\\
\int_{\mathrm{sv}}f(z)\,\dd\zz&=&\int_0\Big(\pi_{\partial_\zz}f(z)\Big)\Big|_{z=0}\dd\zz+\int_0\pi_{\partial_z}\int_{\mathrm{sv}}\partial_zf(z)\,\dd\zz\,\dd z
+\sum_{a\in\Sigma(\zz)\cap\FF}(\overline{\res}_af)\sL_a(z),
\end{eqnarray}
where we consider $z$ and $\zz$ as independent variables. Both integrations in (\ref{inteq}) preserve the $z,\zz$ basis.

\subsection{Extension of GSVHs}\label{sectGSI}
Equation (\ref{inteq}) was used for first calculations with GSVHs in a predecessor of \cite{Shlog}. Because the algorithm was programmed in Maple which is not
ideal for handling large expressions it turned out that the following single-valued representation is more efficient.

We use (\ref{inteq}) to extend single-valued hyperlogarithms $\sL_w$ in Example \ref{SVIex} to words $w(\zz)\in\FLT_\FF(\zz)^\ast$. By linearity, GSVHs can be expressed as sums of
such extended single-valued hyperlogarithms (ESVHs).

\begin{defn}\label{GSVIdef}
Let $\beta(\zz)\in\FLT_\FF(\zz)$ and $w(\zz)\in\FLT_\FF(\zz)^\ast$. We inductively define $\sL_{w(\zz)}(z)$ by $\sL_e(z)=1$ for the empty word $e$ and (see (\ref{SVIeq}))
\begin{equation}
\sL_{w(\zz)\beta(\zz)}(z)=\int_{\mathrm{sv}}\frac{\sL_{w(\zz)}(z)}{z-\beta(\zz)}\,\dd z,
\end{equation}
where $\int_{\mathrm{sv}}\dd z$ is defined by the first identity in (\ref{inteq}) with zero (anti-)residues at non-constant points
($\res_{\beta}=\overline{\res}_{\beta}=0$ if $\beta(\zz)\notin\FF$). Integration by parts in the bottom right sector of the commutative hexagon is defined via (\ref{ibp}).
The space of ESVHs is the free $\sO_\Sigma(z,\zz)$ module
$$
\sE_\Sigma=\langle\sL_{w(\zz)}(z),\,w(\zz)\in\FLT_\FF^\ast(\zz)\rangle_{\sO_\Sigma(z,\zz)}.
$$
The $\CC$-algebra of ESVHs with constant coefficients is $\sE_\Sigma^\CC$. Moreover,
$$
\sE_\FF=\bigcup_{\Sigma(\zz)\subset\FLT_\FF(\zz)}\sE_\Sigma\quad\text{and}\quad\sE_\FF^\CC=\bigcup_{\Sigma(\zz)\subset\FLT_\FF(\zz)}\sE_\Sigma^\CC.
$$
\end{defn}
Note that ESVHs are not always single-valued in the sense of Definition \ref{svdef}. They may fail to have single-valued log-Laurent expansions and they may have non-trivial monodromies.

\begin{ex}\label{counterex2}
Consider the function $f(z)=1/(z-2\zz)$ which is $\CC$-analytic in $\CC\backslash\{0\}$ but fails to have a single-valued log-Laurent expansion at $0$, see Example \ref{zm2zzex}.
Using Definition \ref{GSVIdef} we obtain
$$
\sL_{2\zz}(z)=L_{2\zz}(z)+L_0(\zz)\in\sE_{\{2\zz\}}^\CC.
$$
Because $|z/(2\zz)|=1/2<1$ we find that $L_{2\zz}(z)=\log(1-z/(2\zz))$ has trivial monodromy at 0. Hence $\sM_0\sL_{2\zz}(z)=\sL_{2\zz}(z)-2\pi\ii$.
We have $\sL_{2\zz}(z)\notin\sG_{\{2\zz\}}^\CC$.
\end{ex}

\begin{ex}\label{Dzzex2}
For the weight three GSVH in the letters 0, 1, $\zz$ in Example \ref{Dzzex} we have
\begin{equation}\label{esvhex}
\int_{\mathrm{sv}}\frac{\sL_{01}(z)-\sL_{10}(z)}{z-\zz}\,\dd z=\sL_{01\zz}(z)-\sL_{10\zz}(z)\in\sG_{0,1,\zz}^\CC
\end{equation}
with
\begin{eqnarray*}
\sL_{01\zz}(z)&=&L_{01\zz}(z)+L_0(\zz)L_{1\zz}(z)+L_{10}(\zz)L_\zz(z)\\
&&\quad+\,2L_{100}(\zz)+2L_{010}(\zz)+2L_{001}(\zz)-L_{101}(\zz)-L_{110}(\zz)-\zeta(2)L_1(\zz),\\
\sL_{10\zz}(z)&=&L_{10\zz}(z)+L_1(\zz)L_{0\zz}(z)+L_{01}(\zz)L_\zz(z)\\
&&\quad+\,L_{100}(\zz)+2L_{010}(\zz)+2L_{001}(\zz)-L_{110}(\zz)-\zeta(2)L_1(\zz),\\
\end{eqnarray*}
where $\zeta(2)=\pi^2/6$ is the Riemann zeta function at 2. Neither $\sL_{01\zz}(z)$ nor $\sL_{10\zz}(z)$ are GSVHs because each one is singular on the real axis.
Note that many terms cancel in $\sL_{01\zz}(z)-\sL_{10\zz}(z)$.
\end{ex}

\begin{ex}\label{exempty}
For $\Sigma(\zz)\subset\FLT_\FF^\emptyset(\zz)$, see (\ref{FLTempty}), we have $\sE_\Sigma=\sG_\Sigma$. In particular, the dimension $d_n$ of GSVHs in $\Sigma(\zz)$ with weight $n$
is $|\Sigma(\zz)|^n$. This case also contains the class of single-valued hyperlogarithms defined in \cite{BrSVMP,BrSVMPII} where $\Sigma\subset\FF$.
\end{ex}

Note that the weight zero piece of $\sE_\Sigma$ is $\sO_\Sigma(z,\zz)$, see (\ref{OSigmadef}), in contrast to the weight zero piece $\sO_\Sigma^\emptyset(z,\zz)$ of $\sG_\Sigma$
(see Lemma \ref{wt0lem}). In analogy to Theorem \ref{Gthm} (2) we get the exact sequence
$$
0\longrightarrow\sO_\FF(\zz)\longrightarrow\sE_\FF\stackrel{\partial_z}{\longrightarrow}\sE_\FF\longrightarrow0,
$$
where the existence of primitives uses integration by parts as in Definition \ref{intsvGdef}. The kernel of $\partial_z$ in $\sE_\FF$ is given by anti-analytic functions
in the weight zero piece $\sO_\FF(z,\zz)$. This is $\sO_\FF(\zz)$.

In $\sE_\FF$ single-valued integration with respect to $z$ is trivial. We need the commutative hexagon in Figure \ref{fig:gsvh} to find anti-derivatives $\partial_\zz f$
and anti-primitives $\int_{\mathrm{sv}}f\dd\zz$ of functions $f\in\sE_\FF$. The bottleneck of most calculations in $\sE_\FF$ is the evaluation of ESVHs at certain points.
This evaluation, in general, demands the conversion into the $z,\zz$ basis (\ref{GHeq}). Only for constant letters there exists a significant shortcut:
The evaluation $\sL_w(a)$, $a\in\CC$, can be obtained from the multivalued evaluation $L_w(a)$ by the single-valued map which has a simple formula in the $f$-alphabet \cite{BrSV}.

Expressing GSVHs in terms of ESVHs as in (\ref{esvhex}) makes expressions much more compact. In the Maple implementation \cite{Shlog} this advantage of using ESVHs
clearly outweights the drawbacks.

\subsection{The $f$-alphabet}
A third option to present GSVHs (in addition to the $z,\zz$ basis and ESVHs) is the conversion into the $f$-alphabet.
One can use F. Brown's decomposition algorithm (for motivic numbers (periods) defined in \cite{BrownDecom}) to convert GSVHs into the $f$-alphabet \cite{Bcoact2}. This representation
has the simplest structure at the expense of more lengthy expressions. It has not yet been used and all results are experimental \cite{fhlog}.
However, in future implementations of GSVHs the $f$-alphabet should be considered an option.

\section{Construction of some GSVHs}\label{sectconstruction}
By Corollary \ref{partialcor}, the space $\sG_\FF^\CC$ is (inductively) spanned by GSVHs
\begin{equation}\label{GSVHint}
f(z)=\int_{\mathrm{sv}}\frac{h(z)}{z-\beta(\zz)}\,\dd z,
\end{equation}
with $\beta(\zz)\in\FLT_\FF(\zz)$ and $h(z)/(z-\beta(\zz))\in\partial_z\sG_\FF^\CC$. This leads to the following problem.
\begin{prob}
Fix $\beta\in\FLT_\FF$. Construct functions $h\in\sG_\FF^\CC$ such that $h(z)/(z-\beta(\zz))\in\partial_z\sG_\FF^\CC$.
\end{prob}
\begin{remark}\label{constrrk}\mbox{}
\begin{enumerate}
\item If $\beta\in\FLT_\FF^\emptyset$ then every $h\in\sG_\FF^\CC$ gives $h(z)/(z-\beta(\zz))\in\partial_z\sG_\FF^\CC$ (Example \ref{exempty}).
So, we may restrict ourselves to the case $\beta\notin\FLT_\FF^\emptyset$.
\item Graphical functions live in Euclidean space whereas QFT is built on Minkowski metric. For the transition from Euclidean to Minkowski metric one
needs to analytically continue $\zz$ away from the complex conjugate of $z$ such that $z$ and $\zz$ become independent real variables.
In this situation, letters in $\FLT_\FF^\emptyset$ quickly develop singularities (see, e.g., Example \ref{wt1ex}).
The situation is more stable for $\beta\notin\FLT_\FF^\emptyset$ where singularities are only attained on  higher monodromy sheets. Accordingly we observe in
graphical functions letters which are not in $\FLT_\FF^\emptyset$.
\end{enumerate}
\end{remark}

\subsection{Involutions}\label{sectinvols}
A special case arises when $z\mapsto\beta(\zz)$ is an involution.
\begin{defn}
Let $\overline{\beta}$ be $\beta$ with complex conjugated coefficients and
$$
\FLT_\FF^2=\{\beta\in\FLT_\FF, \beta\circ\overline{\beta}=\mathrm{id}\}.
$$
Equivalently, $\FLT_\FF^2$ is the subset of $\FLT_\FF$ for which the map $\tau:z\mapsto\beta(\zz)$ is an involution, $\tau^2=\mathrm{id}$.
\end{defn}

\begin{ex}
The identity $\mathrm{id}\in\FLT_\FF$ gives $\tau:z\mapsto\zz$, see Example \ref{Dzzex} and Section \ref{sect01zz}.

In Examples \ref{wt2ex} and \ref{wt2exzz} we have $z\mapsto1/z\in\FLT_\FF^2$.
\end{ex}

\begin{lem}\label{betalem2}
Every $\beta\in\FLT_\FF^2$ has a representation
$$
\beta(z)=\frac{az+b}{cz+\aaa},\quad\text{with}\quad\overline{b}=-b,\,\overline{c}=-c.
$$
Every such $\beta$ is uniquely represented by the projective point $(\Re a:\Im a:\Im b:\Im c)\in\PP\RR^3$.
\end{lem}
\begin{proof}
For any $\beta(z)=(az+b)/(cz+d)\in\FLT_\FF^2(z)$ an inverse is given by $\beta^{-1}(z)=(dz-b)/(-cz+a)$ which has to equal $\overline{\beta}$ up to a non-zero common factor $\alpha$
in the coefficients. We get
$$
\aaa=d\alpha,\overline{b}=-b\alpha,\overline{c}=-c\alpha,\overline{d}=a\alpha.
$$
By complex conjugation we get $\alpha\overline{\alpha}=1$. If we multiply the coordinates with $\alpha^{-1/2}$ we obtain the desired representation of $\beta$.
The representation is unique up to a common factor $\eta$ with $\eta=\overline{\eta}\in\RR^\times$.
\end{proof}

\begin{prop}\label{invoprop}
Let $h\in\sG_\FF^\CC$ and $\beta\in\FLT_\FF^2$. Then
\begin{equation}\label{invof}
f(z)=\frac{h(z)-h(\beta(\zz))}{z-\beta(\zz)}
\end{equation}
is a symmetric function (under $\tau:z\mapsto\beta(\zz)$) in $\partial_z\sG_\FF^\CC$.
\end{prop}
\begin{proof}
It is clear by Theorem \ref{GHthm} (4), (5) that $f\in\sG\sH_\FF$. The symmetry of $f$ is evident from (\ref{invof}). We need to show that $f\in\sS\sV_\CC$.

Fix $a\in\CC\cup\{\infty\}$. We have $h(z)-h(\beta(\zz))\in\sS\sV_\Sigma$ for some finite set $\Sigma$, see Remark \ref{svrk} (4), (5).
Because $\beta\in\FLT_\FF^2$ the function $h(z)-h(\beta(\zz))\in\sG_\FF$ vanishes at $z=\beta(\zz)$ for any (independent) values of $\zz$.

If $a\neq\beta(\aaa)$ then $1/(z-\beta(\zz))$ is $\CC$-analytic at $a$ and $f$ has a single-valued log-Laurent expansion at $a$.
By Lemmas \ref{gsvhlem1} and \ref{gsvhlem2} (note that $\beta\notin\FF$) this statement extends to $a=\beta(\aaa)$.
Because $\CC\cup\{\infty\}$ is compact, the function $f$ has a finite number of singular points. Hence $f\in\sS\sV_\CC$.
\end{proof}

\begin{ex}
For $h(z)=\sL_0(z)/2$, $\beta(z)=z^{-1}$ we obtain the function $\log(z\zz)/(z-1/\zz)$ in Example \ref{wt2ex}.

The complex conjugate of $h(z)=\sL_{01}(z)$ is $\sL_{10}(z)$. With $\beta=\mathrm{id}$ we get the integrand $(\sL_{01}(z)-\sL_{10}(z))/(z-\zz)$ in Example \ref{Dzzex}.
\end{ex}

\subsection{Exceptional GSVHs}\label{sectexcept}
In the notion of Section \ref{sectGSI}, the function $f(z)$ in (\ref{GSVHint}) is a linear combination of ESVHs which end in the letter $\beta(\zz)$.

If the zero locus of the denominator in the integrand of (\ref{GSVHint}) is empty, $h$ can be any GSVH in $\sG_\FF^\CC$.
If $\beta\in\FLT_\FF^2$ and $h(z)=[h_1(z)-h_1(\beta(\zz))]h_2(z)$ with $h_1,h_2\in\sG_\FF^\CC$, the construction of $f$ is a trivial extension of Proposition \ref{invoprop}.
If none of these cases is true then $f$ is an exceptional GSVH.

\begin{defn}\label{exceptdef}
Assume $f\in\sG_\FF^\CC$ as is (\ref{GSVHint}). If $\beta\notin\FLT_\FF^\emptyset$ and either $\beta\notin\FLT_\FF^2$ or there exist no $h_1,h_2\in\sG_\FF^\CC$
such that $h(z)=[h_1(z)-h_1(\beta(\zz))]h_2(z)$ then $f$ is an exceptional GSVH.
\end{defn}
Note that factorization is readily checked using an algebra basis of ESVHs in Lyndon words.

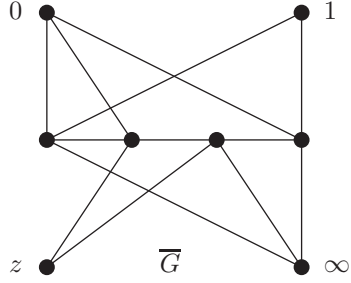
\begin{figure}
\begin{center}
\fcolorbox{white}{white}{
  \begin{picture}(160,121) (67,-11)
    \SetWidth{0.5}
    \SetColor{Black}
    \Line(176,89)(80,41)
    \Vertex(80,41){3}
    \Vertex(112,41){3}
    \Vertex(144,41){3}
    \Vertex(176,41){3}
    \Vertex(80,89){3}
    \Vertex(80,-7){3}
    \Vertex(176,89){3}
    \Vertex(176,-7){3}
    \Line(80,89)(80,41)
    \Line(80,89)(112,41)
    \Line(80,89)(176,41)
    \Line(176,89)(176,41)
    \Line(112,41)(80,-7)
    \Line(144,41)(80,-7)
    \Line(80,41)(176,-7)
    \Line(144,41)(176,-7)
    \Line(176,41)(176,-7)
    \Line(80,41)(176,41)
    \Text(66,87)[lb]{\Black{$0$}}
    \Text(185,87)[lb]{\Black{$1$}}
    \Text(66,-9)[lb]{\Black{$z$}}
    \Text(185,-9)[lb]{\Black{$\infty$}}
    \Text(123,-9)[lb]{\Black{$\overline{G}$}}
  \end{picture}
}
\end{center}
\caption{A graphical function which contains an exceptional GSVH.}
\label{fig:exceptionalGSVH}
\end{figure}

\begin{ex}\label{wt6except}
Consider the ESVH of weight six that emerges from integrating the square of the Bloch-Wigner dilogarithm (\ref{BWD}) twice with respect to $\int_{\mathrm{sv}}\dd z/(z-\zz)$,
\begin{align*}
f_6(z)=&\;3\sL_{0011\zz\zz}-3\sL_{0110\zz\zz}-3\sL_{1001\zz\zz}+3\sL_{1100\zz\zz}+\sL_{01010\zz}-\sL_{01100\zz}-3\sL_{01110\zz}\\
&\quad+\,2\sL_{10010\zz}-2\sL_{10100\zz}-3\sL_{11001\zz}+6\sL_{11100\zz}+6\zeta(3)\sL_{01\zz}.
\end{align*}
Expansions at $0,1,\infty$ to high orders suggest that $f_6\in\sG_{\{0,1,\zz\}}^\CC$, see Section \ref{sect01zz}. In this case $f_6$ is exceptional.
\end{ex}

\begin{ex}\label{exexept}
Consider the four-dimensional graphical function $f_G(z)$ for the graph $G$ whose completion $\overline{G}$ is depicted in Figure \ref{fig:exceptionalGSVH} \cite{gfe,gf}.
An explicit calculation with HyperlogProcedures \cite{Shlog} gives $f_G(z)\in\sO_\Sigma^\emptyset(z,\zz)\partial_z\sG\sH_\Sigma^\CC$ for
$\Sigma(\zz)=\{0,1,\zz,-\zz,1-\zz,\zz/(\zz-1)\,z-\zz,z+\zz\}$. Every graphical function is single-valued \cite{par}.
The function $f_G(z)$ has denominators $z-\zz$ and $z+\zz$. By Proposition \ref{techprop}, single-valued integration of the terms with denominator $z+\zz$ gives a GSVH of weight nine.
It is exceptional. The numerator $h$ in (\ref{GSVHint}) does not factorize. Hence $h_2=1$ in Definition \ref{exceptdef}. An explicit calculation with HyperlogProcedures shows that
$h(z)\neq-h(-\zz)$.
\end{ex}

\begin{quest}\mbox{}
\begin{enumerate}
\item Does there exist a general construction method for (some) exceptional GSVHs?
\item What is the smallest (by weight) exceptional GSVH?
\item Do GSVHs with letters $\notin\FLT_\FF^\emptyset(\zz)\cup\FLT_\FF^2(\zz)$ exist? If yes, what is the smallest (exceptional) example of such a GSVH?
\end{enumerate}
\end{quest}

\subsection{The alphabet 0, 1, $\zz$}\label{sect01zz}
Consider the alphabet $\Sigma(\zz)=\{0,1,\zz\}$. This mathematically interesting case is important in dimensionally regularized pQFT \cite{Duhr,numfunct,7loops}.

Because every GSVH in $\sG_{\{0,1,\zz\}}^\CC$ is a $\CC$-linear combination of ESVHs in $\{0,1,\zz\}$ we obtain finite dimensional vector-spaces at each weight
$$
d_n=\dim \sG_{\{0,1,\zz\}}^\CC(\text{weight}\,=n)\leq 3^n.
$$
The actual dimensions are much smaller than $3^n$. At weights one and two we only have single-valued multiple polylogarithms whereas at weight three
one genuine GSVH exists, see Examples \ref{Dzzex}, \ref{Dzzex2}.

\begin{table}
\begin{center}
\begin{tabular}{l|rrrrrrrrrr}
weight&0&1&2&3&4&5&6&7&8&9\\\hline
dimension $d_n$&1&2&4&9&21&52&134?&358?&986?&2781?\\
\# generators&0&2&1&3&6&16&38?&105?&284?&805?
\end{tabular}
\end{center}
\caption{(Conjectured) dimensions $d_n$ and numbers of generators in $\sG_{\{0,1,\zz\}}^\CC$ for weights $\leq9$.}\label{tab1}
\end{table}

\begin{thm}\label{thmdims}
The dimensions and numbers of generators in $\sG_{\{0,1,\zz\}}^\CC$ for weights $\leq5$ are as given in Table \ref{tab1}.
\end{thm}
\begin{proof}
The numbers of generators follow from the dimensions $d_n$. We get a lower bound for $d_n$ by explicit construction of GSVHs in ${\{0,1,\zz\}}$
using Proposition \ref{invoprop}. Up to weight $n\leq3$ we obtain $2^n$ single-valued multiple polylogarithms plus $\sL_{01\zz}-\sL_{10\zz}$ in Example \ref{Dzzex2}.
By integration with denominators $z$ and $z-1$ the nine GSVHs at weight three give rise to 18 GSVHs of weight four.
The three missing GSVHs of weight four can be constructed with (\ref{invof}). The result is $\sL_{001\zz}-\sL_{100\zz}$, $\sL_{011\zz}-\sL_{110\zz}$, and
$2\sL_{01\zz\zz}-2\sL_{10\zz\zz}-\sL_{010\zz}+\sL_{101\zz}$ (calculations were done with HyperlogProcedures \cite{Shlog}).

At weight 5 we need $52-2\times21=10$ new GSVHs. Nine of them are direct applications of (\ref{invof}),
\begin{center}
\begin{tabular}{ll}
$\sL_{0001\zz}-\sL_{1000\zz}$,\hspace*{4.55cm}&$\sL_{0010\zz}-\sL_{0100\zz}$,\\[3.3pt]
$\sL_{0011\zz}-\sL_{1100\zz}+2\zeta(3)\sL_{1\zz}$,&$\sL_{0101\zz}-\sL_{1010\zz}-4\zeta(3)\sL_{1\zz}$,\\[3.3pt]
$\sL_{0111\zz}-\sL_{1110\zz}-2\zeta(3)\sL_{1\zz}$,&$\sL_{1011\zz}-\sL_{1101\zz}+6\zeta(3)\sL_{1\zz}$,\\[3.3pt]
\multicolumn2l{$2\sL_{001\zz\zz}-2\sL_{100\zz\zz}+\sL_{01\zz0\zz}-\sL_{10\zz0\zz}-\sL_{0010\zz}-\sL_{0100\zz}+\sL_{1001\zz}+\sL_{1010\zz}+2\zeta(3)\sL_{1\zz}$,}\\[3.3pt]
\multicolumn2l{$2\sL_{011\zz\zz}-2\sL_{110\zz\zz}+\sL_{01\zz1\zz}-\sL_{10\zz1\zz}-\sL_{0101\zz}-\sL_{0110\zz}+\sL_{1011\zz}+\sL_{1101\zz}+2\zeta(3)\sL_{1\zz}$,}\\[3.3pt]
\multicolumn2l{$4\sL_{01\zz\zz\zz}-4\sL_{10\zz\zz\zz}-2\sL_{010\zz\zz}+2\sL_{101\zz\zz}-\sL_{01\zz0\zz}-\sL_{01\zz1\zz}+\sL_{10\zz0\zz}+\sL_{10\zz1\zz}+\sL_{0101\zz}-\sL_{1010\zz}$.}
\end{tabular}
\end{center}
The tenth GSVH has the product $(\sL_{01\zz}-\sL_{10\zz})^2/4$ in the numerator (see Definition \ref{exceptdef}) yielding
$$
\sL_{0011\zz}-\sL_{0110\zz}-\sL_{1001\zz}+\sL_{1100\zz}.
$$
Because $\sE_{\{0,1,\zz\}}$ is a free $\sO_{\{0,1,\zz\}}(z,\zz)$ module (see Section \ref{sectGSI}), linear independence is easy to check.

For an upper bound consider the generic linear combination of all $(3^{n+1}-1)/2$ ESVHs in $\{0,1,\zz\}$ up to weight $n$.
We expand each ESVH at $z=0$, $z=1$, and $z=\infty$ in $z$ and $\zz$. In these expansions we consider the ESVHs as generalized hyperlogarithms (Section \ref{sectdefGH})
and expand first in $z$ and then the coefficient in $\zz$. For any word which has at least one letter $\zz$ we obtain increasing pole orders in $\zz$ (with increasing degree in $z$).
In linear combinations which are GSVHs all poles have to cancel (Theorem \ref{Gthm} (6)). This condition gives increasing systems of linear equations whose rank $R$ stabilizes at
$(3^{n+1}-1)/2-\sum_{k=0}^nd_k$. We (inductively) get $d_n\leq (3^{n+1}-1)/2-R-\sum_{k=0}^{n-1}d_k$. The procedure {\tt zzdims} in HyperlogProcedures \cite{Shlog} gives the desired result.
\end{proof}

In practice, the upper bound $(3^{n+1}-1)/2-R-\sum_{k=0}^{n-1}d_k$ for $d_n$ drops rapidly with the size of the linear system (the orders in $z$ and $\zz$) until it stabilizes.
The conjectured dimensions in Table \ref{tab1} are these stabilized upper bounds. At weight six we conjecturally encounter three exceptional GSVHs (see Example \ref{wt6except}).

For $n\geq2$ the dimensions in Table \ref{tab1} are approximated by the recursion
$$
d_n\approx 5d_{n-1}-6d_{n-2}
$$
with surplus elements in dimensions three and five and a missing element in dimension nine.
An ansatz for a generating function is
$$
\sum_{n=0}^\infty d_nx^n\stackrel{?}{=}\,\frac{1}{1-2x}+\frac{x^3+x^5-x^9+\ldots}{(1-2x)(1-3x)},
$$
where the first term counts the number of single-valued multiple polylogarithms in the letters 0 and 1.

The number contents of GSVHs in $\{0,1,\zz\}$ is given by multiple zeta values (MZVs).
\begin{defn}
The $\QQ$-algebra of MZVs is the $\QQ$-span of MZV sums
\begin{equation}\label{MZVdef}
\zeta(n_1,n_2,\ldots,n_r)=\sum_{1\leq k_1<k_2<\ldots<k_r}n_1^{-k_1}n_2^{-k_2}\cdots n_r^{-k_r},\quad n_i=\{1,2,3,\ldots\},\,n_r\geq2.
\end{equation}
The weight of the tuple $(n_1,n_2,\ldots,n_r)$ is $n_1+n_2+\ldots+n_r$. For any finite $\Sigma(\zz)\subset\FLT_\FF(\zz)$ we define
$\sG_\Sigma^{\mathrm{MZV}}\subset\sG_\Sigma^\CC$ as the $\QQ$-algebra of GSVHs in $\Sigma(\zz)$ with MZV coefficients. Likewise, for any finite set $\Sigma_0\subset\CC$,
\begin{align*}
\sO_\Sigma^{\mathrm{MZV}}(z,\zz)&=\mathrm{MZV}[z,((z-\beta(\zz))^{-1})_{\beta(\zz)\in\Sigma(\zz)},\zz,((\zz-\bb)^{-1})_{b\in\Sigma(\zz)\cap\FF}],\\
\sO_{\Sigma_0}^{\mathrm{MZV}}(z)&=\mathrm{MZV}[z,((z-a)^{-1})_{a\in\Sigma_0}]
\end{align*}
are the subrings of $\sO_\Sigma(z,\zz)$ and $\sO_{\Sigma_0}(z)$ which are defined over MZVs (respectively).
\end{defn}

\begin{lem}\label{lemf1}
For any $f\in\sG_{\{0,1,\zz\}}^{\mathrm{MZV}}$ the regularized limit $f(1)$ is an MZV.
\end{lem}
\begin{proof}
We use the representation of $f$ in the $z,\zz$ basis (\ref{GHeq}),
\begin{equation}\label{eqfMZV}
f(z)=\sum_{v,w(\zz)}c_{v,w(\zz)}L_v(\zz)L_{w(\zz)}(z),
\end{equation}
where $v\in\{0,1\}^\ast$, $w(\zz)\in\{0,1,\zz\}^\ast$, and $c_{v,w(\zz)}$ is an MZV. We consider $z$ and $\zz$ as independent variables.
Because $f$ is a GSVH with a single-valued log-Laurent expansion at 1 we may determine $f(1)$ by first taking the limit $\zz\to1$ followed by the limit $z\to1$.
In $L_v(\zz)$ we obtain the regularized limit $L_v(1)$ which is an MZV. In $L_{w(\zz)}(z)$ the value $\zz=1$ differs from the lower bound $0$ and the upper bound $z$ of the iterated integral.
The limit $\zz\to1$ is regular yielding $L_{w(1)}(z)$. Because $w(1)\in\{0,1\}^\ast$ the regularized limit $z\to1$ gives the MZV $L_{w(1)}(1)$.
Therefore every summand in $f(1)$ is the product of three MZVs.
\end{proof}

\begin{thm}\label{thm01zz}
Theorem \ref{Gthm} restricts to $\sG_{\{0,1,\zz\}}(\mathrm{MZV})=\sO_{\{0,1,\zz\}}^{\mathrm{MZV}}(z,\zz)\sG_{\{0,1,\zz\}}^{\mathrm{MZV}}$, where statement (5)
restricts to the M\"obius transformations $z\to\{z,1-z,1/z,1/(1-z),z/(z-1),(z-1)/z\}$ that stabilize the singular points $0,1,\infty$. In particular, the sequence
\begin{equation}\label{eq01zz1}
0\longrightarrow\sO_{\{0,1\}}^{\mathrm{MZV}}(\zz)\longrightarrow\sG_{\{0,1,\zz\}}(\mathrm{MZV})\stackrel{\partial_z}{\longrightarrow}\sG_{\{0,1,\zz\}}(\mathrm{MZV})\longrightarrow0
\end{equation}
is exact. We have
\begin{equation}\label{eq01zz2}
\sG_{\{0,1,\zz\}}^{\mathrm{MZV}}\otimes_\QQ\CC\cong\sG_{\{0,1,\zz\}}^\CC.
\end{equation}
\end{thm}
\begin{proof}
For (\ref{eq01zz1}) it suffices to show that every $f\in\sG_{\{0,1,\zz\}}(\mathrm{MZV})$ has a primitive. The alphabet $0,1,\zz$ is stable under partial fraction decomposition because for $a=0,1$,
$$
\frac{1}{z-a}\frac{1}{z-\zz}=\frac{1}{\zz-a}\Big(\frac{1}{z-\zz}-\frac{1}{z-a}\Big),
$$
and $1/(\zz-a)\in\sO_{0,1,\zz}(z,\zz)$. Moreover, the projections $\pi_{\partial_z}$ and $\pi_{\partial_\zz}$ filter (anti-)resides at 0 and 1.
The (anti-)residue at 0 trivially maps into MZVs. The (anti-)residue at 1 maps into MZVs by Lemma \ref{lemf1}.

The previous arguments are symmetric under complex conjugation so that we can also use the commutative hexagon in Figure 1 to take anti-primitives in $\sG_{\{0,1,\zz\}}(\mathrm{MZV})$,
see (\ref{inteq}).

We use induction over the weight of $f$ to show that $\sG_{\{0,1,\zz\}}(\mathrm{MZV})$ is invariant under complex conjugation.
The weight zero piece $\sO^{\emptyset,\mathrm{MZV}}_{\{0,1,\zz\}}(z,\zz)$ of $\sG_{\{0,1,\zz\}}(\mathrm{MZV})$ is $\sO_{\{0,1\}}^{\mathrm{MZV}}(z,\zz)$ which is invariant under complex conjugation.
Assume $f\in\sG_{\{0,1,\zz\}}(\mathrm{MZV})$ has weight $n\geq1$. Because $\sO_{\{0,1,\zz\}}^{\mathrm{MZV}}(z,\zz)$ is invariant under complex conjugation we may assume without restriction that
$f\in\sG_{\{0,1,\zz\}}^{\mathrm{MZV}}$ with $f(0)=0$. Then $g=\partial_zf$ has weight $n-1$. By induction $g$ has a complex conjugate $\overline{g}$.
The anti-primitive $\overline{f}=\int_{\mathrm{sv}}\overline{g}\,\dd\zz\in\sG_{\{0,1,\zz\}}(\mathrm{MZV})$ possibly differs from the complex conjugate of $f$ by a constant.
Because $f(0)=\overline{f}(0)=0$ the constant is zero.

For (5), consider the M\"obius transformation $\phi(z)=z/(z-1)$. It is an endomorphism of $\sO_{\{0,1\}}^{\mathrm{MZV}}(z,\zz)$ and because it maps $1/(z-\zz)$ to $-(z-1)(\zz-1)/(z-\zz)$ it is also
an endomorphism of $\sO_{\{0,1,\zz\}}^{\mathrm{MZV}}(z,\zz)$. As in the case of complex conjugation we show by induction over the weight that $\sG_{\{0,1,\zz\}}(\mathrm{MZV})$
is invariant under $\phi$. The integration constant is trivial because $\phi(0)=0$.

Consider the M\"obius transformation $\tau(z)=1-z$ which is also an endomorphism of $\sO_{\{0,1,\zz\}}^{\mathrm{MZV}}(z,\zz)$.
If we use induction over the weight of $f\in\sG_{\{0,1,\zz\}}^{\mathrm{MZV}}$ with $f(0)=0$ we run into the difficulty that $f(\tau(0))=f(1)$. By Lemma \ref{lemf1} the value $f(1)$ is an
MZV and we obtain that $f(\tau(z))\in\sG_{\{0,1,\zz\}}^{\mathrm{MZV}}$.

Statement (5) is true because the set of M\"obius transformations that stabilize $0,1,\infty$ is generated by $\phi$ and $\tau$.

To prove (\ref{eq01zz2}) we represent $f\in\sG_{\{0,1,\zz\}}^\CC$ as a $\CC$-linear combination of ESVHs in $\{0,1,\zz\}$. By induction over the weight
(using the commutative hexagon, see Section \ref{sectGSI}) we see that ESVHs in $0,1,\zz$ are defined over MZVs (as in (\ref{eqfMZV})).
The space $\sG_{\{0,1,\zz\}}^\CC$ is given by those linear combinations of ESVHs which have single-valued log-Laurent expansions at $0$, $1$, and $\infty$, see the proof of Theorem \ref{thmdims}.
These expansions provide linear equations for the ESVHs in $\{0,1,\zz\}$ of weight $\leq n$. At 0 these equations have MZV-coefficients, see (\ref{zeta}).
By statement (5) the expansions at 1 and at $\infty$ can be mapped to expansions at 0 by M\"obius transformations which stabilize $\sG_{\{0,1,\zz\}}^{\mathrm{MZV}}$.
So, these expansion also give rise to linear equations with MZV-coefficients. For weight $\leq n$ the space $\sG_{\{0,1,\zz\}}^\CC$ is obtained from the $(3^{n+1}-1)/2$
ESVHs in $\{0,1,\zz\}$ of weight $\leq n$ by a finite number of these linear equations. Because the total weight (the weight of the MZV plus the weight of the ESVH) is a filtration,
the coefficients of ESVHs with weight $n$ are rational (MZVs of weight 0). The solution of the system provides $\sG_{\{0,1,\zz\}}^\CC$ at weight $n$ as MZV-span of ESVHs.
Therefore $\sG_{\{0,1,\zz\}}^{\mathrm{MZV}}\otimes_\QQ\CC$ surjects onto $\sG_{\{0,1,\zz\}}^\CC$. The kernel is zero because the map is an embedding.
\end{proof}
A detailed understanding of GSVHs in $0,1,\zz$ could lead to much faster calculations in this alphabet. However, for most computations at modest weights in pQFT the bottleneck is not handling these GSVHs.

\bibliographystyle{plain}
\renewcommand\refname{References}

\end{document}